\documentclass[sigconf]{aamas} 

\usepackage{balance} 

\usepackage{enumitem}

\usepackage{graphicx}
\usepackage{subfigure}

\usepackage{amsmath}
\usepackage{mathtools}
\usepackage{amsthm}

\usepackage{algorithm}
\usepackage{algorithmic}

\usepackage{wrapfig}
\usepackage{lipsum}
\usepackage{appendix}
\usepackage{titletoc}

\newcommand\DoToC{%
  \startcontents
  \vspace{0.1in}
  \printcontents{0}{0}{\textbf{Contents}\vskip10pt\hrule\vskip5pt}
  \vskip5pt\hrule\vskip5pt
}

\theoremstyle{plain}
\newtheorem{theorem}{Theorem}[section]
\newtheorem{proposition}[theorem]{Proposition}
\newtheorem{lemma}[theorem]{Lemma}

\theoremstyle{definition}

\newtheorem{assumption}[theorem]{Assumption}

\theoremstyle{remark}

\newcommand{\norm}[1]{\left\lVert#1\right\rVert}

\numberwithin{equation}{section}

\DeclareMathOperator*{\argmin}{arg\,min\ }

\newcommand{\RNum}[1]{\uppercase\expandafter{\romannumeral #1\relax}}

\setcopyright{ifaamas}
\acmConference[AAMAS '23]{Proc.\@ of the 22nd International Conference
on Autonomous Agents and Multiagent Systems (AAMAS 2023)}{May 29 -- June 2, 2023}
{London, United Kingdom}{A.~Ricci, W.~Yeoh, N.~Agmon, B.~An (eds.)}
\copyrightyear{2023}
\acmYear{2023}
\acmDOI{}
\acmPrice{}
\acmISBN{}

\author{Jing Wang$^{*,\dagger}$}
\affiliation{
  \institution{New York University}
  \city{New York}
  \country{United States}}
\email{jw5665@nyu.edu}

\author{Meichen Song$^{*,\dagger}$}
\affiliation{
  \institution{Stony Brook University}
  \city{New York}
  \country{United States}}
\email{meichen.song@stonybrook.edu}

\author{Feng Gao$^{\ddag}$}
\affiliation{
  \institution{IIIS, Tsinghua University}
  \city{Beijing}
  \country{China}}
\email{feng.gao220@gmail.com}

\author{Boyi Liu}
\affiliation{
  \institution{Northwestern University}
  \city{Illinois}
  \country{United States}}
\email{boyiliu2018@u.northwestern.edu}

\author{Zhaoran Wang}
\affiliation{
  \institution{Northwestern University}
  \city{Illinois}
  \country{United States}}
\email{zhaoran.wang@northwestern.edu}

\author{Yi Wu}
\affiliation{
  \institution{IIIS, Tsinghua University;\\
  Shanghai Qi Zhi Institute}
  \city{}
  \country{}
  }
\email{jxwuyi@gmail.com}

\title[AAMAS-2023 Formatting Instructions]{Differentiable Arbitrating in Zero-sum Markov Games}

\begin{abstract}
We initiate the study of how to perturb the reward in a zero-sum Markov game with two players to induce a desirable Nash equilibrium, namely arbitrating. Such a problem admits a bi-level optimization formulation. The lower level requires solving the Nash equilibrium under a given reward function, which makes the overall problem challenging to optimize in an end-to-end way. We propose a backpropagation scheme that differentiates through the Nash equilibrium, which provides the gradient feedback for the upper level. In particular, our method only requires a black-box solver for the (regularized) Nash equilibrium (NE). We develop the convergence analysis for the proposed framework with proper black-box NE solvers and demonstrate the empirical successes in two multi-agent reinforcement learning (MARL) environments.
\end{abstract}


\newcommand{\BibTeX}{\rm B\kern-.05em{\sc i\kern-.025em b}\kern-.08em\TeX}

\begin{document}


\pagestyle{fancy}
\fancyhead{}



\maketitle 
\def\thefootnote{*}\footnotetext{These authors contributed equally to this work}
\def\thefootnote{$\dagger$}\footnotetext{Work finished during the internship in Shanghai Qi Zhi Institute.}
\def\thefootnote{$\ddag$}\footnotetext{Feng Gao contributed the entire experiment part of this work.}

\section{Introduction}
Arbitrating the conflict between self-interest and collective interests permeates the development of human societies~\cite{harari2014sapiens}.
We study Markov games~\cite{littman1994markov} as abstractions of human societies. Nash equilibrium (NE)~\cite{bacsar1998dynamic}, where none of the players could benefit from unilaterally deviating its strategy, is an essential concept in Markov games.
While players modeled by Markov games act rationally to maximize their own rewards, the lack of collective consideration may lead to an undesirable NE, which undermines the overall welfare from a system perspective~\cite{10.2307/3690047}.

Incentive design~\cite{ratliff2019perspective}, which was originally developed for bandit settings, aims to arbitrate such a conflict by perturbing the rewards to refine NE so that the self-interested players reach a desirable one, e.g., the higher social welfare or the lager exploration rate. More specifically, the incentive design problem is naturally formulated as a bi-level structure~\cite{luo1996mathematical, colson2007overview}. The upper-level designer aims to determine the optimal incentives and the lower-level players perform under incentive-perturbed rewards. To this end, it needs to anticipate how the lower-level players react to the prescribed incentives when they play a Markov game.
However, taking a derivative through the bi-level structure is difficult~\cite{friesz1990sensitivity}, especially when the lower level has multiple agents. In particular, the lower-level equilibrium appears to be a black box for the upper-level designer, which makes it challenging for the upper-level designer to assess how the prescribed incentives will influence the lower-level equilibrium. As a result, the upper-level designer lacks the gradient feedback for effectively updating its policy for prescribing the incentives~\cite{li2020end}.

Some attempts have been proposed to work around the difficulty of derivatives in solving bi-level optimization.
One direction is to get rid of the bi-level problem structure by converting the problem into a multi-objective problem: i.e., the low-level players simultaneously optimize the game reward and the designer’s objective by exploring multiple NEs~\cite{du2019liir,rahmattalabi2016d++,ibrahimreward,10.1145/860575.860689,li2021celebrating,tang2021discovering}. The simplified problem can be directly solved by an NE solver. However, there is no theoretical guarantee for finding the optimal NE subject to the designer's objective~\cite{devlin2011theoretical}. 
The other direction is to keep the bi-level problem structure and avoid the derivative issue by applying a gradient-free optimizer to the upper level~\cite{mguni2019coordinating}. However, this suffers from a high computation cost since a zeroth-order method typically requires a large number of queries to the lower-level solver to derive the desired NE.
Therefore, a first-order method that can be applied to complicated environments under the bi-level structure with a convergence guarantee is urgently needed to overcome the sample efficiency issue. 

We extend the incentive design problem to the MARL settings and construct a provably \textit{differentiable arbitrating} (DA) framework. We derive the derivatives through the NE for the DA framework, which enables the upper level to utilize the policy gradient feedback from the lower level to obtain the gradient of the designer’s objective. 
Theoretical convergence proof for DA is presented in the setting of the two-player zero-sum game 
by choosing proper NE solvers which have convergence guarantee, e.g. Policy Extragradient Method (PEM)~\cite{cen2021fast} and Entropy-regularized OMWU~\cite{cen2022faster}.
In practice, we implement the DA framework using a multi-agent variant of the soft actor-critic algorithm (MASAC)~\cite{haarnoja2018soft}, leading to a practical algorithm, DASAC. 
We evaluate our algorithm in two two-player zero-sum MARL environments. Empirical results show that the incentive proposed by the designer arbitrates the conflict well: i.e., the upper-level designer's loss function reduces while lower-level self-interested players still attain an NE. The emergent arbitrated behaviors are consistent with human intuitions in each setting. Besides, the sample efficiency for the upper level is significantly improved compared with a zeroth-order method.
\vspace{-0.15in}\\

\textbf{Contributions.} \textbf{i)} We consider the problem of incentive design in MARL settings and tackle the challenge of deriving the gradient of the upper-level objective through the bi-level structure. \textbf{ii)} We develop the \emph{first} differentiable \textit{first-order} framework, \textit{Differentiable Arbitrating} (DA), which can arbitrate the conflict to obtain a desirable NE in Markov games. \textbf{iii)} We theoretically prove the convergence of DA  with proper NE solvers for two-player zero-sum Markov games. \textbf{iv)} We empirically show that a desirable NE with interpretable behaviors can be found by DA more efficiently than zeroth-order methods in two MARL environments. 
\section{Related works}
\textbf{Nash equilibrium for Markov Games.}
There is a large amount of work on finding the Nash equilibrium in multi-agent Markov games. Claus and Boutilier propose \textit{fictitious play}~\cite{claus1998dynamics} (FS), i.e., all players presuming that the opponents are stationary, for equilibrium selection  with guaranteed convergence to a Nash equilibrium. FS was originally developed for normal form games and therefore was not widely applied to complex applications~\cite{lambert2005fictitious,mcmahan2007fast,heinrich2015smooth} until the \textit{fictitious self-play} (FSP)~\cite{pmlr-v37-heinrich15} was proposed.~\cite{heinrich2016deep} extends SP from extensive-form games to imperfect-information games.~\cite{liu2020policy} proposes the smooth FSP algorithm, which expands the PPO algorithm~\cite{schulman2017proximal} from the single-agent to the multi-agent zero-sum Markov game and provides with convergence guarantee. In the multi-agent reinforcement learning setting, MADDPG in~\cite{lowe2017multi} and MAAC in~\cite{iqbal2019actor} adapt the actor-critic~\cite{konda2000actor} and soft-actor-critic~\cite{haarnoja2018soft} algorithms from the single-agent to complex multi-agent setting and empirically show better performances fully than decentralized multi-agent RL on finding NEs. MADDPG and MAAC do not have theoretical convergence guarantee even though they achieve empirical successes. On the other hand, many works have developed convergence guarantee for value-based methods~\cite{chen2021sample,cen2021fast,cen2022faster} and policy-based methods~\cite{zhao2022provably,alacaoglu2022natural} for two-player zero-sum Markov games. Cen~\cite{cen2021fast,cen2022faster} introduce the entropy regularization term into the two-player zero-sum Markov game and propose the Policy Extragradient Methods and Entropy-regularized OMWU Methods, which have convergence guarantee without the NE uniqueness assumption on the two-player zero-sum Markov game.\\
\textbf{Equilibrium Selection/Refinement.} Since even team Markov Games (where the players have common interests) could have multiple NEs, there exists a large volume of work in game theory on NE selection~\cite{selten1965spieltheoretische,myerson1978refinements}, e.g. admissibility~\cite{banks1987equilibrium}, subgame
perfection~\cite{selten1965spieltheoretische}, Pareto efficiency~\cite{bernheim1987coalition} or stability against opponent’s deviation from best response~\cite{fang2013protecting}. \cite{ratliff2019perspective} raises emphasis on \textit{adaptive incentive design}, which modifies intrinsic rewards via additional incentive to balance the individuals' self-interests and system's social welfare for desirable NE selection. The incentive design was first proposed in economics~\cite{pigoue1920welfare} and attracts substantial attentions with many follow-up works~\cite{laffont2009theory,bolton2004contract,weber2011optimal}. During the past decades, people explored the economic incentive to various domains, including energy~\cite{cambini2010incentive,mohsenian2010optimal}, transportation~\cite{melnikow1997effect,kokkinogenis2014policy}, healthcare~\cite{ma2019realigning,aziz2021efficient}, and  education~\cite{macartney2016incentive}. More recently, incentive design has also been extended to the area of  game theory and reinforcement learning.~\cite{liu2021inducing} utilize the incentive design in the multi-bandit problem and prove that the proposed algorithm converges to the global optimum at a sub-linear rate for a broad class of games.~\cite{mguni2019coordinating} provide an incentive-design mechanism for an uncooperative multi-agent system and optimize the upper-level incentive objective with Bayesian optimization, a sample-efficient optimization algorithm, instead of the gradient-based methods because the lower-level MARL problem is a black box.~\cite{yang2020learning} propose a decentralized incentive mechanism that allows each individual to directly give rewards to others and learn its own incentive function, respectively. Also,~\cite{chen2022adaptive} proposes a bi-level incentive mechanism for a single lower-level player, which can be viewed as a special case of our setting. 
Stackelberg game~\cite{leitmann1978generalized,li2017review} is a strategic game in economics in which the leader moves first and the followers behave sequentially. The incentive design could be reformulated as the Stackelberg game by treating the incentive design objective as the leading player.~\cite{wang2022coordinating} proposed a gradient-decent algorithm to find NE for the Stackelberg bandit problem.~\cite{zhong2021can} propose a value iteration method for solving the Stackelberg Markov game with convergence guarantee.~\cite{wang2021learning} features differentiation through policy gradient on Sequential Decision Problems without global convergence guarantee.\\
\textbf{Bi-level optimization.} In machine learning, a large amount of tasks can be formulated as bi-level optimization problems, e.g. adversarial learning~\cite{jiang2021learning}, meta-learning~\cite{finn2017model}, hyperparameter optimization~\cite{baydin2018online}, and end-to-end learning~\cite{amos2017optnet}. Many gradient-based algorithms are used for solving these problems.
Most of them are first-order gradient-based methods~\cite{finn2017model,nichol2018first,goodfellow2020generative,baydin2018online} thanks to both computational efficiency of first-order gradients and fast empirical convergence. 
By contrast, zeroth-order methods~\cite{song2020maml,wang2021zarts} suffer from a heavy workload of data sampling, which requires tremendous computing resources. 
\cite{parker2020ridge} develop a second-order method.
However, the proposed algorithm requires computation of the inverse of Hessian, which restricts the algorithm to simple tabular settings. 
The most challenging part for gradient-based methods to solve the bi-level optimization problem is to derive the derivative of lower-level solution with respect to the incentive variables in the higher level. 
Feng~\cite{feng2021neural} and Yang~\cite{yang2022adaptive} expand the lower-level optimization process as a sequence of gradient updates and directly compute gradients throughout the entire optimization process. Despite of the simplicity, such a method suffers from tremendous computation troubles due to memory  explosion and numerical issues. Therefore, approximation techniques such as truncated gradient with a sliding window would be necessary. By contrast, our DA framework gives an accurate first-order formula for exact gradient computation for the lower-level optima. Wang~\cite{wang2021learning} constructed a bi-level problem to infer the missing parameters in MDP by learning a predictive model. This paper is under the single-agent setting where the objective in the lower-level is to minimize a clear performance measure while our setting is more complex by requiring the computation of NE.
\section{Problem formulation}
\label{section:problem formulation}
This section organizes as follows. Section \ref{subsec:formu of two-player zero-sum Markov game} begins with a typical definition for two-player zero-sum Markov games $\mathcal{G}$. To illustrate how the designer induces the behaviors of players, Section \ref{subsec:irg} introduces the incentive parameter $\theta$ into the reward of the two-player zero-sum Markov game and formulate the incentivized Markov game $\mathcal{G}_\theta$. 
Section \ref{subsec:irg2} designs an entropy-regularized Markov game $\mathcal{G}_\theta'$ by adding entropy regularization terms into the reward of $\mathcal{G}_\theta$ and demonstrates $\mathcal{G}_\theta'$ owns the unique NE. Section \ref{subsec:bi-level} formulates the arbitrating system as a bi-level optimization scheme based on the regularized Markov game with the  incentive-perturbed reward.

\vspace{-0.1in}
\subsection{Two-player Zero-sum Markov Game}
\label{subsec:formu of two-player zero-sum Markov game}
Consider a two-player zero-sum Markov game $\mathcal{G}\!=\!(\mathcal{S},\{\mathcal{A}^i\}_{i\in\{1,2\}},$\\$\mathcal{P},\{r^i\}_{i\in\{1,2\}},\gamma)$, where $\mathcal{S}$ is the state space observed by all players, $\mathcal{A}^i$ is the action space of player $i$ and $\mathcal{A}:=\mathcal{A}^1\times\mathcal{A}^2$ is the joint action space of two players, then $\mathcal{P}:\mathcal{S}\times\mathcal{A}\times\mathcal{S}\rightarrow[0,1]$ denotes the transition probability from state $s\in\mathcal{S}$ to state $s'\in\mathcal{S}$ for taking joint action $a\in\mathcal{A}$. $r^i:\mathcal{S}\times\mathcal{A}\rightarrow \left[-1,1\right]$ is the immediate reward function of player $i$, which implies $r^1=-r^2$ in zero-sum game. $\gamma\in[0,1)$ is the discounted factor. 
]We remark that we consider the two-player zero-sum setting for notation simplicity. Extensions to more general settings will be discussed in Appendix \ref{supp:extension}.
\vspace{-0.05in}
\subsection{Markov Game with the Incentive-perturbed Reward}\label{subsec:irg}
Let $\mathcal{G}_\theta=(\mathcal{S},\{\mathcal{A}^i\}_{i\in\{1,2\}},\mathcal{P},\{r^i(\cdot;\theta)\}_{i\in\{1,2\}},\gamma)$ be a two-player zero-sum Markov game with an  incentive-perturbed reward function $r^i(\cdot;{\theta})$ explicitly parameterized by $\theta$, where $\theta\in\mathbb{R}^m$ represents the incentives added by the designer. 
We assume that $r^i(\cdot;{\theta}):\mathcal{S}\times\mathcal{A}\rightarrow \left[-1,1\right]$ is uniformly bound in $\theta$ and remains zero-sum.
For a given trajectory $\tau=(s_0,\{a_0^i\}_{i\in\{1,2\}},s_1,...,\{a_{T-1}^i\}_{i\in\{1,2\}},s_T)$, we denote the total discounted reward for player $i \in \{1,2\}$ as
\vspace{-0.05in}
\begin{align*}
    R^i(\tau;{\theta})=\sum_{t=0}^{T-1}\gamma^tr^i(s_t,a_t;{\theta}).
\end{align*}
\vspace{-0.1in}\\
Define $\pi^i\!\!:\!\mathcal{S}\!\times\!\mathcal{A}^i\!\!\rightarrow\!\![0,1]$ as the probability distribution of policy of player $i$ and $\pi\!\!=\!\!(\pi^1,\pi^2)$ as that of the joint policy. 

Let $D^{\pi}(\tau)$ denote the probability distribution of the trajectory $\tau$ based on the policy pair $\pi$.
That is,
\vspace{-0.05in}
\begin{align*}
  D^{\pi}(\tau):=\rho_0(s_0)\prod_{t=0}^{T-1}P(s_{t+1}|s_t,a_t) \left(\prod_{i=1}^2\pi^i(a^i_t|s_t)\right),
\end{align*}\\
\vspace{-0.25in}

where $\rho_0(s_0)$ is the distribution for the initial state.
We denote $\mathbb{E}_{\pi}[\;\cdot\;]$ as the expectation over the trajectory $\tau\sim D^{\pi}(\tau)$. Then the performance of the policy pair $\pi$ for the game $\mathcal{G}_{\theta}$ is evaluated by the state-value function $V^i_{\pi}:\mathcal{S}\rightarrow\mathbb{R}$, which is defined as
\vspace{-0.05in}
\begin{align}\label{eq:value}
    V^i_{\pi}(s;\theta)&=\mathbb{E}_{\pi}\left[\sum_{t=0}^{T-1}\gamma^t{\cdot}r^i(s_t,a^i_t,a^{-i}_t;{\theta})|s_0=s\right]\notag\\
    &=\mathbb{E}_{\tau\sim D^{\pi}}\left[R^i(\tau;{\theta})\big|s_0=s\right].
\end{align}
\vspace{-0.3in}\\


The Nash equilibrium (NE) is an essential concept where no agent could benefit by individually changing its policy. The NE for the incentivized two-player zero-sum Markov game $\mathcal{G}_{\theta}$ is the solution for the following min-max game:
\begin{align}\label{NE_ori}
    \min_{\pi^2}\max_{\pi^1}V^1_{\pi}(s;\theta),\quad s\in\mathcal{S}.
\end{align}
Shapley~\cite{shapley1953stochastic} proves that there exists a NE pair for the min-max game \ref{NE_ori} for all state $s\in\mathcal{S}$ and the min-max value is unique~\cite{von2007theory}. 

\subsection{Entropy-regularized Markov Game}\label{subsec:irg2}
To ensure the NE to be unique, we follow~\cite{pmlr-v80-dai18c,cen2021fast,pmlr-v97-zhao19d} and introduce the entropy-regularized counterpart $\mathcal{G}^{'}_\theta\!=\!(\mathcal{S}, \{\mathcal{A}^i\}_{i\in\{1,2\}}, \mathcal{P},\{r^{(i)}\\(\cdot;\theta)\}_{i\in\{1,2\}}, \gamma,$ $\lambda)$, where
$\lambda\geq0$ is the regularization parameter. 
The entropy regularization technique is commonly utilized in reinforcement learning to encourage exploration to avoid being trapped at sub-optimal solutions~\cite{haarnoja2017reinforcement,lee2018sparse,ahmed2019understanding} and keep the policies of different agents away from being heavily affected by the opponents’ strategy~\cite{liu2019policy,xiao2021entropy}.

Specifically, the reward function for player $i$ is replaced by its entropy-regularized reward function $r^{(i)}_{\pi}(\cdot;\theta):\mathcal{S}\!\times\!\mathcal{A}^i\!\times\!\mathcal{A}^{-i}\!\rightarrow\!\mathbb{R}$, which is defined as,
\vspace{-0.05in}
\begin{align}\label{def:regularized reward}
    r^{(i)}_{\pi}(s,a^i,a^{-i};\theta)=&r^i(s,a^i,a^{-i};\theta)-\lambda\log(\pi^i(a^i|s))\notag\\
    &\quad+\lambda\log(\pi^{-i}(a^{-i}|s)).
\end{align}
\vspace{-0.18in}\\
Here we clarify what the entropy we use with a light abuse of notation. The state-reward function and the entropy-regularized state-reward function associated with policy pair $\pi$ are defined as
\begin{align}
    r^i_{\pi}(s;\theta)&\!=\!\mathbb{E}_{(a^i\!,a^{-i})\sim\pi}\!\!\left[r^i(s,a^i\!,a^{-i};\theta)\right],\notag\\
    r^{(i)}_{\pi}(s;\theta)&\!=\!\mathbb{E}_{(a^i\!,a^{-i})\sim\pi}\!\!\left[r^{(i)}_{\pi}(s,a^i\!,a^{-i};\theta)\right]\notag\\
    &\!=\!r^i_{\pi}(s;\theta)\!+\!\lambda H(\pi^i(\cdot|s))-\lambda H(\pi^{-i}(\cdot|s)),
    \label{regularized reward}
\end{align}
where $H(\pi^i(\cdot|s))\!\!=\!-\!\sum_{a^i\in\mathcal{A}^i}\!\pi^i(a^i|s)\log(\pi^i(a^i|s))$ is the Shannon entropy. Correspondingly, the total discounted reward $R^{(i)}_{\pi}$ for a given trajectory $\tau$ of the player $i$ in game $\mathcal{G}_{\theta}'$ and the entropy-regularized state-value function with incentive parameter $V^{(i)}_{\pi}:\mathcal{S}\rightarrow\mathbb{R}$ are defined as
\begin{align}
    R^{(i)}_{\pi}(\tau;{\theta})&=\sum_{t=0}^{T-1}\gamma^tr^{(i)}_{\pi}(s_t,a_t^i,a_t^{-i};{\theta}),\notag\\
    V^{(i)}_{\pi}(s;\theta)&=\mathbb{E}_{\pi}\left[\sum_{t=0}^{T-1}\gamma^t{\cdot}r^{(i)}_{\pi}(s_t,a^i_t,a^{-i}_t;{\theta})|s_0=s\right]\notag\\
    &=\mathbb{E}_{\tau\sim D^{\pi}}\left[R^{(i)}_{\pi}(\tau;{\theta})\big|s_0=s\right].\label{regularized V}
\end{align}\label{eq:value2}
\vspace{-0.15in}

Suppose the NE of the regularized Markov game with  incentive-perturbed reward $\mathcal{G}_{\theta}'$ is $\pi_{*}(\theta)=(\pi_{*}^1(\theta),\pi^2_{*}(\theta))$. We write Player -$i$ as Player $i$’s opponent then the joint policy could also be rewritten as the policy pair $\pi=(\pi^i,\pi^{-i})$, where $i \in \left\{1, 2\right\}$. 
Then the NE can be defined as follows, 
for $\forall i\in\left\{1,2\right\}$, the NE $\pi^*$ satisfies
\begin{align}
\label{equ:def of regularized NE}
    V^{(i)}_{\pi_{*}(\theta)}(s;\theta)=\max_{\pi^i}V^{(i)}_{(\pi^i,\pi_*^{\text{-}i}(\theta))}(s;\theta),\quad \forall s\in\mathcal{S}.
\end{align}
The NE for the regularized Markov game $\mathcal{G}_\theta'$ is equivalent to the solution of the following min-max optimization problem
\vspace{-0.05in}
\begin{align}\label{prob:minmax2}
    \min_{\pi^2}\max_{\pi^1}V^{(1)}_{\pi}(s;\theta),\quad\forall s\in\mathcal{S}.
\end{align}
As is known from~\cite{cen2021fast,cen2022faster,mertikopoulos2016learning}, there exists a unique policy pair $\pi$ for the min-max optimization problem (\ref{prob:minmax2}). Therefore the entropy-regularized Markov game has the unique NE. Moreover,~\cite{cen2021fast} proves that an $\epsilon/2$-optimal NE for the \textbf{regularized} Markov game $\mathcal{G}_\theta'$ is an $\epsilon$-optimal NE for the original when the regularization parameter $\lambda\leq\frac{(1-\gamma)\epsilon}{2(\log(|\mathcal{A}_1|)+\log(|\mathcal{A}_2|))}$, which indicates the difference between NE for the regularized Markov game and NE for the original can be controlled by the regularization parameter. The definition for $\epsilon$-optimal NE is clearly defined in \cite{cen2021fast}, which means the optimal gap of the policy is smaller than $\epsilon$.



\subsection{Bi-level Optimization Scheme}\label{subsec:bi-level}
Our final goal is to arbitrate conflicts between self-interested players and the collective-interested designer by navigating all individuals to the desirable refined NE with higher system welfare. Section \ref{subsec:irg} introduces the incentive parameter $\theta$ to modify the intrinsic reward $r^i(\cdot;\theta)$ of Markov games, which is utilized as a tool to design the system objective and refine NE for arbitration. In this section, we formulate a bi-level optimization scheme with equilibrium constraints 
induced by such an arbitrating system. Among the bi-level optimization scheme, for the lower-level players, we aim at finding the NE for the given incentivized Markov game $\mathcal{G}_\theta'$. Regarding the upper-level designer, we define a system arbitrating objective targeting the optimal $\theta$ that could channel all individuals towards a desirable NE.


For the purpose of gradient computation in Section \ref{sec:algorithm} and experiments in Section \ref{sec:exp1}, we parameterize policy $\pi^i_{\phi^i}$ with $\phi^i\!\in\!\mathbb{R}^{d_i}$ for the player $i$ ($i\in\{1,2\}$) and denote the policy pair as $\pi_{\phi}\!=\!(\pi^1_{\phi^1}\!,\pi^2_{\phi^2})$, where $\phi\!=\!(\phi^1,\phi^2)\!\in\!\mathbb{R}^d, d=d_1+d_2$. To avoid the abuse of subscription, we omit $i$ of $\phi^i$ in $\pi^i_{\phi^i}$ (w.r.t $\pi^i_{\phi}$). The simplification will not cause confusion because the policy $\pi^i$ of a single player $i$ is only determined by $\phi^i$. Then the joint policy could also be rewritten as the policy pair $\pi_\phi=(\pi^i_\phi,\pi^{-i}_\phi)$, where $i \in \left\{1, 2\right\}$.
The NE of the regularized Markov game with  incentive-perturbed reward $\mathcal{G}_{\theta}^{'}$ is denoted as $\pi_{\phi_{*}(\theta)}=(\pi^1_{\phi_{*}(\theta)},\pi^2_{\phi_{*}(\theta)})$, where $\phi_{*}(\theta)=(\phi_{*}^1(\theta),\phi_{*}^2(\theta))$ is the parameter of NE under incentive $\theta$. The definition of NE for regularized Markov game in (\ref{equ:def of regularized NE}) can be rewritten as,
\vspace{-0.03in}
\begin{align*}
    V^{(i)}_{\pi_{\phi_{*}(\theta)}}(s;\theta)=\!\max_{\phi^i\in\mathbb{R}^{d_i}}V^{(i)}_{(\pi^i_{\phi},\pi^{\text{-}i}_{\phi_{*}(\theta)})}(s;\theta),\quad \forall s\in\mathcal{S}.
\end{align*}
\vspace{-0.25in}\\

\textbf{Arbitrating system.}
Suppose the arbitrating objective for game system $\mathcal{G}_{\theta}'$ is only determined by the incentive parameter $\theta$ and the joint policy $\pi_\phi$, written as mapping $f(\theta,\phi)\!:\!\Theta\times\mathbb{R}^d\rightarrow\mathbb{R}$. The system aims to minimize the arbitrating objective loss at the NE of $\mathcal{G}_{\theta}$, $f(\theta,\phi_*(\theta))$ .
Therefore, the arbitrating system could be formulated as a bi-level optimization problem:
\begin{align}\label{bilevel:prob2}
    \min_{\theta\in\Theta}\quad & f_*(\theta)=f(\theta,\phi_*(\theta))\notag\\
    s.t.\quad & \mathbb{E}_{\nu^{*}}\left[V^{(i)}_{\pi_{\phi_{*}(\theta)}}(s;\theta)\right]\notag\\
    &\quad=\!\max_{\phi^i\in\mathbb{R}^{d_i}}\!\mathbb{E}_{\nu^{*}}\left[V^{(i)}_{(\pi^i_\phi,\pi^{\text{-}i}_{\phi_{*}(\theta)})}(s;\theta)\right],\;\forall i\in\{1,2\}\end{align}

\textbf{Remark:} Eq.~\ref{bilevel:prob2} is our major objective. We primarily focus on efficiently learning the upper-level incentive parameter while treating the lower level as a black-box NE solver. In the following sections, we focus on deriving the first-order gradient of the upper-level loss function $f_*(\theta)$, then establish our \textit{Differentiable Arbitrating} (DA) framework based on it.



\section{Method}
\label{sec:algorithm}
The bi-level arbitrating system (Eq.~\ref{bilevel:prob2}) with NE constraints guides players indirectly via the incentive parameter $\theta$ to safeguard the arbitrating objective $f$. For the lower level, we attempt to determine the policy parameter of NE under a certain incentivized environment $\mathcal{G}_\theta'$, remarked as $\phi^{*}(\theta)$. For the upper-level designer, we aim to derive the first-order algorithm to find the optimal incentive parameter $\theta$ such that the objective $f_*(\theta)=f(\theta,\phi_*(\theta))$ can be optimized while players attain their NE. In such a bi-level problem, it is challenging to obtain the gradient of $f_*$, especially when the solver for regularized NE in the lower level is a black box, since the incentive parameter $\theta$ affects objective $f_*$ implicitly via refining the regularized NE for players. Our goal is to derive the gradient of $f_*$ w.r.t $\theta$ with NE constrains and then develop a backpropagation scheme that utilize the feedback of policy gradients from the lower level to establish the gradient of the incentive objective function at a regularized NE, $\nabla_{\theta}f(\theta, \phi^{*}(\theta))$. To be more specific, the gradient $\nabla_{\theta}f(\theta, \phi^{*}(\theta))$ can be derived as,
\begin{align}\label{eq:grad_f1}
    \nabla_{\theta}f_* = \left[\nabla_{\theta} {f}(\theta,\!\phi)\!+\!\nabla_{\theta}\phi^{*}(\theta)^\intercal\nabla_{\phi}{f}(\theta,\phi)\right]\!\big|_{\phi=\phi^{*}(\theta).}
\end{align}
where $\nabla_{\theta} {f}(\theta,\!\phi)$ and $\nabla_{\phi}{f}(\theta,\phi)$ are known while the NE gradient $\nabla_{\theta}\phi_{*}(\theta)$ need to be derived based on the policy gradient feedback from the lower level.
The computing of the NE gradient $\nabla_{\theta}\phi^{*}(\theta)$ is in detail discussed in Section \ref{section:differentiable computing} and then a general backpropagation framework that for the regularized NE is proposed in Section \ref{section:computing framework}.
\subsection{Gradient for the Incentive}
\label{section:differentiable computing}
In this section, we will derive the gradient at NE $\nabla_{\theta}f(\theta, \phi^{*}(\theta))$, in which the key point is to derive NE gradient $\nabla_{\theta}\phi^{*}(\theta)$. For notation simplification, we denote $u^{(i)}_{\theta}(\phi;s)$ as the gradient of $V^{(i)}_{\phi}(s;\theta)$ w.r.t $\phi^i$ then we define
\begin{align}
\label{def: u*}
    u_{\theta}(\phi)&:=\mathbb{E}_{\nu^{(*)}}[u_{\theta}(\phi;s)],\\\notag
    u_{\theta}(\phi;s)\!&:=\!(u_{\theta}^{(1)}(\phi;s),u_{\theta}^{(2)}(\phi;s))
    =(\nabla_{\phi^1}V^{(1)}_{\pi_{\phi}}(s;\theta),\nabla_{\phi^2}V^{(2)}_{\pi_{\phi}}(s;\theta)).
\end{align}
Then, the gradient of $u_\theta(\phi)$ w.r.t. to $\theta$ and $\phi$ are straightforward as follows,
  \begin{align}
  \label{notation:gradient u}
      \nabla_{\theta}u_{\theta}(\phi)\!=\!\left[
      \begin{matrix}
        \mathbb{E}_{\nu^{(*)}}\!\nabla^2_{\theta\phi^1}V^{(1)}_{\pi_{\phi}}(s;\theta)\\
        \mathbb{E}_{\nu^{(*)}}\!\nabla^2_{\theta\phi^2}V^{(2)}_{\pi_{\phi}}(s;\theta)\\
      \end{matrix}
      \right]\in\mathbb{R}^{d\times m},\\
      \nabla_{\phi}u_{\theta}(\phi)=\left[
      \begin{matrix}
        \mathbb{E}_{\nu^{(*)}}\nabla^2_{\phi\phi^1}V^{(1)}_{\pi_{\phi}}(s;\theta)\\
        \mathbb{E}_{\nu^{(*)}}\nabla^2_{\phi\phi^2}V^{(2)}_{\pi_{\phi}}(s;\theta)\\
      \end{matrix}
      \right]\in\mathbb{R}^{d\times d}.
      \label{notation:gradient u2}
  \end{align}
The following Lemma \ref{lma:V2} implies the computation of $\nabla^2_{\theta\phi^i}V^{(i)}_{\pi_{\phi}}(s;\theta)$ and $\nabla^2_{\phi\phi^i}V^{(i)}_{\pi_{\phi}}(s;\theta)$ can be explicitly derived based on the policy gradient, which are backpropagated from NE solver in the lower level. The computation details for
the gradients and Hessians
in Lemma \ref{lma:V2} are presented throughout Lemma \ref{lma:V} in appendix due to limitation of space. Then in light of (\ref{notation:gradient u}-\ref{notation:gradient u2}) and Lemma \ref{lma:V2}, $\nabla_{\theta}\phi^{*}(\theta)$ can be expressed in policy gradient as derived in the following Lemma \ref{lma:phitheta}, inducing the final expression of objective gradient at NE $\nabla_{\theta} f_{*}(\theta)$.
\begin{lemma}\label{lma:V2}
  In an incentive regularized Markov game $\mathcal{G}_{\theta}^{'}$, we denote $\pi^i_{\phi^i}:=\pi^i_{\phi^i}(a_t^i|s_t)$, $\pi^{-i}_{\phi^{-i}}:=\pi^{-i}_{\phi^{-i}}(a_t^{-i}|s_t)$ and $\pi_{\phi}:=\pi_{\phi}(a_t|s_t)$ for notation simplicity,
then we have
\begin{align*}
    &\nabla_{\theta}V^{(i)}_{\pi_{\phi}}(s;\theta)=\nabla_{\theta}V^{i}_{\pi_{\phi}}(s;\theta);\quad
    \nabla^2_{\theta\phi^i}V^{(i)}_{\pi_{\phi}}(s;\theta)=\nabla^2_{\theta\phi^i}V^{i}_{\pi_{\phi}}(s;\theta);
\end{align*}
\begin{align*}
    &\nabla\!_{\phi^i}\! V^{(i)}_{\pi_{\phi}}(s;\theta)=\nabla_{\phi^i} V^{i}_{\pi_{\phi}}(s;\theta)-\!\lambda\mathbb{E}_{D^{\pi_{\phi}}}\!\left[\sum_{t=0}^{T-1}\!\gamma^t\nabla_{\phi^i}\!\log\pi^i_{\phi^i}\right.\notag\\
    &\quad+\left.\left(\sum_{t=0}^{T-1}\!\gamma^t\!\left(\log\!\pi^i_{\!\phi^i}-\log\!\pi^{-i}_{\!\phi^{-i}}\right)\right)\!\left(\sum_{t=0}^{T-1}\!\nabla_{\phi^i}\log\pi^i_{\phi^i}\right)\right];\\
    &\nabla^2_{\phi\phi^i}V^{(i)}_{\pi_{\phi}}(s;\theta)=\nabla^2_{\phi\phi^i}V^{i}_{\pi_{\phi}}(s;\theta)\\
    &\quad-\lambda\mathbb{E}_{D^{\pi\!_\phi}}\!\left[\left(\sum_{t=0}^{T-1}\!\nabla\!_{\phi^i}\!\log\!\pi^i_{\phi^i}\right)\left(\sum_{t=0}^{T-1}\!\gamma^t\left(\nabla\!_{\phi}\!\log\!\pi^i_{\phi^i}-\nabla\!_{\phi}\!\log\!\pi^{-i}_{\phi^{-i}}\right)\right)^{\!\!\intercal}\right.\\
    &\quad+\left(\!\sum_{t=0}^{T-1}\!\gamma^t\!\left(\log\pi_{\phi^i}^i-\log\pi_{\phi^{-i}}^{-i}\!\right)\right)\!
     \left(\!\sum_{t=0}^{T-1}
      \!\nabla_{\phi^i}\!\log\pi_{\phi^i}^i\right)\!\left(\!\sum_{t=0}^{T-1}\!\!\nabla_{\phi}\!\log\!\pi_{\phi}\right)^{\!\intercal}\notag\\
      &\quad+\left(\sum_{t=0}^{T-1}\!\gamma^t\left(\log\pi_{\phi^i}^i-\log\pi_{\phi^{-i}}^{-i}\!\right)\right)\!\left(\sum_{t=0}^{T-1}\nabla^2_{\phi\phi^i}\!\log\pi_{\phi}\right)^{\!\intercal}\\
      &\quad+\left(\!\sum_{t=0}^{T-1}\!\gamma^t\nabla_{\phi^i}\!\log\! \pi^i_{\phi^i}\right)\!\left(\sum_{t=0}^{T-1}\!\nabla_{\phi}\log \pi_{\phi}\!\right)^{\!\!\intercal}\left.+\!\left(\!\sum_{t=0}^{T-1}\!\!\gamma^t\nabla^2_{\phi\phi^i}\!\log \!\pi_{\phi}\right)
      \right],
  \end{align*}

  where $\mathbb{E}_{D^{\pi\!_\phi}}[\;\!\!\cdot\!\;\!]$ denotes the expectation over the trajectory $\tau\sim D^{\pi_\phi}(\tau)$
  and $\nabla_{\theta}V^{i}_{\pi_{\phi}}(s;\theta)$, $\nabla_{\phi^i} V^{i}_{\pi_{\phi}}(s;\theta)$, $\nabla^2_{\theta\phi^i}V^{i}_{\phi}(s;\theta)$ and $\nabla^2_{\phi\phi^i}V^{i}_{\phi}(s;\theta)$ are computed from Lemma \ref{lma:V} in appendix.
\end{lemma}
Lemma \ref{lma:V2} utilizes the policy gradient information, which is easily accessible in a NE solver, to describe how self-interested players react to the adjustment in the incentive parameter.  
It is straightforward to prove but is essential to provide the lower-level information embedded in the policy gradient to guide the direction of updates for the incentive parameter against a high collective loss.\looseness=-1

In light of Lemma \ref{lma:V2}, Lemma \ref{lma:phitheta} presents the gradient of NE w.r.t. the incentive parameter, which is crucial to leverage the backpropagated information from the lower level to update the incentive parameter. It describes how the equilibrium among players is affected by the incentive from the designer and shows that such information is related to the Hessian of $V^{(i)}_{\pi_{\phi}}(s;\theta)$ at NE.
\begin{lemma}\label{lma:phitheta}
   The following formula holds for the policy parameter $\phi^{(*)}(\theta)$ of NE $\pi_{\phi^{*}(\theta)}$ for incentive regularized Markov game $\mathcal{G}_{\theta}^{'}$:
  \begin{align}\label{eq:grad_phi_new}
      \nabla_{\theta}\phi^{(*)}(\theta)=-\left[\nabla_{\phi}u_{\theta}(\phi)\right]^{-1}\nabla_{\theta}u_{\theta}(\phi)\big|_{\phi=\phi^*(\theta)}.
  \end{align}
\end{lemma}
Then substituting Eq.~(\ref{eq:grad_phi_new}) into Eq.~(\ref{eq:grad_f1}) leads to the final explicit expression of $\nabla_{\theta} f_{*}(\theta)$:
\begin{align}\label{eq:grad_f}
    &\nabla_{\theta}
    f_{*}(\theta,\phi^*(\theta))\notag\\
    =&\!\left[\nabla_{\theta} {f}(\theta,\!\phi)\!\!-\!\!\nabla_{\theta}u_{\theta}(\phi)^\intercal\left[\nabla_{\phi}u_{\theta}(\phi)\right]^{-1}\nabla_\phi{f}(\theta,\phi)\right]\big|_{\phi=\phi^{*}(\theta).}
\end{align}

\subsection{Differentiable Arbitrating}
\label{section:computing framework}

  

Based on the gradient computation in Section \ref{section:differentiable computing}, we propose the \textit{Differentiable Arbitrating }(DA) framework to solve the arbitrating system (Eq.~\ref{bilevel:prob2}). DA differentiates the upper-level designer's objective function through NE and backpropagates the gradient information of policies from the lower level to derive the upper-level gradient. In this scheme, the black-box solver for the regularized NE in the lower level could be any methods containing gradient information of the policies, for example, SAC~\cite{haarnoja2018soft}, MADDPG~\cite{lowe2017multi}, MAAC~\cite{iqbal2019actor}, Policy Extragradient Method~\cite{cen2021fast} and Entropy-regularized OMWU~\cite{cen2022faster}. 
\begin{algorithm}[H]
\floatname{algorithm}{Framework}
\caption{DA: Differentiable Arbitrating in MARL.}
\label{frame_IncenRl}
\begin{algorithmic}
\STATE {\bfseries Input:} $\beta_k$: learning rate for upper-level iteration
\STATE {\bfseries Output:} $\theta$: incent. param.; $\phi$: param. of policy $\pi_\phi$
\STATE Initialize the incentive parameter $\theta=\theta_0$.

\FOR{k=0,1,...}
  \STATE Initialize param. $\phi=\phi_0$ of policy $\pi_{\phi_0}$ for game $\mathcal{G}_{\theta_k}'$.
  
  \FOR{t=0,1...}
  \STATE Update $\phi=\phi_t$ with NE solver until corresponding policy $\pi_{\phi_t}$ converge to the Nash equilibrium of game $\mathcal{G}_{\theta_k}'$, w.r.t $\phi^{*}(\theta_k)$.
  \ENDFOR
  \STATE Update incentive parameter
  $$\theta=\theta_{k+1}\leftarrow\theta_k-\beta_k\nabla f_{*}(\theta_k),$$
  where $\nabla f_{*}(\theta_k)$ is defined in (\ref{eq:grad_f}).
\ENDFOR
\end{algorithmic}
\end{algorithm}
\vspace{-0.05in}
The workflow of DA is shown in Framework \ref{frame_IncenRl}. 
In this work, we specifically develop the convergence guarantee for the DA framework with lower-level NE solvers which have convergence guarantee and
implement our DA framework using a multi-agent variant of SAC, which will be described in detail later in Section \ref{sec:imp}.

\textbf{DA framework is a \emph{first-order} method.} This is because the update of upper-level incentive $\theta$ only relies on the first-order gradient $\nabla f_{*}(\theta)$ without any hessian information of $\theta$. The second-order gradients defined in Lemma 4.1 are hessian matrices of lower-level value function $V$ and are used for the computation of first-order gradient $\nabla f_{*}(\theta)$.

  
\textbf{Extension to general settings.}
Although DA framework is described under two-player zero-sum fully observable Markov games, it could be extended to more general settings: \textit{N-players}, \textit{general-sum} and \textit{partially observable Markov decision process (POMDP)}. 
For N-player general-sum Markov games, the game $\mathcal{G}_\theta$ is required to be monotone to ensure the entropy-regularized game $\mathcal{G}_\theta'$ to be strongly-monotone and therefore has unique NE. Then the DA framework could leverage a proper lower-level black-box NE solver for the (regularized) NE. 
For POMDPs, some minor adjustments are required, replacing the state-based policy $\pi^i(a^i|s)$ with the observation-based policy $\pi^i(a^i|o^i)$, where $o^i\!=\!O^i(s)$ is the observation of the player $i$ in state $s$. Such adjustments will not change the gradient computation, which is essential for the DA framework. The detailed discussions of problem formulation and the gradient derivation for the general settings are deferred to Appendix \ref{supp:extension}. 
% ...


\vspace{-0.02in}
\section{Convergence analysis}\label{sec:cov}
In this section, we develop the convergence analysis for the DA framework with proper lower-level NE solvers in the two-player zero-sum setting. 
Section \ref{subsec:inner_loop} discuss lower-level NE solvers with convergence guarantee in literature that could be applied to our DA framework. Section \ref{subsec:outer_loop} guarantees the convergence of incentive parameter $\theta$ in the upper level for the bi-level DA framework under convergent lower-level NE solvers.


\vspace{-0.05in}
\subsection{Convergence of the Lower Level}\label{subsec:inner_loop}

The convergence of the bi-level DA framework requires the lower-level NE solver for having the convergence guarantee to the NE $\pi_*$ of the entropy-regularized Markov game $\mathcal{G}_\theta'$. There are diverse NE solvers that have good empirical performance in practical games~\cite{haarnoja2018soft,lowe2017multi,iqbal2019actor}. Some proposed NE solvers have the theoretical convergence proofs~\cite{szepesvari1999unified,singh2000convergence,hu2003nash}. For the illustration of the DA framework convergence, we take two NE solvers proposed under entropy-regularization that are directly related with our setting.
It is known that the Policy Extragradient Method (PEM)~\cite{cen2021fast} and Entropy-regularized OMWU~\cite{cen2022faster} all have the convergence guarantee under our setting. This section restates the main convergence results for PEM and Entropy-regularized OMWU for the purpose of completeness of convergence guarantee for DA framework.

\begin{proposition}[Theorem 3 in~\cite{cen2021fast}]\label{prop:PEM}
 Assume $|\mathcal{S}_1|\geq|\mathcal{S}_2|$ and $\lambda\leq 1$. Setting learning rate $\eta_t\!=\!\eta\!=\!\frac{1-\gamma}{2(1+\lambda(\log|\mathcal{S}_1|+1-\gamma))}$, the Policy Extragradint Method (Algorithm \ref{alg_PEM} in supplementary) takes no more than $\widetilde{\mathcal{O}}\left(\frac{1}{\lambda(1-\gamma)^2}\log^2\left(\frac{1}{\epsilon}\right)\right)$ iterations to achieve the policy pair $\pi=(\pi_1,\pi_2)$ that satisfies:
 \vspace{-0.08in}
    \begin{align*}
    \max_{\pi_1',\pi_2'}\mathbb{E}_{s\sim\rho}\left[\left(V^{(1)}_{\pi_1',\pi_2}(s;\theta)-V^{(1)}_{\pi_1,\pi_2'}(s;\theta)\right)\right]\leq\epsilon,
  \end{align*}
\vspace{-0.1in}\\
  where $\rho$ is an arbitrary distribution over the state space $\mathcal{S}$.
\end{proposition}

\begin{proposition}[Theorem 1 in~\cite{cen2022faster}]\label{prop:OMWU}
 Setting $0<\eta\leq\frac{(1-\gamma)^3}{32000|\mathcal{S}|}$ and $\alpha_t=\eta\lambda$, the Entropy-regularized OMWU (Algorithm \ref{alg_EROMWU} in supplementary) takes no more than $\widetilde{\mathcal{O}}\left(\frac{1}{(1-\gamma)\eta\lambda}\log\frac{1}{\epsilon}\right)$ iterations to achieve achieve the policy pair $\pi=(\pi_1,\pi_2)$ that satisfies:
  \vspace{-0.05in}
    \begin{align*}
      \max_{s\in\mathcal{S},\pi_1',\pi_2'}\left(V^{(1)}_{\pi_1',\pi_2}(s;\theta)-V^{(1)}_{\pi_1,\pi_2'}(s;\theta)\right)\leq\epsilon.
  \end{align*}
\end{proposition}

Proposition \ref{prop:PEM} and Proposition \ref{prop:OMWU} demonstrate that the PEM and Entropy-regularized OMWU algorithms converges to the NE for the lower-level incentivized entropy-regularized Markov games $\mathcal{G}_\theta'$ in linear rate.

\subsection{Convergence of the Upper Level}\label{subsec:outer_loop}
Assume that $\theta^{*}=\arg\min_{\theta}f_{*}(\theta)$ represents the optimal point of incentive parameter $\theta$ in the arbitrating system (\ref{bilevel:prob2}). In this section, we are going to show that the framework \ref{frame_IncenRl} guarantees that the incentive parameter $\theta$ converges to $\theta^{*}$ under certain assumptions.

\begin{assumption}\label{asm:f_1}
We assume that $\nabla_\theta f_*(\theta)$ is L-Lipschitz continuous (L-smooth) w.r.t $\theta\in\Theta$, which means
\begin{align*}
    \norm{\nabla_{\theta} f_{*}(\theta)-\nabla_{\theta} f_{*}(\theta')}\leq L\norm{\theta-\theta'},\quad \forall \theta, \theta' \in \Theta
\end{align*}
We also assume that $f_*(\theta)$ is bounded w.r.t $\theta\in\Theta$, which means there exists $M>0$ such that 
\vspace{-0.05in}
\begin{align*}
    |f_{*}(\theta)|\leq M,\quad \forall \theta \in \Theta
\end{align*}
\end{assumption}
L-smooth assumption is a common assumption in convergence proofs for different algorithms~\cite{kingma2014adam,yang2016unified,ramezani2018generalization,li2018sharp,abbaszadehpeivasti2022exact}. We also illustrate a concrete example to justify Assumption \ref{asm:f_1} is rational in practical setting in Appendix \ref{sup:asmf1}. Assumption \ref{asm:f_1} does not require the composite function $f_*(\theta)=f(\theta,\phi^*(\theta))$ to be convex and the convergence guarantee for the upper level even suitable for non-convex cases.

\begin{theorem}\label{thm:outer_conv}
Suppose that Assumption \ref{asm:f_1} holds and the lower-level NE solver has convergence guarantee. In Framework \ref{frame_IncenRl}, let the update rule $\theta_{k+1}=\theta_k-\beta_k\nabla_{\theta} f_{*}(\theta_k)$
for incentive parameter $\theta_k$ run for $T$ iterations with learning rates $\beta_k=\frac{1}{L}$, then we have
\begin{align*}
    \min_{k=0,...,T}\norm{\nabla_\theta f_*(\theta_k)}^2\leq\frac{2L(f_*(\theta_0)-f_*(\theta^*))}{T+1}\leq\frac{4LM}{T+1}.
\end{align*}
\end{theorem}
Theorem \ref{thm:outer_conv} implies that there exist a subsequence in the sequence of gradient norms $\{\norm{\nabla_\theta f_*(\theta_k)}\}_{k=0}^{\infty}$ which converges to zero. Therefore, Theorem \ref{thm:outer_conv} guarantees that the incentive parameter $\theta$ could always converge to a sub-optimal point with sublinear convergence rate in the upper level. It is easy to know if the composite function $f^*(\theta)$ is convex, it converges to the global optimal point.
The detailed proof is shown in Appendix \ref{sup:proof_outer}.

\section{Implementations}\label{sec:imp}
The direct implementation of PEM and Entropy-regularized OMWU are not practical/efficient although they have the theoretical convergence guarantee. PEM is a double-loop algorithm that invokes solving the NE of an entropy-regularized Matrix game for each inner-loop. Although Entropy-regularized OMWU is a single-loop algorithm that adapts a two-timescale iteration, its update rules are for the tabular setting, in which the policy are per state updated. 
Therefore, we modify the vanilla algorithm of Entropy-regularized OMWU with some techniques inspired by the multi-agent soft-actor-critic (MASAC) and propose the DA-SAC algorithm (Algo. \ref{alg_IncenRl_sac}). We organize the modifications as the following: 


\begin{enumerate}
    \item[i)] In stead of the per-state policy update in Entropy-regularized OMWU,  we optimize a separate policy model via minimizing its KL divergence from the exponential energy function, i.e., we establish an actor-critic structure as widely accepted in SAC.
    
    
    \item[ii)] We use one-step gradient update with great practical success accepted by most of MARL practical algorithms, e.g. SAC~\cite{haarnoja2018soft}, MADDPG~\cite{lowe2017multi} and MAAC~\cite{iqbal2019actor}, instead of the extragradient technique introduced in Entropy-regularized OMWU. 
    
    

    \item[iii)] We follow SAC to set the energy function as a multiple of the Q function, where the scale factor is the reciprocal of the auto-tuned temperature coefficient~\cite{haarnoja2018soft2}.
\end{enumerate}

Updated pseudo-code and all remaining details are included in Appendix \ref{supp:exp}.

\begin{figure*}[h]
\centering
\subfigure[RWS ($\theta_0$)]{\includegraphics[width=0.21\textwidth]{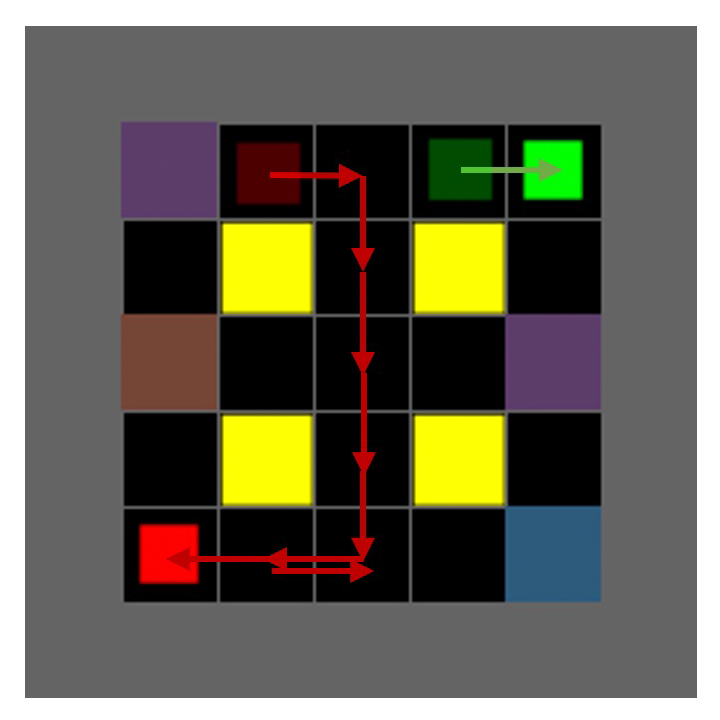}}
\label{subfigure:RWS0}
\subfigure[RWS ($\theta^*$)]{\includegraphics[width=0.21\textwidth]{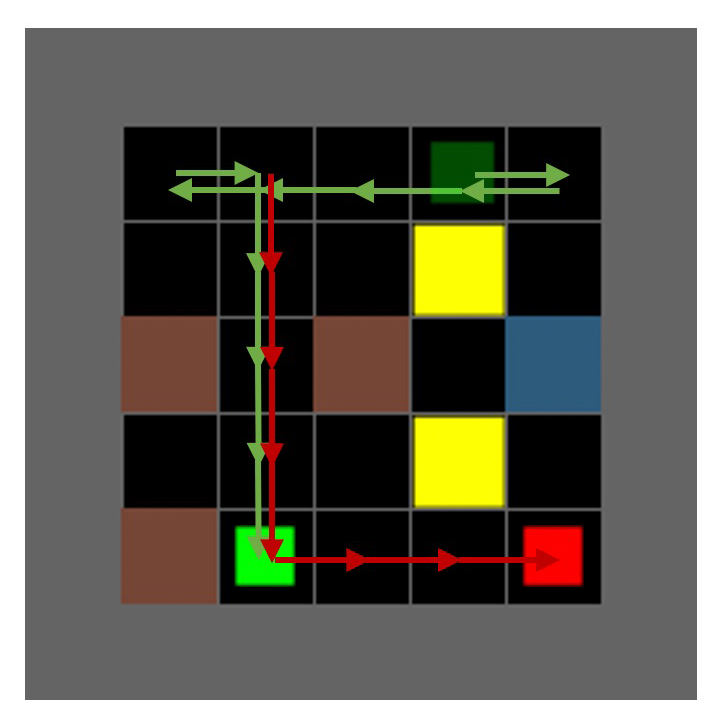}}
\subfigure[PP ($\theta_0$)]{\includegraphics[width=0.21\textwidth]{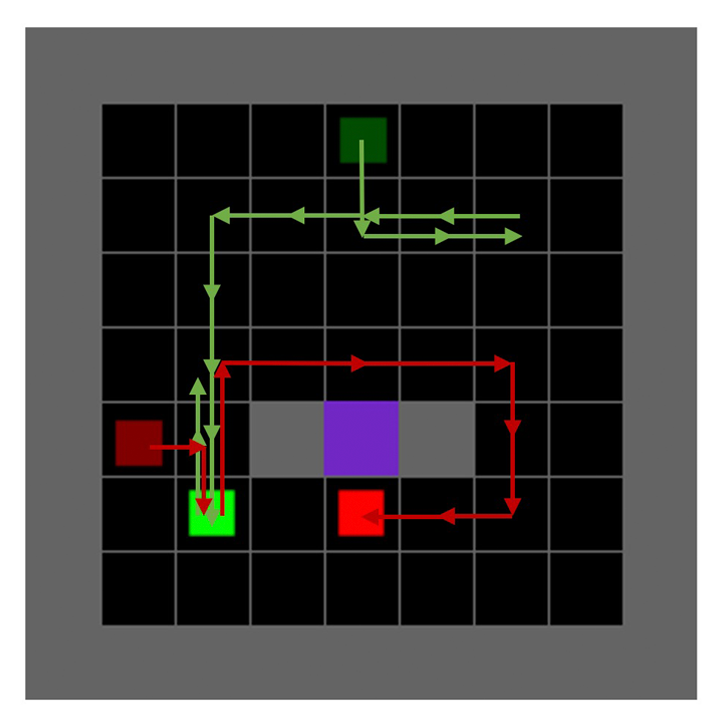}}
\subfigure[PP ($\theta^*$)]{\includegraphics[width=0.21\textwidth]{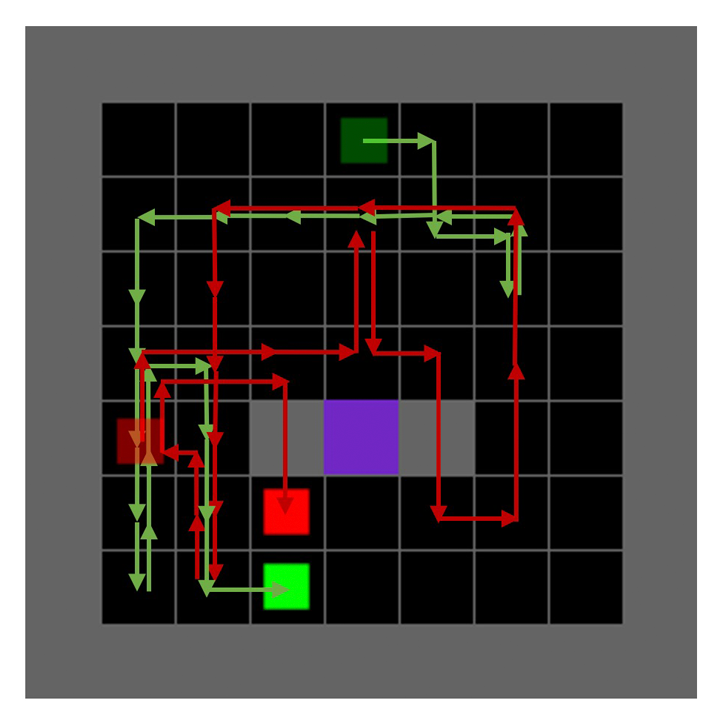}}
\vspace{-0.15in}
\caption{Trace for incentivized RWS and incentivized PP with different incentive parameters $\theta_0$ and $\theta^*$ (left to right). $\theta_0$ is the initial incentive parameter and $\theta^*$ is the optimized parameter with DA framework.
The NEs after arbitration (Figure (b) \& (d)) ensure higher exploration rate in the grid-world games compared with the initial NEs (Figure (a) \& (c)) without arbitration.}
\label{fig:result}
\end{figure*}

\section{Experiments}\label{sec:exp1}
We evaluate the performance of the DA framework in sense of sample efficiency for the upper level in two zero-sum Markov games with two players and their incentivized variant, running with scissors (RWS)~\cite{vezhnevets2020options} which contains cyclic reinforcement learning challenges~\cite{balduzzi2018re,leibo2019autocurricula,omidshafiei2019alpha} and a standard predator-prey, both of which are implemented on a grid-world environment~\cite{gym_minigrid}.

\subsection{Evaluation Environments}\label{subsec:env}
\textbf{Running with scissors (RWS).} 
In $5\times5$ grid-world RWS,
resources (rock, scissors, or paper) are tailed in resources pools (Figure \ref{fig:RWS}). Three of them are deterministic pools spaced with fixed resources and six others are nondeterministic filled with random resources.
Two players randomly spawn among the free grids and fully observe the environment and resources owned by their opponent.
Each chooses to move one grid in four directions at one step. When the player $i$ steps on the grid with the resource, it collects the resource to its inventory $v^i$, and the resource is removed from the grid. 
Four rocks, papers and scissors are randomly distributed to two players at the initial state and each player is assigned at least one for each resource. After 25 steps, the confrontation occurs and the payoffs for each player, $r^0$ and $r^1$, are calculated on the basis of the standard antisymmetric matrix $M$~\cite{hofbauer2003evolutionary} as
\vspace{-0.1in}
\begin{align*}
r^0=\frac{v^0}{\|v^0\|} M \left(\frac{v^1}{\|v^1\|}\right)^\intercal = -r^1,
\quad \text{where  }M =  \left[\begin{array}{ccc}
0 & -1 & 1\\
1 & 0 & -1\\
-1 & 1 & 0
\end{array} 
\right].
\end{align*}
\vspace{-0.1in}\\
Besides, the rules of the game are assumed to be a blank to the players and they must explore to discover them.

\begin{figure}[t]
    \vspace{-0.05in}
    \includegraphics[width=0.35\textwidth]{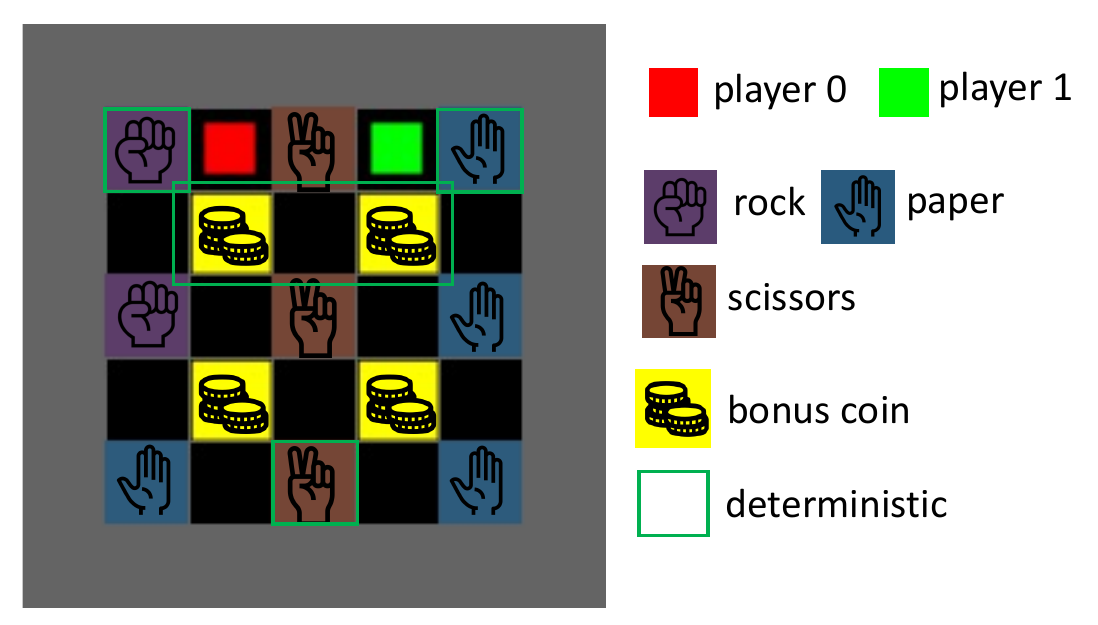}
    \vspace{-0.1in}
    \caption{Configuration of the Rock-with-Scissors (RWS). Two players are randomly spawned among the free grids and collect resources for confrontation.
    The bonus coins are the incentive added by the designer to encourage exploration.
    }
    \label{fig:RWS}
\end{figure}

RWS is an extension of the classic matrix game rock-paper-scissors (RPS) with increasing complexity, and the game-theoretic structures of the RPS are satisfied in RWS. Since NE is to keep the same number of each resource in their inventory~\cite{van2007rock}, players tend to stay around a grid after they achieve NE instead of continuing to explore, while the upper-level designer hopes more grids can be explored. Given such a conflict, the designer intervenes in the rewards by scattering gold coins as an incentive to channel players exploring more grids. 

There are four grids that generate coins. Two of them are filled with fixed bonuses 0.5, and others are filled with adjustable bonuses marked as the incentive parameter $\theta\!=\!(\theta_1,\theta_2)\in\mathbb{R}^2$. Once the player $i$ collects coins, the corresponding bonuses are added to its total reward. Conversely, the bonuses collected by player $i$ are deducted from its opponent $-i$. The total reward for the player $i$ is
\begin{align*}
    \vspace{-0.2in}
    r^i_{total} =& r^i + \sum\nolimits_{j=1,2} \theta_j (I^i_j-I^{-i}_j)+0.5(N^i_{\text{fixed}}-N^{-i}_{\text{fixed}}),
\end{align*}
where $I^i_j\!=\!1$ if the coin with incentive bonus $\theta_j$ is collected by the player $i$, otherwise $I^i_j\!=\!0$ and $N^i_{\text{fixed}}$ is the number of fixed coins collected by the player $i$. Then the exploration rate ($ER$) can be defined as
\begin{align}\label{eq:ER}
\vspace{-0.05in}
    ER(\theta) = \frac{\sum_iG_e^i(\theta)}{G_{\text{total}}},
\end{align}
where the $G_e^i(\theta)$ denotes the number of grids explored by player $i$ and $G_{\text{total}}$ denotes the total number of grids. Therefore, the designer's loss is defined as $f_*(\theta) = 1-ER(\theta)$.



\begin{figure}[t]
    \vspace{-0.05in}
    \includegraphics[width=0.35\textwidth]{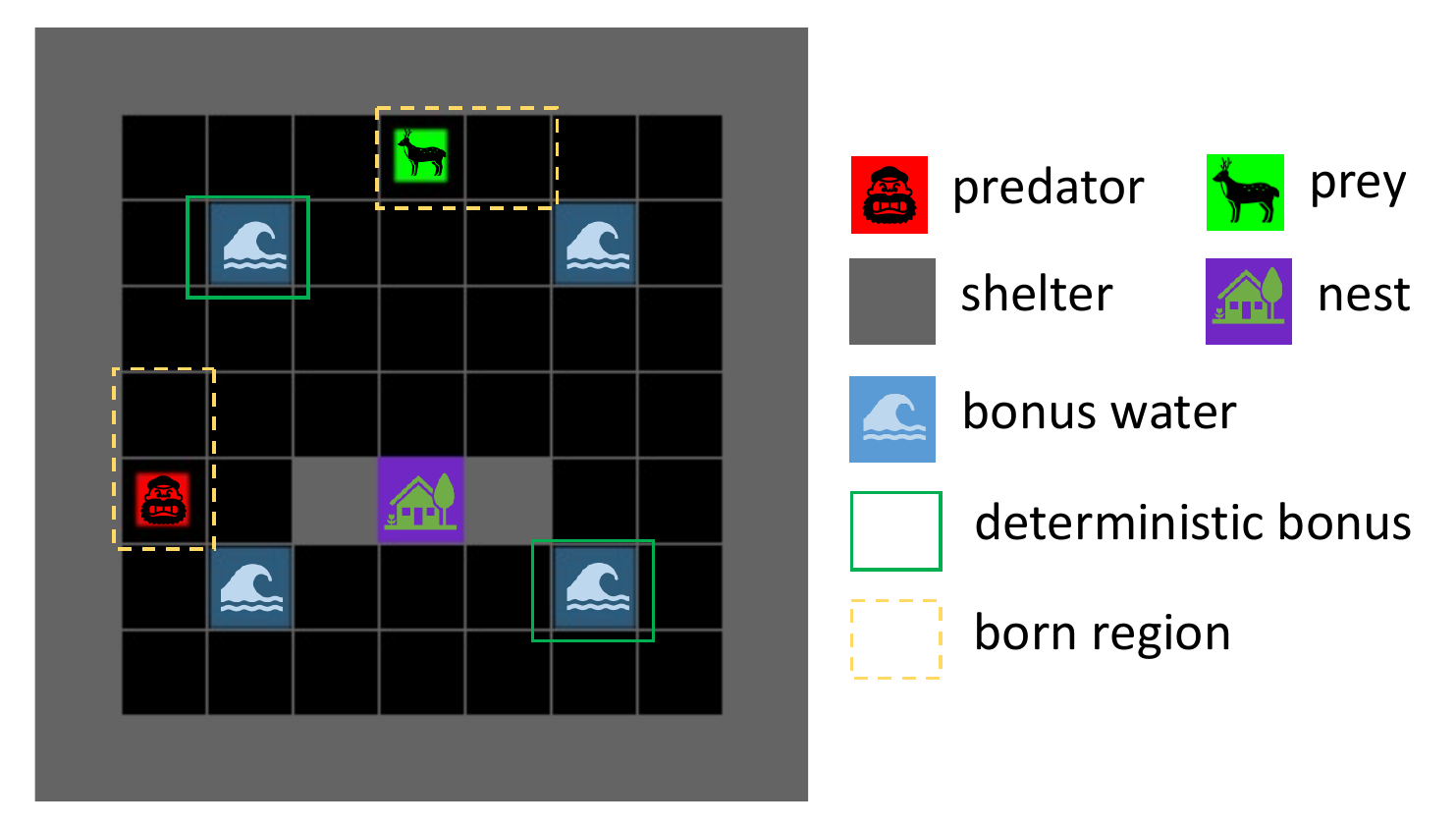}
    \vspace{-0.1in}
    \caption{Configuration of the Predator-Prey (PP). The predator aims to catch the prey and the prey aims to return to the nest. Both of them are randomly spawned in born region. The water pools are set as the incentive by the designer.}
    \label{fig:pp}
\end{figure}


\textbf{Predator-prey (PP).} The environment for the predator-prey is set in a $7\times7$ grid-world as Figure \ref{fig:pp}, containing the prey's nest, two shelters, and water pools. Both the predator and prey cannot stop at the shelter. At each step, the prey can move one grid and the predator can move one or two grids in four directions. The game is terminated under three conditions: i) the episode achieves maximum length 25; ii) the prey returns the nest, then the reward for the prey, $r_{\text{prey}}$, is $+1$ and the reward for the predator, $r_{\text{pred}}$, is $-1$; iii) the predator catches the prey (they arrive on the same grid at the same time), then the reward for prey, $r_{\text{prey}}$, is $-1$ and the reward for predator, $r_{\text{pred}}$, is $+1$. From a whole-ecosystem perspective, like the dissemination of plant seeds, the designer stimulates them to explore more places by setting four pools with different volumes of water in the four fixed grids. Two of them are filled with fixed water volume with an additional 0.1 bonuses and others are filled with adjustable water volume, whose bonuses are denoted as the incentive parameter $\theta = (\theta_1,\theta_2) \in \mathbb{R}^2$. The player gets the corresponding bonus if it finds and drinks up a pool of water while its opponent gets the corresponding penalty. Then the total pay-off for the predator is
\begin{align*}
    r_{\text{pred}}\!= &I_{\text{catch}} \!-\! I_{\text{nest}} \!+\! \sum\nolimits_{j=1,2} \theta_j(I_j^{\text{pred}}-I_{j}^{\text{prey}})\!+\!0.5(N_{\text{fixed}}^{\text{pred}}-N_{\text{fixed}}^{\text{prey}}),
\end{align*}
where $I_{\text{catch}}$ and $I_{\text{nest}}$ are two indicators for that the predator catch the prey or the prey back to the nest, $I_j^{\text{pred}}$ (or $I_{j}^{\text{prey}}$) is the indication functions of whether the pool with $\theta_j$ amount of water is found by the predator (or prey) and $N_{\text{fixed}}^{\text{pred}}$ (or $N_{\text{fixed}}^{\text{prey}}$) represents the number of fixed pools discovered by the predator (or prey). In addition, we define the reward for prey as $r_{\text{prey}} = -r_{\text{pred}}$. The arbitrating objective for PP is the same as what is defined for RWS. 


\subsection{Results and Analysis}
We use \textit{GridSearch(M)} to denote the zeroth-order method,  which spreads $M$ number of points into the feasible region for incentive parameter $\theta$, as a competitor to DASAC. Additionally, we adopt Bayesian optimization to improve the efficiency of the zeroth-order method, denoted as \textit{BayesOpt}, which is also used in~\cite{mguni2019coordinating}.
The experimental details of GridSearch(M) and BayesOpt are included in Appendix \ref{supp:GSM}.
Figure \ref{fig:f} shows the trend of the arbitrating objective $f_*$ 
after applying DASAC onto the environments. The green dotted line in Figure \ref{fig:f} is the best objective
score of GridSearch(100).
On the one hand, Figure \ref{fig:f} shows the differentiable first-order framework we propose is capable to admit a better incentive parameter $\theta$ with a lower objective loss. On the other hand, Figure \ref{fig:f} indicates that GridSearch(M) requires 100 samples, in another word, solving the lower-level Markov game 100 times, to achieve an equivalent performance of DASAC. 
DASAC can achieve better efficiency (~5 outer-loop iterations) than these zeroth-order methods even though BayesOpt (~12 outer-loop iterations, >2 times slower than ours) has improved the efficiency of GridSearch (100 outer-loop iterations). After considering the gradient computation, DASAC requires ~3.5h for one iteration and the zeroth-order methods (Gridsearch and BayesOpt) requires ~3h for one iteration under the same GPU setting. Our method is ~2x faster than BayesOpt and ~17x faster than Gridsearch(100).
It means
our first-order framework requires fewer evaluations for the upper level to obtain a desirable NE policy, leading to a higher sample efficiency. 

\begin{figure}[t]
\centering
\includegraphics[width=.85\linewidth]{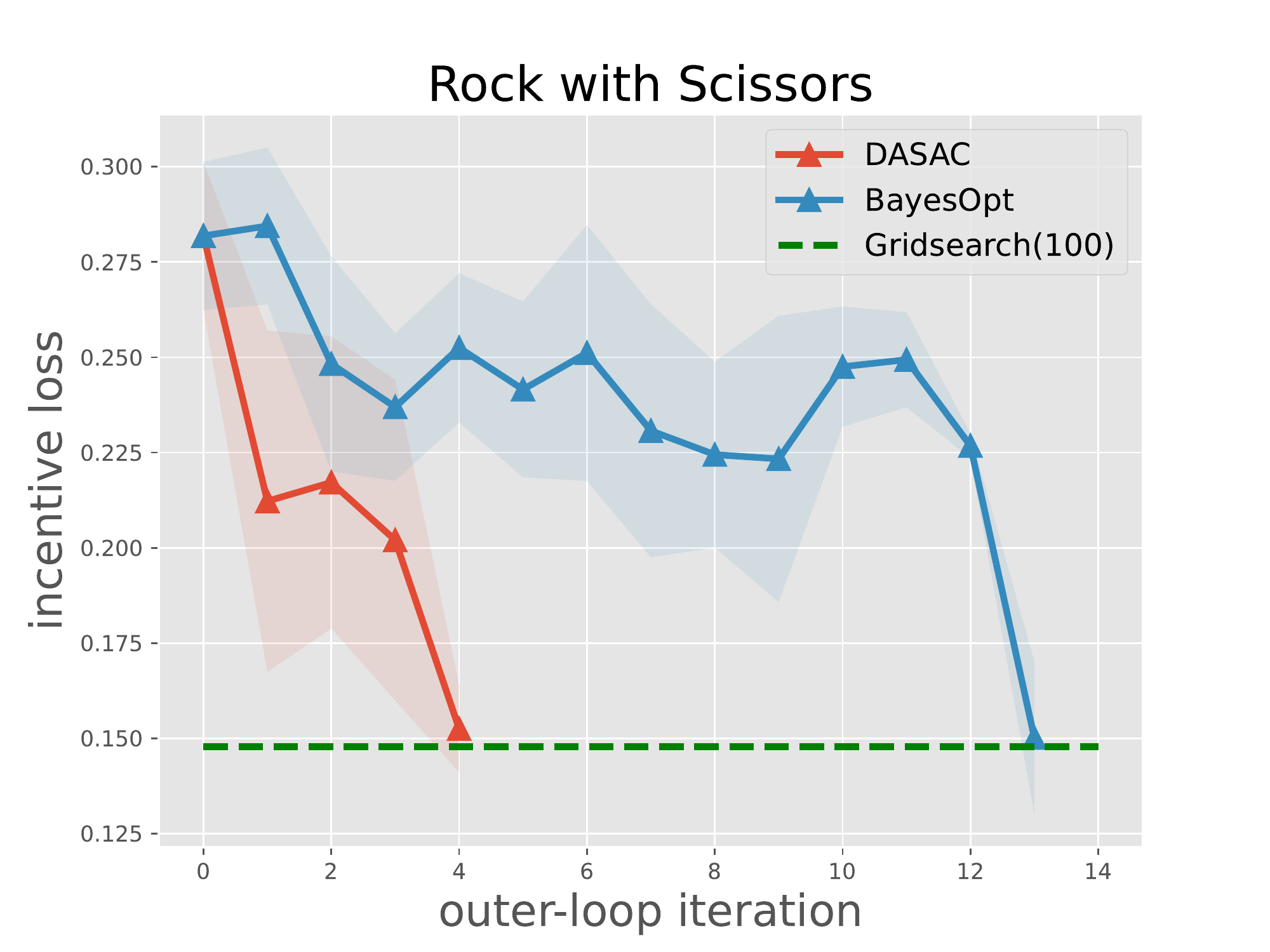}
\includegraphics[width=.85\linewidth]{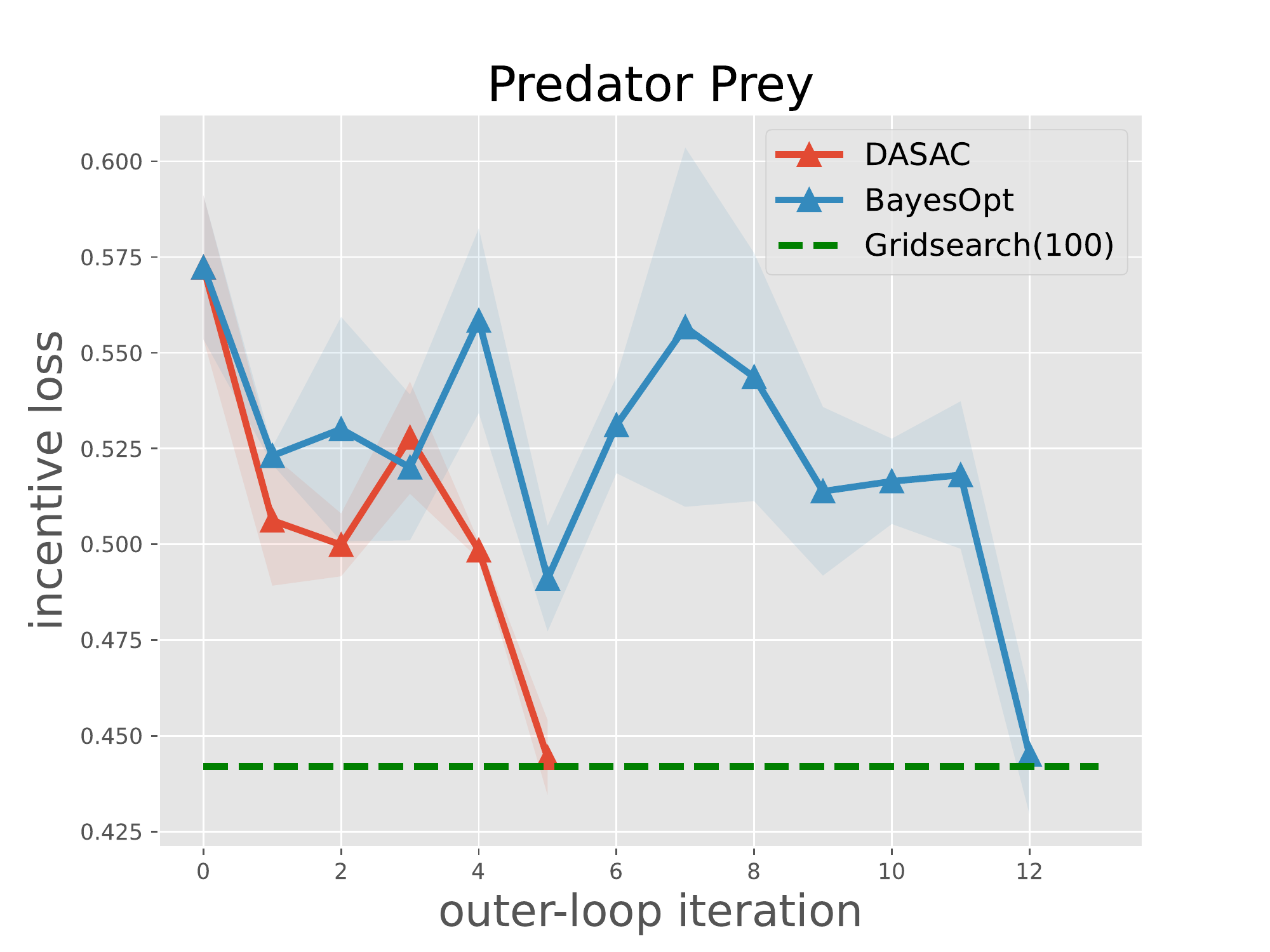}
\vspace{-3mm}
\caption{Arbitrating objective loss $f_*$ for DASAC, BayesOpt, and GridSearch(100) in RWS (top) and PP (bottom). The green dotted line indicates the best objective score of GridSearch(100). Our DASAC can outperform two zeroth-order baselines, GridSearch(100) and BayesOpt, by a large margin, which shows that our first-order framework requires fewer evaluations for the upper level to obtain a desirable NE policy, leading to a higher sample efficiency. }
\vspace{-0.25in}
\label{fig:f}
\end{figure}
Furthermore, we will present our DASAC tends to reach a desirable NE via analyzing the NE behavior at the initial incentive $\theta_0$ and the optimal incentive $\theta^*$ (marked in Figure \ref{fig:f}).
Figure \ref{fig:result} illustrates the trace of players when taking the NE policies in RWS and PP environments with initial incentive $\theta_0$ and the optimal incentive $\theta^*$ respectively. 
It demonstrates that the optimal $\theta^*$ navigates the system to a better NE that players are more exploration-minded. 
To be more specific, with the initial incentive parameter $\theta_0$, although players start to take account of the exploration rate and think outside the habitual behaviors in an unincentivized system, e.g., standstill but not collecting resources after a few steps in RWS or stalemate in a relative diagonal position in PP, they still tend to maintain the habitual behavior after a few explorations. As the training for the incentive parameter goes on, the habitual behaviors in an unincentivized system are gradually abandoned and they develop new patterns of behavior that lead to a higher exploration rate and ensure their own goals at the same time. For example, in PP, the predator still tries their best to catch the prey, and the prey tries to escape as soon as possible. The difference is that the prey tends to circle around, exploring the water source to replenish its energy instead of staying in one place and avoiding risky moves.

\section{Conclusion}\label{sec:con}
Our work initiates a provably differentiable framework in context with MARL to solve a bi-level arbitrating problem. We provide the convergence proof and empirically validate the effectiveness of the DASAC on arbitrating in two Markov games.
Our work can be extended to multiple-player and general-sum settings with proper NE solvers whose convergence properties have been empirically shown. 
Therefore, we will empirically test the performance of the performance of the DA framework on more general cases in our future work.

\bibliographystyle{ACM-Reference-Format} 
\bibliography{sample}


\onecolumn
\newpage
\appendix
\hrule height 4pt
\vskip 0.15in
\vskip -\parskip
\begin{center}
{\LARGE\bf Differentiable Arbitrating in Zero-sum Markov Games \par} 
\end{center}
\vskip 0.25in
\vskip -\parskip
\hrule height 1pt

\DoToC
\section{Proof for Lemma \ref{lma:V2}}\label{supp:V2}
In this section, we provide the formulation of gradients $\nabla_{\theta}V^{(i)}_{\pi_{\phi}}(s;\theta), \nabla_{\phi^i} V^{(i)}_{\pi_{\phi}}(s;\theta), \nabla^2_{\theta\phi^i}V^{(i)}_{\pi_{\phi}}(s;\theta)$ and $ \nabla^2_{\phi\phi^i}V^{(i)}_{\pi_{\phi}}(s;\theta)$. We first introduce the derivation of gradients for $V^i_{\pi_{\phi}}(s;\theta)$ in the following Lemma \ref{lma:V} and then move into the proof for Lemma \ref{lma:V2}.
\begin{lemma}\label{lma:V}
  For incentivized Markov game $\mathcal{G}_{\theta}=(\mathcal{N},\mathcal{S},\{\mathcal{A}^i\}_{i\in\{1,2\}},\mathcal{P},\{r^i(\cdot;\theta)\}_{i\in\{1,2\}},\gamma)$, let $R^i(\tau;\theta)$ be the total reward for an sample trajectory $\tau$ on following $\pi_{\phi}$ for $T$ steps, starting from initial state $s_0=s$. We have
  \begin{align*}
      \nabla_{\theta}V^i_{\pi_{\phi}}(s;\theta)&=\mathbb{E}_{\tau\sim D^{\pi_\phi}}\left[\nabla_{\theta}R^i(\tau;\theta)|s_0=s\right],\\
      \nabla_{\phi^i} V^i_{\pi_{\phi^i}}(s;\theta)&=\mathbb{E}_{\tau\sim D^{\pi_\phi}}\left[R^i(\tau;\theta)\sum_{t=0}^{T-1}\nabla_{\phi^i}\log\pi^i_{\phi^i}(a_t|s_t)\Big|s_0=s\right],\\
      \nabla^2_{\theta\phi^i}V^i_{\pi_{\phi}}(s;\theta)&=\mathbb{E}_{\tau\sim D^{\pi_\phi}}\left[\left(\sum_{t=0}^{T-1}\nabla_{\phi^i}\log\pi^i_{\phi^i}(a_t|s_t)\right)\nabla_{\theta}R^i(\tau;\theta)^\intercal\Big|s_0=s\right],\\
      \nabla^2_{\phi\phi^i}V^i_{\pi_{\phi}}(s;\theta)&=\mathbb{E}_{\tau\sim D^{\pi_\phi}}\left[R^i(\tau;\theta)\left((\sum_{t=0}^{T-1}\nabla_{\phi^i}\log\pi^i_{\phi^i}(a_t|s_t))(\sum_{t=0}^{T-1}\nabla_{\phi}\log\pi_{\phi}(a_t|s_t))^\intercal\right.\right.\notag\\
      &\left.\left.\quad\quad\quad\quad\quad\quad+\sum_{t=0}^{T-1}\nabla^2_{\phi\phi^i}\log\pi_{\phi}(a_t|s_t)\right)\Big|s_0=s\right],
  \end{align*}
  where
  \begin{align*}
      \nabla_{\phi}\log\pi_{\phi}(a_t|s_t)&=(\nabla_{\phi^1}\log\pi^1_{\phi^1}(a_t|s_t),\nabla_{\phi^2}\log\pi^2_{\phi^2}(a_t|s_t))\notag\\
      \nabla^2_{\phi\phi^i}\log\pi_{\phi}(a_t|s_t)&=(\nabla^2_{\phi^1\phi^i}\log\pi^1_{\phi^1}(a_t|s_t),\nabla^2_{\phi^2\phi^i}\log\pi^2_{\phi^2}(a_t|s_t)).\notag
  \end{align*}   
\end{lemma}
\begin{proof}[Proof for Lemma \ref{lma:V}.]
First, we find out $\nabla_{\phi^i}\log{D^{\pi_\phi}(\tau|s_0=s)}$:
  \begin{align*}
      \nabla_{\phi^i}\log{D^{\pi_\phi}(\tau|s_0=s)}
      &=\nabla_{\phi^i}\log\left[\prod_{t=0}^{T-1}P(s_{t+1}|s_t,a_t) \left(\prod_{i=1}^2\pi^i_{\phi^i}(a^i_t|s_t)\right)\right]\\
      &=\nabla_{\phi^i}\!\!\left[\sum_{t=0}^{T-1}\log{P(s_{t+1}|s_t,a_t)}\!+\!\sum_{t=0}^{T-1}\log{\pi^i_{\phi^i}(a_t|s_t)}\!+\!\sum_{t=0}^{T-1}\log{\pi^{-i}_{\phi^{-i}}(a_t|s_t)}\right]=\sum_{t=0}^{T-1}\nabla_{\phi^i}\log{\pi^i_{\phi^i}(a_t|s_t)}.\\
  \end{align*}
  Similarly, we have
  $\nabla_{\phi}\log{D^{\pi_\phi}(\tau|s_0=s)}\!=\!\sum_{t=0}^{T-1}\nabla_{\phi}\log{\pi_{\phi}(a_t|s_t)},$
  where $\nabla_{\phi}\log\pi_{\phi}(a_t|s_t)\!=\!(\nabla_{\phi^1}\log\pi^1_{\phi^1}(a_t|s_t),$ $\nabla_{\phi^2}\log\pi^2_{\phi^2}(a_t|s_t))$. Moreover, recall that
  $
      \nabla_{\phi\phi^i}\log\pi_{\phi}(a_t|s_t)=(\nabla_{\phi^1\phi^i}\log\pi^1_{\phi^1}(a_t|s_t),\nabla_{\phi^2\phi^i}\log\pi^2_{\phi^2}(a_t|s_t)),
 $
  we could derive $\nabla_{\theta}V^i_{\pi_{\phi}}(s;\theta)$, $\nabla^2_{\phi^i\theta}V^i_{\pi_\phi}(s;\theta)$,$\nabla_{\phi^i}V^i_{\pi_\phi}(s;\theta)$ and $\nabla^2_{\phi\phi^i}V^i_{\pi_\phi}(s;\theta)$ as follows,
  \begin{align*}
      \nabla_{\theta}V^i_{\pi_{\phi}}(s;\theta)&=\nabla_{\theta}\mathbb{E}_{\tau\sim D^{\pi_\phi}}\left[R^i(\tau;\theta)|s_0=s\right]=\mathbb{E}_{\tau\sim D^{\pi_\phi}}\left[\nabla_{\theta}R^i(\tau;\theta)|s_0=s\right],\\
      \nabla_{\phi^i}V^i_{\pi_\phi}(s;\theta)&=\nabla_{\phi^i}\mathbb{E}_{\tau\sim D^{\pi_\phi}}\left[R^i(\tau;\theta)|s_0=s\right]
      =\nabla_{\phi^i}\left[\sum_{\tau\sim D^{\pi_\phi}}D^{\pi_\phi}(\tau|s_0=s)R^i(\tau;\theta)\right]\\
      &=\sum_{\tau\sim D^{\pi_\phi}}D^{\pi_\phi}(\tau|s_0=s)\nabla_{\phi^i}\log{D^{\pi_\phi}(\tau|s_0=s)}R^i(\tau;\theta)\\
      &=\mathbb{E}_{\tau\sim
      D^{\pi_\phi}}\left[\nabla_{\phi^i}\log{D^{\pi_\phi}(\tau|s_0=s)}R^i(\tau;\theta)|s_0=s\right]\\
      &=\mathbb{E}_{\tau\sim
      D^{\pi_\phi}}\left[R^i(\tau;\theta)\sum_{t=0}^{T-1}\nabla_{\phi^i}\log{\pi^i_{\phi^i}(a_t|s_t)}|s_0=s\right],
 \end{align*}
where the third equation follows as $\nabla_{\phi^i}D^{\pi_\phi}=D^{\pi_\phi}\frac{\nabla_{\phi^i}D^{\pi_\phi}}{D^{\pi_\phi}}=D^{\pi_\phi}\nabla_{\phi^i}\log{D^{\pi_\phi}}$ and interchanging the gradient and summation.    
\begin{align*}
    \nabla^2_{\theta\phi^i}V^i_{\pi_{\phi}}(s;\theta)&=\nabla_{\theta}\mathbb{E}_{\tau\sim
      D^{\pi_\phi}}\left[R^i(\tau;\theta)\sum_{t=0}^{T-1}\nabla_{\phi^i}\log{\pi^i_{\phi^i}(a_t|s_t)}|s_0=s\right]\\
      &=\mathbb{E}_{\tau\sim
      D^{\pi_\phi}}\left[\left(\sum_{t=0}^{T-1}\nabla_{\phi^i}\log{\pi^i_{\phi^i}(a_t|s_t)}\right)\nabla_{\theta}R^i(\tau;\theta)^\intercal|s_0=s\right],\\
      \nabla^2_{\phi\phi^i}V^i_{\pi_{\phi}}(s;\theta)&=\nabla_{\phi\phi^i}\left[\sum_{\tau\sim D^{\pi_{\phi}}}D^{\pi_{\phi}}(\tau|s_0=s)R^i(\tau;\theta)\right]\\
      &=\nabla_{\phi}\left[\sum_{\tau\sim D^{\pi_\phi}}D^{\pi_\phi}(\tau|s_0=s)\nabla_{\phi^i}\log{D^{\pi_{\phi}}(\tau|s_0=s)}R^i(\tau;\theta)\right]\\
      &=\nabla_{\phi}\left[\sum_{\tau\sim D^{\pi_\phi}}D^{\pi_\phi}(\tau|s_0=s)\left(\sum_{t=0}^{T-1}\nabla_{\phi^i}\log{\pi^i_{\phi^i}(a_t|s_t)}\right)R^i(\tau;\theta)\right]\\
    &=\sum_{\tau\sim D^{\pi_\phi}}R^i(\tau;\theta)\left[\left(\sum_{t=0}^{T-1}\nabla_{\phi^i}\log{\pi^i_{\phi^i}(a_t|s_t)}\right)\nabla_{\phi}D^{\pi_{\phi}}(\tau|s_0=s)^\intercal\right.\\
    &\quad\quad\quad\quad\quad\quad\quad\quad\quad\quad\left.+D^{\pi_{\phi}}(\tau|s_0=s)\left(\sum_{t=0}^{T-1}\nabla_{\phi\phi^i}\log{\pi_{\phi}(a_t|s_t)}\right)\right]\\
    &=\!\!\sum_{\tau\sim D^{\pi_\phi}}\!\!R^i(\tau;\theta)D^{\pi_{\phi}}(\tau|s_0\!=\!s)\!\left[\!\left(\!\sum_{t=0}^{T-1}\nabla_{\phi^i}\log{\pi^i_{\phi^i}(a_t|s_t)}\right)\!\nabla_{\phi}\!\log D^{\pi_{\phi}}(\tau|s_0\!=\!s)^\intercal\right.\\
    &\quad\quad\quad\quad\quad\quad\quad\quad\quad\quad\quad\quad\quad\quad\quad\quad\left.+\sum_{t=0}^{T-1}\nabla_{\phi\phi^i}\log{\pi_{\phi}(a_t|s_t)}\right]\\
    &=\mathbb{E}_{\tau\sim D^{\pi_\phi}}\left[R^i(\tau;\theta)\left((\sum_{t=0}^{T-1}\nabla_{\phi^i}\log\pi^i_{\phi^i}(a_t|s_t))(\sum_{t=0}^{T-1}\nabla_{\phi}\log\pi_{\phi}(a_t|s_t))^\intercal\right.\right.\notag\\
    &\quad\quad\quad\quad\quad\quad\quad\quad\quad\quad\quad\quad\quad\quad\quad\quad\left.\left.+\sum_{t=0}^{T-1}\nabla^2_{\phi\phi^i}\log\pi_{\phi}(a_t|s_t)\right)\Big|s_0=s\right].
\end{align*}
\end{proof}
In the remaining of this section, we complete the proof of Lemma \ref{lma:V2} in main body.
\begin{proof}[Proof for Lemma \ref{lma:V2}.]Recall the definition of $V^{(i)}_{\pi_{\phi}}(s;\theta)$ in (\ref{regularized V}) and  $r^{i}_{\pi}(s,a^i,a^{-i};{\theta})$ in (\ref{def:regularized reward}), we have,
\begin{align}
    V^{(i)}_{\pi_{\phi}}(s;\theta)&=\mathbb{E}_{\tau\sim D^{\pi_{\phi}}}\left[\sum_{t=0}^{T-1}\gamma^t{\cdot}r^{(i)}_{\pi}(s_t,a^i_t,a^{-i}_t;{\theta})|s_0=s\right]\notag\\
    &=\mathbb{E}_{\tau\sim D^{\pi_{\phi}}}\left[\sum_{t=0}^{T-1}\gamma^t{\cdot}r^{i}_{\pi}(s_t,a^i_t,a^{-i}_t;{\theta})-\lambda\sum_{t=0}^{T-1}\gamma^t\left(\log\pi_{\phi^i}^i(a_t^i|s_t)-\log\pi_{\phi^{-i}}^{-i}(a_t^{-i}|s_t)\right)|s_0=s\right]\notag\\
    \vspace{-0.05in}
    &=V^{i}_{\pi_{\phi}}(s;\theta)-\lambda\mathbb{E}_{\tau\sim D^{\pi_{\phi}}}\left[\sum_{t=0}^{T-1}\gamma^t\left(\log\pi_{\phi^i}^i(a_t^i|s_t)-\log\pi_{\phi^{-i}}^{-i}(a_t^{-i}|s_t)\right)|s_0=s\right].\label{VU}
    \vspace{-0.25in}
\end{align}
For easy notation, we define
\begin{align*}
    U_{\pi_{\phi}}^i(s)=\mathbb{E}_{\tau\sim D^{\pi_{\phi}}}\left[\sum_{t=0}^{T-1}\gamma^t\left(\log\pi_{\phi^i}^i(a_t^i|s_t)-\log\pi_{\phi^{-i}}^{-i}(a_t^{-i}|s_t)\right)|s_0=s\right].
\end{align*}
Then we have,
\begin{align}
    \nabla_{\theta}U_{\pi_{\phi}}^i(s)=&0;\label{U_theta}\\
    \nabla_{\phi^i}U_{\pi_{\phi}}^i(s)=&\nabla_{\phi^i}\left[\sum_{\tau}D^{\pi_{\phi}}(\tau|s_0\!=\!s)\sum_{t=0}^{T-1}\gamma^t\left(\log\pi_{\phi^i}^i(a_t^i|s_t)-\log\pi_{\phi^{-i}}^{-i}(a_t^{-i}|s_t)\right)\right]\notag\\
    =&\sum_{\tau}D^{\pi_{\phi}}(\tau|s_0\!=\!s)\left[\nabla_{\phi_i}\log D^{\pi_{\phi}}(\tau|s_0\!=\!s)\sum_{t=0}^{T-1}\gamma^t\left(\log\pi_{\phi^i}^i(a_t^i|s_t)-\log\pi_{\phi^{-i}}^{-i}(a_t^{-i}|s_t)\right)\right.\notag\\
    &\quad\quad\quad\left.+\sum_{t=0}^{T-1}\gamma^t\nabla_{\phi^i}\log\pi_{\phi^i}^i(a_t^i|s_t)\right]\notag\\
    =&\mathbb{E}_{\tau\sim D^{\pi_{\phi}}}\left[\left(\sum_{t=0}^{T-1}\gamma^t\left(\log\pi_{\phi^i}^i(a_t^i|s_t)-\log\pi_{\phi^{-i}}^{-i}(a_t^{-i}|s_t)\right)\right)\left(\sum_{t=0}^{T-1}\nabla_{\phi^i}\log\pi_{\phi^i}^i(a_t^i|s_t)\right)\right.\notag\\
    &\quad\quad\quad\left.+\sum_{t=0}^{T-1}\gamma^t\nabla_{\phi^i}\log\pi_{\phi^i}^i(a_t^i|s_t)\right];\label{U_phi_i}\\
    \nabla_{\theta\phi^i}U_{\pi_{\phi}}^i(s)=&0;\label{U_theta_phi_i}\\
    \nabla_{\phi\phi^i}U_{\pi_{\phi}}^i(s)=&\overbrace{\nabla_{\phi}\left[\sum_{\tau}D^{\pi_\phi}(\tau|s_0=s)\left(\sum_{t=0}^{T-1}\gamma^t\left(\log\pi_{\phi^i}^i(a_t^i|s_t)-\log\pi_{\phi^{-i}}^{-i}(a_t^{-i}|s_t)\right)\right)\left(\sum_{t=0}^{T-1}\nabla_{\phi^i}\log\pi_{\phi^i}^i(a_t^i|s_t)\right)\right]}^{\displaystyle{:=(\RNum{1})}}\notag\\
    &+\underbrace{\nabla_{\phi}\left[\sum_{\tau}D^{\pi_\phi}(\tau|s_0=s)\sum_{t=0}^{T-1}\gamma^t\nabla_{\phi^i}\log\pi_{\phi^i}^i(a_t^i|s_t)\right]}_{\displaystyle{:=(\RNum{2})}}\notag.
\end{align}
Next we are going to compute (\RNum{1}) and (\RNum{2}).
\begin{align*}
    (\RNum{1})=&\sum_{\tau}\!D^{\pi_\phi}\!(\tau|s_0\!=\!s)\!\left(\sum_{t=0}^{T-1}\!\gamma^t\!\left(\log\pi_{\phi^i}^i(a_t^i|s_t)-\log\pi_{\phi^{-i}}^{-i}(a_t^{-i}|s_t)\right)\!\right)\!\!\left(\!\sum_{t=0}^{T-1}\nabla_{\phi^i}\!\log\!\pi_{\phi^i}^i(a_t^i|s_t\!)\!\right)\!\nabla_{\phi}\!\log D^{\pi_\phi}(\tau|s_0=s)^\intercal\\
    &+\sum_{\tau}D^{\pi_\phi}(\tau|s_0=s)\nabla_{\phi}\left[\left(\sum_{t=0}^{T-1}\gamma^t\left(\log\pi_{\phi^i}^i(a_t^i|s_t)-\log\pi_{\phi^{-i}}^{-i}(a_t^{-i}|s_t)\right)\right)\left(\sum_{t=0}^{T-1}\nabla_{\phi^i}\log\pi_{\phi^i}^i(a_t^i|s_t)\right)\right]\\
\end{align*}
\begin{align*}
    (\RNum{1})=&\!\sum_{\tau}D^{\pi_\phi}\!(\tau|s_0=s\!)\!\left(\!\sum_{t=0}^{T-1}\!\gamma^t\left(\log\pi_{\phi^i}^i(a_t^i|s_t)-\log\pi_{\phi^{-i}}^{-i}(a_t^{-i}|s_t)\right)\!\right)\!\!\left(\sum_{t=0}^{T-1}\!\nabla_{\phi^i}\!\log\pi_{\phi^i}^i(a_t^i|s_t\!)\!\right)\!\!\left(\sum_{t=0}^{T-1}\nabla_{\phi}\log \pi_{\phi}(a_t|s_t\!)\!\right)^\intercal\\
    &+\sum_{\tau}D^{\pi_\phi}(\tau|s_0=s)\left[\left(\sum_{t=0}^{T-1}\nabla_{\phi^i}\log\pi_{\phi^i}^i(a_t^i|s_t)\right)\left(\sum_{t=0}^{T-1}\gamma^t\left(\nabla_{\phi}\log\pi_{\phi^i}^i(a_t^i|s_t)-\nabla_{\phi}\log\pi_{\phi^{-i}}^{-i}(a_t^{-i}|s_t)\right)\right)^\intercal\right.\\
    &\quad\quad\quad\quad\quad\quad\quad\quad\quad\left.+\left(\sum_{t=0}^{T-1}\gamma^t\left(\log\pi_{\phi^i}^i(a_t^i|s_t)-\log\pi_{\phi^{-i}}^{-i}(a_t^{-i}|s_t)\right)\right)\left(\sum_{t=0}^{T-1}\nabla^2_{\phi\phi^i}\log\pi_{\phi}(a_t^i|s_t)\right)\right]\\
    =&\mathbb{E}_{\tau\sim D^{\pi_\phi}}\!\!\!\left[\left(\sum_{t=0}^{T-1}\gamma^t\left(\log\pi_{\phi^i}^i(a_t^i|s_t)-\log\pi_{\phi^{-i}}^{-i}(a_t^{-i}|s_t)\right)\right)\left(\sum_{t=0}^{T-1}\nabla_{\phi^i}\log\pi_{\phi^i}^i(a_t^i|s_t)\right)\left(\sum_{t=0}^{T-1}\nabla_{\phi}\log \pi_{\phi}(a_t|s_t)\right)^\intercal\right.\\
    &\quad\quad\quad\quad\left.+\left(\sum_{t=0}^{T-1}\!\nabla_{\phi^i}\!\log\pi_{\phi^i}^i(a_t^i|s_t)\right)\left(\!\sum_{t=0}^{T-1}\!\gamma^t\left(\nabla_{\phi}\!\log\pi_{\phi^i}^i(a_t^i|s_t)-\nabla_{\phi}\!\log\pi_{\phi^{-i}}^{-i}(a_t^{-i}|s_t)\right)\right)^\intercal\right.\\
    &\quad\quad\quad\quad\left.+\left(\!\sum_{t=0}^{T-1}\!\gamma^t\left(\log\pi_{\phi^i}^i(a_t^i|s_t)-\log\pi_{\phi^{-i}}^{-i}(a_t^{-i}|s_t)\right)\right)\left(\!\sum_{t=0}^{T-1}\!\nabla^2_{\phi\phi^i}\!\log\pi_{\phi}(a_t^i|s_t)\right)\right];
\end{align*}
\begin{align*}
    (\RNum{2})=&\sum_{\tau}D^{\pi_\phi}(\tau|s_0=s)\left[\left(\sum_{t=0}^{T-1}\gamma^t\nabla_{\phi^i}\log \pi^i_{\phi^i}(a_t|s_t)\right)\left(\sum_{t=0}^{T-1}\nabla_{\phi}\log \pi_{\phi}(a_t|s_t)\right)^\intercal\right.\\
    &\left.\quad\quad\quad+\left(\sum_{t=0}^{T-1}\!\gamma^t\nabla^2_{\phi\phi^i}\log \pi_{\phi}(a_t|s_t)\!\right)\!\right]\\
    =&\mathbb{E}_{\tau\sim D^{\pi_\phi}}\!\!\left[\!\left(\!\sum_{t=0}^{T-1}\!\gamma^t\!\nabla_{\phi^i}\!\log \pi^i_{\phi^i}(a_t|s_t\!)\!\right)\!\!\left(\sum_{t=0}^{T-1}\!\nabla_{\phi}\!\log \pi_{\phi}(a_t|s_t\!)\!\right)^\intercal\!\!\!\!+\!\!\left(\sum_{t=0}^{T-1}\gamma^t\nabla
    ^2_{\phi\phi^i}\!\log \pi_{\phi}(a_t|s_t\!)\!\right)\!\right].
\end{align*}
Therefore,
\begin{align}\label{U_phi_phi_i}
    &\nabla_{\phi\phi^i}U_{\pi_\phi}^{i}(s)\notag=\mathbb{E}_{\tau\sim D^{\pi_\phi}}\!\!\!\left[\left(\sum_{t=0}^{T-1}\gamma^t\left(\log\pi_{\phi^i}^i(a_t^i|s_t)-\log\pi_{\phi^{-i}}^{-i}(a_t^{-i}|s_t)\right)\right)\left(\sum_{t=0}^{T-1}\nabla_{\phi^i}\log\pi_{\phi^i}^i(a_t^i|s_t)\right)\left(\sum_{t=0}^{T-1}\nabla_{\phi}\log \pi_{\phi}(a_t|s_t)\right)^\intercal\right.\\
    &\quad\quad\quad\quad\left.+\left(\sum_{t=0}^{T-1}\!\nabla_{\phi^i}\!\log\pi_{\phi^i}^i(a_t^i|s_t)\right)\left(\!\sum_{t=0}^{T-1}\!\gamma^t\left(\nabla_{\phi}\!\log\pi_{\phi^i}^i(a_t^i|s_t)-\nabla_{\phi}\!\log\pi_{\phi^{-i}}^{-i}(a_t^{-i}|s_t)\right)\right)^\intercal\right.\notag\\
    &\quad\quad\quad\quad\left.+\left(\!\sum_{t=0}^{T-1}\!\gamma^t\left(\log\pi_{\phi^i}^i(a_t^i|s_t)-\log\pi_{\phi^{-i}}^{-i}(a_t^{-i}|s_t)\right)\right)\left(\!\sum_{t=0}^{T-1}\!\nabla^2_{\phi\phi^i}\!\log\pi_{\phi}(a_t^i|s_t)\right)\right.\notag\\
    &\quad\quad\quad\quad\left.+\left(\!\sum_{t=0}^{T-1}\!\gamma^t\nabla_{\phi^i}\log \pi^i_{\phi^i}(a_t|s_t)\right)\left(\sum_{t=0}^{T-1}\nabla_{\phi}\log \pi_{\phi}(a_t|s_t\!)\!\right)^\intercal+\left(\sum_{t=0}^{T-1}\gamma^t\nabla
    ^2_{\phi\phi^i}\!\log \pi_{\phi}(a_t|s_t)\right)\right]    
\end{align}
Since (\ref{VU}) can be written as
\begin{align*}
    V^{(i)}_{\pi_{\phi}}(s;\theta)=V^{i}_{\pi_{\phi}}(s;\theta)-\lambda U_{\pi_\phi}^{i}(s),
\end{align*}
then together with (\ref{U_theta}-\ref{U_phi_phi_i}), we have
\begin{align*}
    &\nabla_{\theta}V^{(i)}_{\pi_{\phi}}(s;\theta)=\nabla_{\theta}V^{i}_{\pi_{\phi}}(s;\theta);\\
    &\nabla_{\phi^i} V^{(i)}_{\pi_{\phi}}(s;\theta)=\nabla_{\phi^i} V^{i}_{\pi_{\phi}}(s;\theta)-\lambda\mathbb{E}_{\tau\sim D^{\pi_{\phi}}}\left[\left(\sum_{t=0}^{T-1}\gamma^t\left(\log\pi_{\phi^i}^i(a_t^i|s_t)-\log\pi_{\phi^{-i}}^{-i}(a_t^{-i}|s_t)\right)\right)\left(\sum_{t=0}^{T-1}\nabla_{\phi^i}\log\pi_{\phi^i}^i(a_t^i|s_t)\right)+
    \sum_{t=0}^{T-1}\gamma^t\nabla_{\phi^i}\log\pi_{\phi^i}^i(a_t^i|s_t)\right];\\
    &\nabla^2_{\theta\phi^i}V^{(i)}_{\pi_{\phi}}(s;\theta)=\nabla^2_{\theta\phi^i}V^{i}_{\pi_{\phi}}(s;\theta);\\
    &\nabla^2_{\phi\phi^i}V^{(i)}_{\pi_{\phi}}(s;\theta)=\nabla^2_{\phi\phi^i}V^{i}_{\pi_{\phi}}(s;\theta)\!-\!\lambda\mathbb{E}_{\tau\sim D^{\pi_\phi}}\left[\left(\sum_{t=0}^{T-1}\gamma^t\left(\log\pi_{\phi^i}^i(a_t^i|s_t)-\log\pi_{\phi^{-i}}^{-i}(a_t^{-i}|s_t)\right)\right)\left(\sum_{t=0}^{T-1}\nabla_{\phi^i}\log\pi_{\phi^i}^i(a_t^i|s_t)\right)\left(\sum_{t=0}^{T-1}\nabla_{\phi}\log \pi_{\phi}(a_t|s_t)\right)^\intercal\right.\\
    &\quad\quad\quad\quad\quad\quad\quad\quad\quad\quad\quad\quad\quad\quad\quad\quad\quad\quad\quad\left.+\left(\sum_{t=0}^{T-1}\!\nabla_{\phi^i}\!\log\pi_{\phi^i}^i(a_t^i|s_t)\right)\left(\!\sum_{t=0}^{T-1}\!\gamma^t\left(\nabla_{\phi}\!\log\pi_{\phi^i}^i(a_t^i|s_t)-\nabla_{\phi}\!\log\pi_{\phi^{-i}}^{-i}(a_t^{-i}|s_t)\right)\right)^\intercal\right.\notag\\
    &\quad\quad\quad\quad\quad\quad\quad\quad\quad\quad\quad\quad\quad\quad\quad\quad\quad\quad\quad\left.+\left(\!\sum_{t=0}^{T-1}\!\gamma^t\left(\log\pi_{\phi^i}^i(a_t^i|s_t)-\log\pi_{\phi^{-i}}^{-i}(a_t^{-i}|s_t)\right)\right)\left(\!\sum_{t=0}^{T-1}\!\nabla^2_{\phi\phi^i}\!\log\pi_{\phi}(a_t^i|s_t)\right)\right.\notag\\
    &\quad\quad\quad\quad\quad\quad\quad\quad\quad\quad\quad\quad\quad\quad\quad\quad\quad\quad\quad\left.+\left(\!\sum_{t=0}^{T-1}\!\gamma^t\nabla_{\phi^i}\log \pi^i_{\phi^i}(a_t|s_t)\right)\left(\sum_{t=0}^{T-1}\nabla_{\phi}\log \pi_{\phi}(a_t|s_t\!)\!\right)^\intercal+\left(\sum_{t=0}^{T-1}\gamma^t\nabla
    ^2_{\phi\phi^i}\!\log \pi_{\phi}(a_t|s_t)\right)\right]
\end{align*}
\end{proof}


\section{Proof for Lemma \ref{lma:phitheta}}\label{sub:proof_grad}

In this section, we prove Lemma \ref{lma:phitheta} based on Lemma \ref{lma:V2}.
\begin{proof}[Proof for Lemma \ref{lma:phitheta}.]
Since $\phi^{(*)}(\theta)$ is the policy parameter of Nash equilibrium $\pi_{\phi^{(*)}(\theta)}$ for the regularized Markov game with the incentivized reward $\mathcal{G}_{\theta}'$, $u_\theta(\phi^{(*)}(\theta))=(\mathbb{E}_{\nu^{*}}\nabla_{\phi^1}V^{(1)}_{\pi_{\phi}}(\theta),\mathbb{E}_{\nu^{*}}\nabla_{\phi^2}V^{(2)}_{\pi_{\phi}}(\theta))\big|_{\phi=\phi^{(*)}(\theta)}=0$. Then, by differentiating the equality with respect to $\theta$ on both sides, for any $i\in\{1,2\}$, we have
$$\mathbb{E}_{\nu^{*}}\nabla_{\theta\phi^i}V^{(i)}_{\pi_{\phi}}(\theta)+\mathbb{E}_{\nu^{*}}\nabla_{\phi\phi^i}V^{(i)}_{\pi_{\phi}}(\theta)[\nabla_{\theta}\phi^{(*)}(\theta)]=0.$$
Thus, define
\begin{align*}
      \nabla_{\theta}u_{\theta}(\phi)=\left[
      \begin{matrix}
        \mathbb{E}_{\nu^{*}}\nabla_{\theta\phi^1}V^{(1)}_{\pi_{\phi}}(\theta)\\
        \mathbb{E}_{\nu^{*}}\nabla_{\theta\phi^2}V^{(2)}_{\pi_{\phi}}(\theta)\\
      \end{matrix}
      \right]\in\mathbb{R}^{d\times m},\\
      \nabla_{\phi}u_{\theta}(\phi)=\left[
      \begin{matrix}
        \mathbb{E}_{\nu^{*}}\nabla_{\phi\phi^1}V^{(1)}_{\pi_{\phi}}(\theta)\\
        \mathbb{E}_{\nu^{*}}\nabla_{\phi\phi^2}V^{(2)}_{\pi_{\phi}}(\theta)\\
      \end{matrix}
      \right]\in\mathbb{R}^{d\times d},
  \end{align*}
  where $\nu^{*}$ is $\nu_{\pi_{\phi^{(*)}(\theta)}}(s)$. We have,
  $$\nabla_{\theta}\phi^{(*)}(\theta)=-\left[\nabla_{\phi}u_{\theta}(\phi)\right]^{-1}\nabla_{\theta}u_{\theta}(\phi)\big|_{\phi=\phi^*(\theta)}.$$
\end{proof}
\section{Extensions of DA framework}\label{supp:extension}
Note that the DA framework itself is general given a proper lower-level NE solver. In this section, we detail the extension of DA framework to two general cases: i) \textit{N-players} and \textit{no restriction to reward}; ii) \textit{partial observations}. Since the primary part in DA framework is the derivation of upper-level loss gradient at NE $\nabla f_*(\theta)$, in the following we present the calculation of $\nabla f_*(\theta)$ based on the policy gradient information from the lower level in extended settings.

\textbf{Extension to N-player and no reward restriction.} We first define the incentivized Markov game with $N$ players as  $\mathcal{G}_{\theta} = (\mathcal{S}, \{\mathcal{A}^i,\\ r^{i}(\cdot;\theta)\}_{i\in\{1,\cdots,N\}}, \lambda ,\mathcal{P}, \gamma)$. Then the total discounted reward $R^i_\pi(\tau;\theta)$ and value function $V^i_{\pi}(s;\theta)$ for all $i\in\{1,...,N\}$ are defined as
\begin{align}
    R_\pi^i(\tau;\theta) &= \sum_{t=0}^{T-1}\gamma^t r^i(s_t,a_t^i,a_{t}^{-i};\theta),\\
    V_\pi^i(s;\theta) &= \mathbb{E}_{\pi}\left[\sum_{t=0}^{T-1}\gamma^t{\cdot}r^{i}_{\pi}(s_t,a^i_t,a^{-i}_t;{\theta})|s_0=s\right]\notag\\
    &=\mathbb{E}_{\tau\sim D^{\pi}}\left[R^{i}_{\pi}(\tau;{\theta})\big|s_0=s\right].     
\end{align}

 Then define the entropy-regularized $\mathcal{G}'_{\theta} = (\mathcal{S}, \{\mathcal{A}^i, r^{(i)}_\pi(\cdot;\theta)\}_{i\in\{1,\cdots,N\}}, \lambda ,\mathcal{P}, \gamma)$, where $r^{(i)}(\cdot;\theta): \mathcal{S} \times \mathcal{A}^i \times \mathcal{A}^{-i} \rightarrow \mathcal{R}$ is defined as:
 \begin{align*}
     r^{(i)}_\pi(\cdot;\theta) = r^i(\cdot;\theta)-\lambda \log(\pi^{i}(a^{i}|s))+\frac{\lambda}{N-1}\sum_{j=1,...,n\atop j\neq i}\log(\pi^{j}(a^{j}|s)).
 \end{align*}
 Then the entropy-regularized value function is
 \begin{align}\label{eq:V_nplayer}
    R_\pi^{(i)}(\tau;\theta) &= \sum_{t=0}^{T-1}\gamma^t r^{(i)}_\pi(s_t,a_t^i,a_{t}^{-i};\theta),\\
    V^{(i)}_{\pi}(s;\theta)&=\mathbb{E}_{\pi}\left[\sum_{t=0}^{T-1}\gamma^t{\cdot}r^{(i)}_{\pi}(s_t,a^i_t,a^{-i}_t;{\theta})|s_0=s\right]\notag\\
    &=\mathbb{E}_{\tau\sim D^{\pi}}\left[R^{(i)}_{\pi}(\tau;{\theta})\big|s_0=s\right].     
 \end{align}
 
 Other definitions remain same as defined in Section \ref{subsec:irg} except for the zero-sum restriction to $r^i(\cdot;\theta)$. We are going to paramterized the joint policy $\pi$ with $\phi$ for computation simplicity. Therefore, $\pi_{\phi}\!=\!(\pi^1_{\phi^1}\!,\cdots,\pi^N_{\phi^N})$, where $\phi=(\phi^1,\cdots,\phi^N) \in \mathbb{R}^d, d=\sum_{i=1}^N d_i$. 
 In N-player Markov game, we denote 
 \begin{align*}
      \nabla_{\phi}\log\pi_{\phi}(a_t|s_t)&=(\nabla_{\phi^1}\log\pi^1_{\phi^1}(a_t|s_t),\cdots,\nabla_{\phi^N}\log\pi^N_{\phi^N}(a_t|s_t))\notag\\
      \nabla^2_{\phi\phi^i}\log\pi_{\phi}(a_t|s_t)&=(\nabla^2_{\phi^1\phi^i}\log\pi^1_{\phi^1}(a_t|s_t),\cdots,\nabla^2_{\phi^N\phi^i}\log\pi^N_{\phi^N}(a_t|s_t)).\notag
  \end{align*} 

Lemma \ref{lma:V in Nplayer} and Lemma \ref{lma:V2 in Nplayer} extends the Lemma \ref{lma:V} and Lemma \ref{lma:V2} into the N-player general-sum game respectively. 

For extension of Lemma \ref{lma:phitheta},
 we switch the notation in (\ref{def: u*}) and (\ref{notation:gradient u}) as 
\begin{align}
\label{def: u* for N players}
    u_{\theta}(\phi;s)\!&:=(\nabla_{\phi^1}V^{(1)}_{\pi_{\phi}}(s;\theta),\cdots,\nabla_{\phi^N}V^{(N)}_{\pi_{\phi}}(s;\theta)), \quad
    u_{\theta}(\phi):=\mathbb{E}_{\nu^{(*)}}[u_{\theta}(\phi;s)],\\
    \label{notation:gradient u for N} 
      \nabla_{\theta}u_{\theta}(\phi)\!&=\!\left[
      \begin{matrix}
        \mathbb{E}_{\nu^{(*)}}\!\nabla^2_{\theta\phi^1}V^{(1)}_{\pi_{\phi}}(s;\theta)\\
        \cdots\\
        \mathbb{E}_{\nu^{(*)}}\!\nabla^2_{\theta\phi^N}V^{(N)}_{\pi_{\phi}}(s;\theta)\\
      \end{matrix}
      \right]\in\mathbb{R}^{d\times m},
      \nabla_{\phi}u_{\theta}(\phi)=\left[
      \begin{matrix}
        \mathbb{E}_{\nu^{(*)}}\nabla^2_{\phi\phi^1}V^{(1)}_{\pi_{\phi}}(s;\theta)\\
        \cdots\\
        \mathbb{E}_{\nu^{(*)}}\nabla^2_{\phi\phi^N}V^{(N)}_{\pi_{\phi}}(s;\theta)\\
      \end{matrix}
      \right]\in\mathbb{R}^{d\times d}.
\end{align}
Then (\ref{eq:grad_phi_new}) in Lemma \ref{lma:phitheta} and the gradient of upper-level loss $\nabla f_*(\theta)$ in (\ref{eq:grad_f}) remain the same.

\begin{lemma}\label{lma:V in Nplayer}
  For an incentivized Markov game $\mathcal{G}_{\theta} = (\mathcal{S}, \{\mathcal{A}^i, \Omega^i,\mathcal{O}^i, r^{i}(\cdot;\theta) \}_{i\in\{1,\cdots,N\}},\mathcal{P}, \gamma)$, let $R^i(\tau;\theta)$ be the total reward for an sample trajectory $\tau$ following $\pi_{\phi}$ for $T$ steps, starting from initial state $s_0=s$. We have
  \begin{align*}
      \nabla_{\theta}V^i_{\pi_{\phi}}(s;\theta)&=\mathbb{E}_{\tau\sim D^{\pi_\phi}}\left[\nabla_{\theta}R^i(\tau;\theta)|s_0=s\right],\\
      \nabla_{\phi^i} V^i_{\pi_{\phi^i}}(s;\theta)&=\mathbb{E}_{\tau\sim D^{\pi_\phi}}\left[R^i(\tau;\theta)\sum_{t=0}^{T-1}\nabla_{\phi^i}\log\pi^i_{\phi^i}(a_t|s_t)\Big|s_0=s\right],\\
      \nabla^2_{\theta\phi^i}V^i_{\pi_{\phi}}(s;\theta)&=\mathbb{E}_{\tau\sim D^{\pi_\phi}}\left[\left(\sum_{t=0}^{T-1}\nabla_{\phi^i}\log\pi^i_{\phi^i}(a_t|s_t)\right)\nabla_{\theta}R^i(\tau;\theta)^\intercal\Big|s_0=s\right],\\
      \nabla^2_{\phi\phi^i}V^i_{\pi_{\phi}}(s;\theta)&=\mathbb{E}_{\tau\sim D^{\pi_\phi}}\left[R^i(\tau;\theta)\left((\sum_{t=0}^{T-1}\nabla_{\phi^i}\log\pi^i_{\phi^i}(a_t|s_t))(\sum_{t=0}^{T-1}\nabla_{\phi}\log\pi_{\phi}(a_t|s_t))^\intercal\right.\right.\notag\\
      &\left.\left.\quad\quad\quad\quad\quad\quad+\sum_{t=0}^{T-1}\nabla^2_{\phi\phi^i}\log\pi_{\phi}(a_t|s_t)\right)\Big|s_0=s\right],
  \end{align*}
  where
  \begin{align*}
      \nabla_{\phi}\log\pi_{\phi}(a_t|s_t)&=(\nabla_{\phi^1}\log\pi^1_{\phi^1}(a_t|o^1_t),\cdots, \nabla_{\phi^N}\log\pi^N_{\phi^N}(a_t|o^N_t)),\notag\\
      \nabla^2_{\phi\phi^i}\log\pi_{\phi}(a_t|s_t)&=(\nabla^2_{\phi^1\phi^i}\log\pi^1_{\phi^1}(a_t|o^1_t),\cdots,\nabla^2_{\phi^N\phi^i}\log\pi^N_{\phi^N}(a_t|o^N_t)).\notag
  \end{align*}   
\end{lemma}
\begin{proof}[Proof for Lemma \ref{lma:V in Nplayer}.]
According to the proof of Lemma \ref{lma:V}, the derivation for $\nabla_{\phi^i}V^i_{\pi_{\phi}}(s;\theta)$ and $\nabla^2_{\phi\phi^i}V^i_{\pi_{\phi}}(s;\theta)$ are based on $\nabla_{\phi^i}\log{D^{\pi_\phi}(\tau|s_0=s)}$. Recall $D^{\pi_\phi}(\tau)=\Pi_{t=0}^{T-1} \rho(s_0)P(s_{t+1}|a_t,s_t)\left[\Pi_{i=1}^N \pi^i_{\phi^i}(a_t^i|o^i)\right]$, we have
\begin{align}
\label{equ:gradient of log D in POMDP}
      \nabla_{\phi^i}\log{D^{\pi_\phi}(\tau|s_0=s)}
      &\!=\!\nabla_{\phi^i}\!\left[\!\sum_{t=0}^{T-1}\!\log{P(s_{t+1}|s_t,a_t)}\!+\!\sum_{t=0}^{T-1}\!\log{\pi^i_{\phi^i}(a_t|s_t)}\!+\!\sum_{t=0}^{T-1}\log{\pi^{-i}_{\phi^{-i}}(a_t|o^{-i}_t)}\right]\notag\\
      &\!=\!\sum_{t=0}^{T-1}\nabla_{\phi^i}\log{\pi^i_{\phi^i}(a_t|s_t)}.
\end{align}
Similarly, we have
  $\nabla_{\phi}\log{D^{\pi_\phi}(\tau|s_0=s)}\!=\!\sum_{t=0}^{T-1}\nabla_{\phi}\log{\pi_{\phi}(a_t|s_t)}$.
The remaining of the proof based on (\ref{equ:gradient of log D in POMDP}) is exactly the same as the proof of Lemma \ref{lma:V}.
\end{proof}
In the light of Lemma \ref{lma:V in Nplayer}, we replace Lemma \ref{lma:V2} by the following Lemma \ref{lma:V2 in Nplayer} for the gradient of regularized state-value function based on the incentivized reward.

\begin{lemma}\label{lma:V2 in Nplayer}
  In an incentive regularized Markov game $\mathcal{G}_{\theta}^{'}$, we have
  \begin{align*}
    &\nabla_{\theta}V^{(i)}_{\pi_{\phi}}(s;\theta)=\nabla_{\theta}V^{i}_{\pi_{\phi}}(s;\theta);\\
    &\nabla_{\phi^i} V^{(i)}_{\pi_{\phi}}(s;\theta)=\nabla_{\phi^i} V^{i}_{\pi_{\phi}}(s;\theta)-\lambda\mathbb{E}_{\tau\sim D^{\pi_{\phi}}}\left[\left(\sum_{t=0}^{T-1}\gamma^t\left(\log\pi_{\phi^i}^i(a_t^i|s_t)-\frac{1}{N-1}\sum_{j=1,...,N\atop j\neq i}\log\pi_{\phi^{j}}^{j}(a_t^{j}|s_t)\right)\right)\left(\sum_{t=0}^{T-1}\nabla_{\phi^i}\log\pi_{\phi^i}^i(a_t^i|s_t)\right)\right.\notag\\
    &\quad\quad\quad\quad\quad\quad\quad\quad\quad\quad\quad\quad\quad\quad\quad+\left.\sum_{t=0}^{T-1}\gamma^t\nabla_{\phi^i}\log\pi_{\phi^i}^i(a_t^i|s_t)\right];\\
    &\nabla^2_{\theta\phi^i}V^{(i)}_{\pi_{\phi}}(s;\theta)=\nabla^2_{\theta\phi^i}V^{i}_{\pi_{\phi}}(s;\theta);\\
\end{align*}
\begin{align*}
    &\nabla^2_{\phi\phi^i}V^{(i)}_{\pi_{\phi}}(s;\theta)=\nabla^2_{\phi\phi^i}V^{i}_{\pi_{\phi}}(s;\theta)\!-\!\lambda\mathbb{E}_{\tau\sim D^{\pi_\phi}}\left[\left(\sum_{t=0}^{T-1}\!\nabla_{\phi^i}\!\log\pi_{\phi^i}^i(a_t^i|s_t)\right)\left(\!\sum_{t=0}^{T-1}\!\gamma^t\left(\nabla_{\phi}\!\log\pi_{\phi^i}^i(a_t^i|s_t)-\frac{1}{N-1}\sum_{j=1,...,N\atop j\neq i}\nabla_{\phi}\log\pi_{\phi^{j}}^{j}(a_t^{j}|s_t)\right)\right)^\intercal\right.\\
    &\quad\quad\quad\quad+\left.\left(\sum_{t=0}^{T-1}\gamma^t\left(\log\pi_{\phi^i}^i(a_t^i|s_t)-\frac{1}{N-1}\sum_{j=1,...,N\atop j\neq i}\log\pi_{\phi^{j}}^{j}(a_t^{j}|s_t)\right)\right)\left(\sum_{t=0}^{T-1}\nabla_{\phi^i}\log\pi_{\phi^i}^i(a_t^i|s_t)\right)\left(\sum_{t=0}^{T-1}\nabla_{\phi}\log \pi_{\phi}(a_t|s_t)\right)^\intercal\right.\\
    &\quad\quad\quad\quad\left.+\left(\!\sum_{t=0}^{T-1}\!\gamma^t\left(\log\pi_{\phi^i}^i(a_t^i|s_t)-\frac{1}{N-1}\sum_{j=1,...,N\atop j\neq i}\log\pi_{\phi^{j}}^{j}(a_t^{j}|s_t)\right)\right)\left(\!\sum_{t=0}^{T-1}\!\nabla^2_{\phi\phi^i}\!\log\pi_{\phi}(a_t^i|s_t)\right)\right.\notag\\
    &\quad\quad\quad\quad\left.+\left(\!\sum_{t=0}^{T-1}\!\gamma^t\nabla_{\phi^i}\log \pi^i_{\phi^i}(a_t|s_t)\right)\left(\sum_{t=0}^{T-1}\nabla_{\phi}\log \pi_{\phi}(a_t|s_t\!)\!\right)^\intercal+\left(\sum_{t=0}^{T-1}\gamma^t\nabla
    ^2_{\phi\phi^i}\!\log \pi_{\phi}(a_t|s_t)\right)\right]
\end{align*}
  
\end{lemma}
\begin{proof}[Proof for Lemma \ref{lma:V2 in Nplayer}]Recall the definition of $V^{(i)}_{\pi_{\phi}}(s;\theta)$ in (\ref{eq:V_nplayer}), we have,
\begin{align}
    V^{(i)}_{\pi_{\phi}}(s;\theta)&=\mathbb{E}_{\tau\sim D^{\pi_{\phi}}}\left[\sum_{t=0}^{T-1}\gamma^t{\cdot}r^{(i)}_{\pi}(s_t,a^i_t,a^{-i}_t;{\theta})|s_0=s\right]\notag\\
    &=\mathbb{E}_{\tau\sim D^{\pi_{\phi}}}\left[\sum_{t=0}^{T-1}\gamma^t{\cdot}r^{i}_{\pi}(s_t,a^i_t,a^{-i}_t;{\theta})-\lambda\sum_{t=0}^{T-1}\gamma^t\left(\log\pi_{\phi^i}^i(a_t^i|s_t)-\frac{1}{N-1}\sum_{j=1,...,N\atop j\neq i}\log\pi_{\phi^{j}}^{j}(a_t^{j}|s_t)\right)|s_0=s\right]\notag\\
    \vspace{-0.05in}
    &=V^{i}_{\pi_{\phi}}(s;\theta)-\lambda\mathbb{E}_{\tau\sim D^{\pi_{\phi}}}\left[\sum_{t=0}^{T-1}\gamma^t\left(\log\pi_{\phi^i}^i(a_t^i|s_t)-\frac{1}{N-1}\sum_{j=1,...,N\atop j\neq i}\log\pi_{\phi^{j}}^{j}(a_t^{j}|s_t)\right)|s_0=s\right].\label{VU_2}
    \vspace{-0.25in}
\end{align}
For easy notation, we define
\begin{align*}
    U_{\pi_{\phi}}^i(s)=\mathbb{E}_{\tau\sim D^{\pi_{\phi}}}\left[\sum_{t=0}^{T-1}\gamma^t\left(\log\pi_{\phi^i}^i(a_t^i|s_t)-\frac{1}{N-1}\sum_{j=1,...,N\atop j\neq i}\log\pi_{\phi^{j}}^{j}(a_t^{j}|s_t)\right)|s_0=s\right].
\end{align*}
Then we have,
\begin{align}
    \nabla_{\theta}U_{\pi_{\phi}}^i(s)=&0;\label{U_theta_2}\\
    \nabla_{\phi^i}U_{\pi_{\phi}}^i(s)=&\nabla_{\phi^i}\left[\sum_{\tau}D^{\pi_{\phi}}(\tau|s_0\!=\!s)\sum_{t=0}^{T-1}\gamma^t\left(\log\pi_{\phi^i}^i(a_t^i|s_t)-\frac{1}{N-1}\sum_{j=1,...,N\atop j\neq i}\log\pi_{\phi^{j}}^{j}(a_t^{j}|s_t)\right)\right]\notag\\
    =&\sum_{\tau}D^{\pi_{\phi}}(\tau|s_0\!=\!s)\left[\nabla_{\phi_i}\log D^{\pi_{\phi}}(\tau|s_0\!=\!s)\sum_{t=0}^{T-1}\gamma^t\left(\log\pi_{\phi^i}^i(a_t^i|s_t)-\frac{1}{N-1}\sum_{j=1,...,N\atop j\neq i}\log\pi_{\phi^{j}}^{j}(a_t^{j}|s_t)\right)\right.\notag\\
    &\quad\quad\quad\left.+\sum_{t=0}^{T-1}\gamma^t\nabla_{\phi^i}\log\pi_{\phi^i}^i(a_t^i|s_t)\right]\notag\\
    =&\mathbb{E}_{\tau\sim D^{\pi_{\phi}}}\left[\left(\sum_{t=0}^{T-1}\gamma^t\left(\log\pi_{\phi^i}^i(a_t^i|s_t)-\frac{1}{N-1}\sum_{j=1,...,N\atop j\neq i}\log\pi_{\phi^{j}}^{j}(a_t^{j}|s_t)\right)\right)\left(\sum_{t=0}^{T-1}\nabla_{\phi^i}\log\pi_{\phi^i}^i(a_t^i|s_t)\right)\right.\notag\\
    &\quad\quad\quad\left.+\sum_{t=0}^{T-1}\gamma^t\nabla_{\phi^i}\log\pi_{\phi^i}^i(a_t^i|s_t)\right];\label{U_phi_i_2}\\
\end{align}
\begin{align}
    \nabla_{\theta\phi^i}U_{\pi_{\phi}}^i(s)=&0;\label{U_theta_phi_i_2}\\
    \nabla_{\phi\phi^i}U_{\pi_{\phi}}^i(s)=&\overbrace{\nabla_{\phi}\left[\sum_{\tau}D^{\pi_\phi}(\tau|s_0=s)\left(\sum_{t=0}^{T-1}\gamma^t\left(\log\pi_{\phi^i}^i(a_t^i|s_t)-\frac{1}{N-1}\sum_{j=1,...,N\atop j\neq i}\log\pi_{\phi^{j}}^{j}(a_t^{j}|s_t)\right)\right)\left(\sum_{t=0}^{T-1}\nabla_{\phi^i}\log\pi_{\phi^i}^i(a_t^i|s_t)\right)\right]}^{\displaystyle{:=(\RNum{1})}}\notag\\
    &+\underbrace{\nabla_{\phi}\left[\sum_{\tau}D^{\pi_\phi}(\tau|s_0=s)\sum_{t=0}^{T-1}\gamma^t\nabla_{\phi^i}\log\pi_{\phi^i}^i(a_t^i|s_t)\right]}_{\displaystyle{:=(\RNum{2})}}\notag.
\end{align}
Next we are going to compute (\RNum{1}) and (\RNum{2}).
\begin{align*}
    (\RNum{1})=&\sum_{\tau}\!D^{\pi_\phi}\!(\tau|s_0\!=\!s)\!\left(\sum_{t=0}^{T-1}\!\gamma^t\!\left(\log\pi_{\phi^i}^i(a_t^i|s_t)-\frac{1}{N-1}\sum_{j=1,...,N\atop j\neq i}\log\pi_{\phi^{j}}^{j}(a_t^{j}|s_t)\right)\!\right)\!\!\left(\!\sum_{t=0}^{T-1}\nabla_{\phi^i}\!\log\!\pi_{\phi^i}^i(a_t^i|s_t\!)\!\right)\!\nabla_{\phi}\!\log D^{\pi_\phi}(\tau|s_0=s)^\intercal\\
    &+\sum_{\tau}D^{\pi_\phi}(\tau|s_0=s)\nabla_{\phi}\left[\left(\sum_{t=0}^{T-1}\gamma^t\left(\log\pi_{\phi^i}^i(a_t^i|s_t)-\frac{1}{N-1}\sum_{j=1,...,N\atop j\neq i}\log\pi_{\phi^{j}}^{j}(a_t^{j}|s_t)\right)\right)\left(\sum_{t=0}^{T-1}\nabla_{\phi^i}\log\pi_{\phi^i}^i(a_t^i|s_t)\right)\right]\\
    \!=\!&\!\sum_{\tau}D^{\pi_\phi}\!(\tau|s_0=s\!)\!\left(\!\sum_{t=0}^{T-1}\!\gamma^t\left(\log\pi_{\phi^i}^i(a_t^i|s_t)-\frac{1}{N-1}\sum_{j=1,...,N\atop j\neq i}\log\pi_{\phi^{j}}^{j}(a_t^{j}|s_t)\right)\!\right)\!\!\left(\sum_{t=0}^{T-1}\!\nabla_{\phi^i}\!\log\pi_{\phi^i}^i(a_t^i|s_t\!)\!\right)\!\!\left(\sum_{t=0}^{T-1}\nabla_{\phi}\log \pi_{\phi}(a_t|s_t\!)\!\right)^\intercal\\
    &+\sum_{\tau}D^{\pi_\phi}(\tau|s_0=s)\left[\left(\sum_{t=0}^{T-1}\nabla_{\phi^i}\log\pi_{\phi^i}^i(a_t^i|s_t)\right)\left(\sum_{t=0}^{T-1}\gamma^t\left(\nabla_{\phi}\log\pi_{\phi^i}^i(a_t^i|s_t)-\frac{1}{N-1}\sum_{j=1,...,N\atop j\neq i}\nabla_\phi\log\pi_{\phi^{j}}^{j}(a_t^{j}|s_t)\right)\right)^\intercal\right.\\
    &\quad\quad\quad\quad\quad\quad\quad\quad\quad\left.+\left(\sum_{t=0}^{T-1}\gamma^t\left(\log\pi_{\phi^i}^i(a_t^i|s_t)-\frac{1}{N-1}\sum_{j=1,...,N\atop j\neq i}\log\pi_{\phi^{j}}^{j}(a_t^{j}|s_t)\right)\right)\left(\sum_{t=0}^{T-1}\nabla^2_{\phi\phi^i}\log\pi_{\phi}(a_t^i|s_t)\right)\right]\\
    =&\mathbb{E}_{\tau\sim D^{\pi_\phi}}\!\!\!\left[\left(\sum_{t=0}^{T-1}\gamma^t\left(\log\pi_{\phi^i}^i(a_t^i|s_t)-\frac{1}{N-1}\sum_{j=1,...,N\atop j\neq i}\log\pi_{\phi^{j}}^{j}(a_t^{j}|s_t)\right)\right)\left(\sum_{t=0}^{T-1}\nabla_{\phi^i}\log\pi_{\phi^i}^i(a_t^i|s_t)\right)\left(\sum_{t=0}^{T-1}\nabla_{\phi}\log \pi_{\phi}(a_t|s_t)\right)^\intercal\right.\\
    &\quad\quad\quad\quad\left.+\left(\sum_{t=0}^{T-1}\!\nabla_{\phi^i}\!\log\pi_{\phi^i}^i(a_t^i|s_t)\right)\left(\!\sum_{t=0}^{T-1}\!\gamma^t\left(\nabla_{\phi}\!\log\pi_{\phi^i}^i(a_t^i|s_t)-\frac{1}{N-1}\sum_{j=1,...,N\atop j\neq i}\nabla_{\phi}\log\pi_{\phi^{j}}^{j}(a_t^{j}|s_t)\right)\right)^\intercal\right.\\
    &\quad\quad\quad\quad\left.+\left(\!\sum_{t=0}^{T-1}\!\gamma^t\left(\log\pi_{\phi^i}^i(a_t^i|s_t)-\frac{1}{N-1}\sum_{j=1,...,N\atop j\neq i}\log\pi_{\phi^{j}}^{j}(a_t^{j}|s_t)\right)\right)\left(\!\sum_{t=0}^{T-1}\!\nabla^2_{\phi\phi^i}\!\log\pi_{\phi}(a_t^i|s_t)\right)\right];\\
    (\RNum{2})=&\sum_{\tau}D^{\pi_\phi}(\tau|s_0=s)\left[\left(\sum_{t=0}^{T-1}\gamma^t\nabla_{\phi^i}\log \pi^i_{\phi^i}(a_t|s_t)\right)\left(\sum_{t=0}^{T-1}\nabla_{\phi}\log \pi_{\phi}(a_t|s_t)\right)^\intercal\right.\\
    &\left.\quad\quad\quad+\left(\sum_{t=0}^{T-1}\!\gamma^t\nabla^2_{\phi\phi^i}\log \pi_{\phi}(a_t|s_t)\!\right)\!\right]\\
    =&\mathbb{E}_{\tau\sim D^{\pi_\phi}}\!\!\left[\!\left(\!\sum_{t=0}^{T-1}\!\gamma^t\!\nabla_{\phi^i}\!\log \pi^i_{\phi^i}(a_t|s_t\!)\!\right)\!\!\left(\sum_{t=0}^{T-1}\!\nabla_{\phi}\!\log \pi_{\phi}(a_t|s_t\!)\!\right)^\intercal\!\!\!\!+\!\!\left(\sum_{t=0}^{T-1}\gamma^t\nabla
    ^2_{\phi\phi^i}\!\log \pi_{\phi}(a_t|s_t\!)\!\right)\!\right].
\end{align*}
Therefore,
\begin{align}\label{U_phi_phi_i_2}
    &\nabla_{\phi\phi^i}U_{\pi_\phi}^{i}(s)\notag=\mathbb{E}_{\tau\sim D^{\pi_\phi}}\!\!\!\left[\left(\sum_{t=0}^{T-1}\gamma^t\left(\log\pi_{\phi^i}^i(a_t^i|s_t)-\frac{1}{N-1}\sum_{j=1,...,N\atop j\neq i}\log\pi_{\phi^{j}}^{j}(a_t^{j}|s_t)\right)\right)\left(\sum_{t=0}^{T-1}\nabla_{\phi^i}\log\pi_{\phi^i}^i(a_t^i|s_t)\right)\left(\sum_{t=0}^{T-1}\nabla_{\phi}\log \pi_{\phi}(a_t|s_t)\right)^\intercal\right.\\
    &\quad\quad\quad\quad\left.+\left(\sum_{t=0}^{T-1}\!\nabla_{\phi^i}\!\log\pi_{\phi^i}^i(a_t^i|s_t)\right)\left(\!\sum_{t=0}^{T-1}\!\gamma^t\left(\nabla_{\phi}\!\log\pi_{\phi^i}^i(a_t^i|s_t)-\frac{1}{N-1}\sum_{j=1,...,N\atop j\neq i}\nabla_{\phi}\log\pi_{\phi^{j}}^{j}(a_t^{j}|s_t)\right)\right)^\intercal\right.\\
    &\quad\quad\quad\quad\left.+\left(\!\sum_{t=0}^{T-1}\!\gamma^t\left(\log\pi_{\phi^i}^i(a_t^i|s_t)-\frac{1}{N-1}\sum_{j=1,...,N\atop j\neq i}\log\pi_{\phi^{j}}^{j}(a_t^{j}|s_t)\right)\right)\left(\!\sum_{t=0}^{T-1}\!\nabla^2_{\phi\phi^i}\!\log\pi_{\phi}(a_t^i|s_t)\right)\right.\notag\\
    &\quad\quad\quad\quad\left.+\left(\!\sum_{t=0}^{T-1}\!\gamma^t\nabla_{\phi^i}\log \pi^i_{\phi^i}(a_t|s_t)\right)\left(\sum_{t=0}^{T-1}\nabla_{\phi}\log \pi_{\phi}(a_t|s_t\!)\!\right)^\intercal+\left(\sum_{t=0}^{T-1}\gamma^t\nabla
    ^2_{\phi\phi^i}\!\log \pi_{\phi}(a_t|s_t)\right)\right]    
\end{align}
Since (\ref{VU_2}) can be written as
\begin{align*}
    V^{(i)}_{\pi_{\phi}}(s;\theta)=V^{i}_{\pi_{\phi}}(s;\theta)-\lambda U_{\pi_\phi}^{i}(s),
\end{align*}
then together with (\ref{U_theta}-\ref{U_phi_phi_i}), we have
\begin{align*}
    &\nabla_{\theta}V^{(i)}_{\pi_{\phi}}(s;\theta)=\nabla_{\theta}V^{i}_{\pi_{\phi}}(s;\theta);\\
    &\nabla_{\phi^i} V^{(i)}_{\pi_{\phi}}(s;\theta)=\nabla_{\phi^i} V^{i}_{\pi_{\phi}}(s;\theta)-\lambda\mathbb{E}_{\tau\sim D^{\pi_{\phi}}}\left[\left(\sum_{t=0}^{T-1}\gamma^t\left(\log\pi_{\phi^i}^i(a_t^i|s_t)-\frac{1}{N-1}\sum_{j=1,...,N\atop j\neq i}\log\pi_{\phi^{j}}^{j}(a_t^{j}|s_t)\right)\right)\left(\sum_{t=0}^{T-1}\nabla_{\phi^i}\log\pi_{\phi^i}^i(a_t^i|s_t)\right)\right.\notag\\
    &\quad\quad\quad\quad\quad\quad\quad\quad\quad\quad\quad\quad\quad\quad\quad+\left.\sum_{t=0}^{T-1}\gamma^t\nabla_{\phi^i}\log\pi_{\phi^i}^i(a_t^i|s_t)\right];\\
    &\nabla^2_{\theta\phi^i}V^{(i)}_{\pi_{\phi}}(s;\theta)=\nabla^2_{\theta\phi^i}V^{i}_{\pi_{\phi}}(s;\theta);\\
    &\nabla^2_{\phi\phi^i}V^{(i)}_{\pi_{\phi}}(s;\theta)=\nabla^2_{\phi\phi^i}V^{i}_{\pi_{\phi}}(s;\theta)\!-\!\lambda\mathbb{E}_{\tau\sim D^{\pi_\phi}}\left[\left(\sum_{t=0}^{T-1}\!\nabla_{\phi^i}\!\log\pi_{\phi^i}^i(a_t^i|s_t)\right)\left(\!\sum_{t=0}^{T-1}\!\gamma^t\left(\nabla_{\phi}\!\log\pi_{\phi^i}^i(a_t^i|s_t)-\frac{1}{N-1}\sum_{j=1,...,N\atop j\neq i}\nabla_{\phi}\log\pi_{\phi^{j}}^{j}(a_t^{j}|s_t)\right)\right)^\intercal\right.\\
    &\quad\quad\quad\quad+\left.\left(\sum_{t=0}^{T-1}\gamma^t\left(\log\pi_{\phi^i}^i(a_t^i|s_t)-\frac{1}{N-1}\sum_{j=1,...,N\atop j\neq i}\log\pi_{\phi^{j}}^{j}(a_t^{j}|s_t)\right)\right)\left(\sum_{t=0}^{T-1}\nabla_{\phi^i}\log\pi_{\phi^i}^i(a_t^i|s_t)\right)\left(\sum_{t=0}^{T-1}\nabla_{\phi}\log \pi_{\phi}(a_t|s_t)\right)^\intercal\right.\\
    &\quad\quad\quad\quad\left.+\left(\!\sum_{t=0}^{T-1}\!\gamma^t\left(\log\pi_{\phi^i}^i(a_t^i|s_t)-\frac{1}{N-1}\sum_{j=1,...,N\atop j\neq i}\nabla_\phi\log\pi_{\phi^{j}}^{j}(a_t^{j}|s_t)\right)\right)\left(\!\sum_{t=0}^{T-1}\!\nabla^2_{\phi\phi^i}\!\log\pi_{\phi}(a_t^i|s_t)\right)\right.\notag\\
    &\quad\quad\quad\quad\left.+\left(\!\sum_{t=0}^{T-1}\!\gamma^t\nabla_{\phi^i}\log \pi^i_{\phi^i}(a_t|s_t)\right)\left(\sum_{t=0}^{T-1}\nabla_{\phi}\log \pi_{\phi}(a_t|s_t\!)\!\right)^\intercal+\left(\sum_{t=0}^{T-1}\gamma^t\nabla
    ^2_{\phi\phi^i}\!\log \pi_{\phi}(a_t|s_t)\right)\right]
\end{align*}
\end{proof}

\textbf{Extension to partial observations.} We first formulate the multi-agent POMDP with incentivized reward. With a light abuse of notation, we denote $\mathcal{G}_{\theta} = (\mathcal{S}, \{\mathcal{A}^i, \Omega^i,\mathcal{O}^i, r^{i}(\cdot;\theta) \}_{i\in\{1,\cdots,N\}},\mathcal{P}, \gamma)$, where $\Omega^i$ is the observation set for agent $i$, $\mathcal{O}^i: \mathcal{S} \rightarrow \Omega^i$ is the observation function for agent $i$, and we denote $o^i = \mathcal{O}^i(s)$ as the observation of agent $i$ at state $s$. Other elements in the tuple remain same as in Section \ref{subsec:irg} except for the zero-sum restriction to $r^i(\cdot;\theta)$. We define observation-based policy of player $i$ as $\pi^i: \Omega^i \times \mathcal{A}^i \rightarrow[0,1]$ and parameterize the policy for the player $i$ by $\phi^i$. Then based on $\mathcal{G}_{\theta}$, we formulate the regularized multi-agent POMDP with incentivized reward as $\mathcal{G}'_{\theta} = (\mathcal{S}, \{\mathcal{A}^i, \Omega^i,\mathcal{O}^i, r^{(i)}(\cdot;\theta) \}_{i\in\{1,\cdots,N\}}, \lambda,\mathcal{P}, \gamma)$, where $r^{(i)}(s,a^i,a^{-i};\theta)$ represents $r_{\pi_{\phi}}^{(i)}(s,a^i,a^{-i};\theta)=r^{i}(s,a^i,a^{-i};\theta) - \lambda\log(\pi^i_{\phi^i}(a^i|o^i))$ and $\lambda \geq0$.
For a trajectory $\tau=(s_0,a_0,s_1, \cdots,a_{T-1},s_T)$, where $a_t=(a^1_t,\cdots,a^N_t)$, $t=0,\cdots,T$, the probability distribution of the trajectory is,
\begin{align*}
    D^{\pi_\phi}(\tau):=\prod_{t=0}^{T-1} \rho(s_0)P(s_{t+1}|a_t,s_t)\left[\prod_{i=1}^N \pi^i_{\phi^i}(a_t^i|o_t^i)\right],
\end{align*}
where $o_t^i=\mathcal{O}^i(s_t)$.
In this setting, we replace Lemma \ref{lma:V in Nplayer} by the following Lemma \ref{lma:V in POMDP} for the gradient of state-value function.

\begin{lemma}\label{lma:V in POMDP}
  For an incentivized Markov game $\mathcal{G}_{\theta} = (\mathcal{S}, \{\mathcal{A}^i, \Omega^i,\mathcal{O}^i, r^{i}(\cdot;\theta) \}_{i\in\{1,\cdots,N\}},\mathcal{P}, \gamma)$, let $R^i(\tau;\theta)$ be the total reward for an sample trajectory $\tau$ following $\pi_{\phi}$ for $T$ steps, starting from initial state $s_0=s$. We have
  \begin{align*}
      \nabla_{\theta}V^i_{\pi_{\phi}}(s;\theta)&=\mathbb{E}_{\tau\sim D^{\pi_\phi}}\left[\nabla_{\theta}R^i(\tau;\theta)|s_0=s\right],\\
      \nabla_{\phi^i} V^i_{\pi_{\phi^i}}(s;\theta)&=\mathbb{E}_{\tau\sim D^{\pi_\phi}}\left[R^i(\tau;\theta)\sum_{t=0}^{T-1}\nabla_{\phi^i}\log\pi^i_{\phi^i}(a_t|o^i_t)\Big|s_0=s\right],\\
      \nabla^2_{\theta\phi^i}V^i_{\pi_{\phi}}(s;\theta)&=\mathbb{E}_{\tau\sim D^{\pi_\phi}}\left[\left(\sum_{t=0}^{T-1}\nabla_{\phi^i}\log\pi^i_{\phi^i}(a_t|o^i_t)\right)\nabla_{\theta}R^i(\tau;\theta)^\intercal\Big|s_0=s\right],\\
      \nabla^2_{\phi\phi^i}V^i_{\pi_{\phi}}(s;\theta)&=\mathbb{E}_{\tau\sim D^{\pi_\phi}}\left[R^i(\tau;\theta)\left((\sum_{t=0}^{T-1}\nabla_{\phi^i}\log\pi^i_{\phi^i}(a_t|o^i_t))(\sum_{t=0}^{T-1}\nabla_{\phi}\log\pi_{\phi}(a_t|o_t))^\intercal\right.\right.\notag\\
      &\left.\left.\quad\quad\quad\quad\quad\quad+\sum_{t=0}^{T-1}\nabla^2_{\phi\phi^i}\log\pi_{\phi}(a_t|o_t)\right)\Big|s_0=s\right],
  \end{align*}
  where
  \begin{align*}
      \nabla_{\phi}\log\pi_{\phi}(a_t|o_t)&=(\nabla_{\phi^1}\log\pi^1_{\phi^1}(a_t|o^1_t),\cdots, \nabla_{\phi^N}\log\pi^N_{\phi^N}(a_t|o^N_t)),\notag\\
      \nabla^2_{\phi\phi^i}\log\pi_{\phi}(a_t|o_t)&=(\nabla^2_{\phi^1\phi^i}\log\pi^1_{\phi^1}(a_t|o^1_t),\cdots,\nabla^2_{\phi^N\phi^i}\log\pi^N_{\phi^N}(a_t|o^N_t)).\notag
  \end{align*}   
\end{lemma}
In the light of Lemma \ref{lma:V in POMDP}, we replace Lemma \ref{lma:V2 in Nplayer} by the following Lemma \ref{lma:V2 in POMDP} for the gradient of regularized state-value function based on partially observed state.
\begin{lemma}
\label{lma:V2 in POMDP}
For a regularized Markov game with incentivized reward $\mathcal{G}'_{\theta} = (\mathcal{S}, \{\mathcal{A}^i, \Omega^i,\mathcal{O}^i, r^{(i)}(\cdot;\theta) \}_{i\in\{1,\cdots,N\}}, \lambda,\mathcal{P}, \gamma)$, we have
\begin{align*}
    &\nabla_{\theta}V^{(i)}_{\pi_{\phi}}(s;\theta)=\nabla_{\theta}V^{i}_{\pi_{\phi}}(s;\theta);\\
    &\nabla_{\phi^i} V^{(i)}_{\pi_{\phi}}(s;\theta)=\nabla_{\phi^i} V^{i}_{\pi_{\phi}}(s;\theta)-\lambda\mathbb{E}_{\tau\sim D^{\pi_{\phi}}}\left[\left(\sum_{t=0}^{T-1}\gamma^t\left(\log\pi_{\phi^i}^i(a_t^i|o_t^i)-\frac{1}{N-1}\sum_{j=1,...,N\atop j\neq i}\log\pi_{\phi^{j}}^{j}(a_t^{j}|o_t^j)\right)\right)\left(\sum_{t=0}^{T-1}\nabla_{\phi^i}\log\pi_{\phi^i}^i(a_t^i|o_t^i)\right)\right.\notag\\
    &\quad\quad\quad\quad\quad\quad\quad\quad\quad\quad\quad\quad\quad\quad\quad+\left.\sum_{t=0}^{T-1}\gamma^t\nabla_{\phi^i}\log\pi_{\phi^i}^i(a_t^i|o_t^i)\right];\\
    &\nabla^2_{\theta\phi^i}V^{(i)}_{\pi_{\phi}}(s;\theta)=\nabla^2_{\theta\phi^i}V^{i}_{\pi_{\phi}}(s;\theta);\\
    &\nabla^2_{\phi\phi^i}V^{(i)}_{\pi_{\phi}}(s;\theta)=\nabla^2_{\phi\phi^i}V^{i}_{\pi_{\phi}}(s;\theta)\!-\!\lambda\mathbb{E}_{\tau\sim D^{\pi_\phi}}\left[\left(\sum_{t=0}^{T-1}\!\nabla_{\phi^i}\!\log\pi_{\phi^i}^i(a_t^i|o_t^i)\right)\left(\!\sum_{t=0}^{T-1}\!\gamma^t\left(\nabla_{\phi}\!\log\pi_{\phi^i}^i(a_t^i|o_t^i)-\frac{1}{N-1}\sum_{j=1,...,N\atop j\neq i}\nabla_{\phi}\log\pi_{\phi^{j}}^{j}(a_t^{j}|o_t^j)\right)\right)^\intercal\right.\\
    &\quad\quad\quad\quad+\left.\left(\sum_{t=0}^{T-1}\gamma^t\left(\log\pi_{\phi^i}^i(a_t^i|o_t^i)-\frac{1}{N-1}\sum_{j=1,...,N\atop j\neq i}\log\pi_{\phi^{j}}^{j}(a_t^{j}|o_t^j)\right)\right)\left(\sum_{t=0}^{T-1}\nabla_{\phi^i}\log\pi_{\phi^i}^i(a_t^i|a_t^i)\right)\left(\sum_{t=0}^{T-1}\nabla_{\phi}\log \pi_{\phi}(a_t|o_t)\right)^\intercal\right.\\
    &\quad\quad\quad\quad\left.+\left(\!\sum_{t=0}^{T-1}\!\gamma^t\left(\log\pi_{\phi^i}^i(a_t^i|o_t^i)-\frac{1}{N-1}\sum_{j=1,...,N\atop j\neq i}\log\pi_{\phi^{j}}^{j}(a_t^{j}|o_t^j)\right)\right)\left(\!\sum_{t=0}^{T-1}\!\nabla^2_{\phi\phi^i}\!\log\pi_{\phi}(a_t^i|o_t^i)\right)\right.\notag\\
    &\quad\quad\quad\quad\left.+\left(\!\sum_{t=0}^{T-1}\!\gamma^t\nabla_{\phi^i}\log \pi^i_{\phi^i}(a_t|o_t)\right)\left(\sum_{t=0}^{T-1}\nabla_{\phi}\log \pi_{\phi}(a_t|o_t\!)\!\right)^\intercal+\left(\sum_{t=0}^{T-1}\gamma^t\nabla
    ^2_{\phi\phi^i}\!\log \pi_{\phi}(a_t|o_t)\right)\right]
\end{align*}
\end{lemma}

The proofs for Lemma \ref{lma:V in POMDP} and \ref{lma:V2 in POMDP} are quite similar to the proofs for \ref{lma:V in Nplayer} and \ref{lma:V2 in Nplayer} with little change on the notation. Therefore, we do not discuss the proof in this section.

\newpage
\section{PEM and Entropy-regularized OMWU}
In this section, we are going to briefly state the pseudocode for NE solvers Policy Extragradient Method (PEM) \cite{cen2021fast} and Entropy-regularized OMWU \cite{cen2022faster} in this section.

\subsection{Policy Extragradient Method (PEM)}
To begin with, we first introduce the entropy-regualrized zero-sum two-player matrix game.
\begin{align}\label{problem:matrix}
    \max_{\pi_1\in\Delta(\mathcal{A}_1)}\min_{\pi_2\in\Delta(\mathcal{A}_2)}f_\lambda(A; \pi_1,\pi_2):=\pi_1^\intercal A\pi_2+\lambda\mathcal{H}(\pi_1)-\lambda\mathcal{H}(\pi_2),
\end{align}
where $A\in\mathbb{R}^{m\times n}$ denotes the payoff matrix, $\pi_1\in\Delta(\mathcal{A}_1)$ and $\pi_2\in\Delta(\mathcal{A}_2)$ stand for the mixed/randomized policies of each player, defined as distributions over the probability simplex $\Delta(\mathcal{A}_1)$ and $\Delta(\mathcal{A}_2)$, $\mathcal{H}(\pi)=-\sum_{i}\pi_i\log(\pi_i)$, and $\lambda$ is the regularization parameter.

Cen \cite{cen2021fast} proposed two extragradient methods: PU and OMWU with linear convergence guarantees to solve the entropy-regularized zero-sum two-player matrix game (\ref{problem:matrix}).

\begin{minipage}{.47\textwidth}
\centering
\begin{algorithm}[H]
\caption{The PU method}
\label{alg_PU}
\begin{algorithmic}
\STATE {\bfseries Initialization: } $\pi_1^{(0)},\pi_2^{(0)}$.
\STATE {\bfseries Parameters: } Learning rate $\eta_t$

\FOR{t=0,1,2,...}

    \STATE Update $\overline{\pi}_1$ and $\overline{\pi}_2$ for every action $a=(a_1,a_2)\in\mathcal{A}=\mathcal{A}_1\times\mathcal{A}_2$ according to
    \begin{align*}
        \overline{\pi}_1^{(t+1)}(a_1)&\propto\pi_1^{(t)}(a_1)^{1-\eta_t\lambda}\exp(\eta_t[A\pi_2^{(t)}]_{a_1})\\
        \overline{\pi}_2^{(t+1)}(a_2)&\propto\pi_2^{(t)}(a_2)^{1-\eta_t\lambda}\exp(-\eta_t[A^\intercal\pi_1^{(t)}]_{a_2})
    \end{align*}
    
    \STATE Update $\pi_1$ and $\pi_2$ for every action $a=(a_1,a_2)\in\mathcal{A}=\mathcal{A}_1\times\mathcal{A}_2$ according to
    \begin{align*}
        \pi_1^{(t+1)}(a_1)&\propto\pi_1^{(t)}(a_1)^{1-\eta_t\lambda}\exp(\eta_t[A\overline{\pi}_2^{(t)}]_{a_1})\\
        \pi_2^{(t+1)}(a_2)&\propto\pi_2^{(t)}(a_2)^{1-\eta_t\lambda}\exp(-\eta_t[A^\intercal\overline{\pi}_1^{(t)}]_{a_2})
    \end{align*}
\ENDFOR

\end{algorithmic}
\end{algorithm}
\end{minipage}
\hfill
\begin{minipage}{.47\textwidth}
\centering
\begin{algorithm}[H]
\caption{The OMWU method}
\label{alg_OMWU}
\begin{algorithmic}
\STATE {\bfseries Initialization: } $\pi_1^{(0)}=\overline{\pi}_1^{(0)},\pi_2^{(0)}=\overline{\pi}_2^{(0)}$.
\STATE {\bfseries Parameters: } Learning rate $\eta_t$

\FOR{t=0,1,2,...}

    \STATE Update $\overline{\pi}_1$ and $\overline{\pi}_2$ for every action $a=(a_1,a_2)\in\mathcal{A}=\mathcal{A}_1\times\mathcal{A}_2$ according to
    \begin{align*}
        \overline{\pi}_1^{(t+1)}(a_1)&\propto\pi_1^{(t)}(a_1)^{1-\eta_t\lambda}\exp(\eta_t[A\pi_2^{(t)}]_{a_1})\\
        \overline{\pi}_2^{(t+1)}(a_2)&\propto\pi_2^{(t)}(a_2)^{1-\eta_t\lambda}\exp(-\eta_t[A^\intercal\pi_1^{(t)}]_{a_2})
    \end{align*}
    
    \STATE Update $\pi_1$ and $\pi_2$ for every action $a=(a_1,a_2)\in\mathcal{A}=\mathcal{A}_1\times\mathcal{A}_2$ according to
    \begin{align*}
        \pi_1^{(t+1)}(a_1)&\propto\pi_1^{(t)}(a_1)^{1-\eta_t\lambda}\exp(\eta_t[A\overline{\pi}_2^{(t)}]_{a_1})\\
        \pi_2^{(t+1)}(a_2)&\propto\pi_2^{(t)}(a_2)^{1-\eta_t\lambda}\exp(-\eta_t[A^\intercal\overline{\pi}_1^{(t)}]_{a_2})
    \end{align*}
\ENDFOR

\end{algorithmic}
\end{algorithm}
\end{minipage}\\

The Policy Extragradient Method (PEM) is based on the PU and OMWU methods. We provide the pseudocode for PEM (Algorithm \ref{alg_PEM}) below for your reference.

\begin{algorithm}[H]
\caption{PEM: Policy Extragradient Method}
\label{alg_PEM}
\begin{algorithmic}
\STATE {\bfseries Initialization: } $Q^{(0)}=0$.
\FOR{t=0,1,2,...,$T_{\text{main}}$}
    \STATE Let $Q^{(t)}$ denote
    \begin{align*}
        Q^{(t)}(s,a_1,a_2) = r^1(s, a_1, a_2) + \lambda \mathbb{E}_{s'\sim P(\cdot|s,a_1,a_2)}V^{(t)}(s')
    \end{align*}
    \STATE Invode PU (Algorithm \ref{alg_PU}) or OMWU (Algorithm \ref{alg_OMWU}) for Tsub iterations to solve the following entropy-regularized matrix game for every state $s$, where initial state is set as uniform distribution
    \begin{align*}
        \max_{\pi_1(s)\in\Delta(\mathcal{A}_1)}\min_{\pi_2(s)\in\Delta(\mathcal{A}_2)}f_{\lambda}(Q^{(t)}(s);\pi_1(s),\pi_2(s)).
    \end{align*}
    Return the last iterate $\overline{\pi}_1^{(t,T_{sub})}(s)$, $\overline{\pi}_2^{(t,T_{sub})}(s)$.
    \STATE Set $V^{(t+1)}(s)=f_{\lambda}(Q^{(t)}(s);\overline{\pi}_1^{(t,T_{sub})}(s),\overline{\pi}_2^{(t,T_{sub})}(s))$.
\ENDFOR

\end{algorithmic}
\end{algorithm}

\subsection{Entropy-regularized OMWU}
In this section, we provide the pseudocode of Entropy-regularized OMWU method \cite{cen2022faster}.

\begin{algorithm}[H]
\caption{Entropy-regularized OMWU}
\label{alg_EROMWU}
\begin{algorithmic}
\STATE {\bfseries Input: } Regularization parameter $\lambda>0$, learning rate for policy update $\eta>0$, learning rate for value update $\{\alpha_t\}_{t=0}^{\infty}$.
\STATE {\bfseries Initialization: } Set $\pi_1^{0}$, $\overline{\pi}_1^{0}$, $\pi_2^{0}$ and $\overline{\pi}_2^{0}$ as uniform policies; and set
$$Q^{(0)}=0,\quad V^{(0)}=\lambda(\log|\mathcal{A}_1|+\log|\mathcal{A}_2|).$$

\FOR{t=0,1,2,...}
    \FOR{all $s\in\mathcal{S}$}
        \STATE When $t\geq 1$, update policy pair $\pi=(\pi_1,\pi_2)$ as:
        \begin{align*}
            \pi_1^{(t+1)}(a_1|s)&\propto\pi_1^{(t)}(a_1|s)^{1-\eta_t\lambda}\exp(\eta[Q^{(t)}(s)\overline{\pi}_2^{(t)}(s)]_{a_1})\\
            \pi_2^{(t+1)}(a_2|s)&\propto\pi_2^{(t)}(a_2|s)^{1-\eta_t\lambda}\exp(-\eta[Q^{(t)}(s)\overline{\pi}_1^{(t)}(s)]_{a_2})
        \end{align*}
        
        \STATE Update policy pair $\overline{\pi}=(\overline{\pi}_1,\overline{\pi}_2)$ as:
        \begin{align*}
            \overline{\pi}_1^{(t+1)}(a_1|s)&\propto\pi_1^{(t)}(a_1|s)^{1-\eta_t\lambda}\exp(\eta[Q^{(t)}(s)\overline{\pi}_2^{(t)}(s)]_{a_1})\\
            \overline{\pi}_2^{(t+1)}(a_2|s)&\propto\pi_2^{(t)}(a_2|s)^{1-\eta_t\lambda}\exp(-\eta[Q^{(t)}(s)\overline{\pi}_1^{(t)}(s)]_{a_2})
        \end{align*}
        
        \STATE Update $Q^{(t+1)}(s)$ and $V^{(t+1)}(s)$ as
        \begin{align*}
            Q^{(t+1)}(s,a_1,a_2)&=r^1(s,a_1,a_2)+\lambda\mathbb{E}_{s'\sim P(\cdot|s,a_1,a_2)}[V^{(t)}(s')]\\
            V^{(t+1)}(s)&=(1-\alpha_{t+1})V^{(t)}(s)+\alpha_{t+1}[\overline{\pi}_1^{(t+1)\intercal}
            Q^{(t+1)}(s)\overline{\pi}_2^{(t+1)}+\lambda\mathcal{H}(\overline{\pi}_1^{(t+1)})-\lambda\mathcal{H}(\overline{\pi}_2^{(t+1)})]
        \end{align*}
    \ENDFOR
\ENDFOR

\end{algorithmic}
\end{algorithm}





\section{Clarification for Assumption \ref{asm:f_1}}\label{sup:asmf1}
Based on the definition $f_*(\theta)=f(\theta, \phi_*(\theta)),$
$f_*(\theta)$ is a composite function consists of $f(\theta, \phi):\Theta\times\Phi\rightarrow\mathbb{R}$ and $\phi_*(\theta):\Theta\rightarrow\Phi$, where $\Theta$ and $\Phi$ are the feasible areas of incentive parameter $\theta$ and policy parameter $\phi$, respectively. Therefore, directly making assumptions for $f_*(\theta)$ might be aggressive and hard to verified. In this section we give an example that if the loss function $f(\theta,\phi)$ and value function $V^{(i)}_{\pi_\phi}(s;\theta)$ defined in equation (\ref{eq:value2}) satisfy proper conditions, the composite function $f_*(\theta)$ satisfy Assumption \ref{asm:f_1} when the loss function $f(\theta,\phi)$ and the value functions $V^{(i)}_{\pi_\phi}(s;
\theta)$ satisfies proper conditions.

Assumption \ref{asm:f_1} consists of 2 parts. Firstly, we require the composite function $f_*(\theta)$ to be bounded. This assumption is mild and easy to satisfy if loss $f(\theta, \phi)$ is bounded. Secondly, we require $f_*$ to be $L$-smooth ($\nabla_\theta f_*(\theta)$ is $L$-lipschitz continuous). This assumption need to be verified. The following Proposition \ref{prop:asm} provide an exact example that $f_*(\theta)$ is both bounded and $L$-smooth.

\begin{assumption}\label{asm:f_ori}
    Assume that loss function $f(\theta,\phi):\Theta\times\Phi\rightarrow\mathbb{R}$ satisfies the following conditions:
    \begin{enumerate}
        \item \label{it:1} \textbf{(Bounded)} $\exists M_f>0$, such that $|f(\theta,\phi)|\leq M_f$, $\forall (\theta, \phi)\in\Theta\times\Phi\subset\mathbb{R}^3$. 
        \item \label{it:2} \textbf{(Lip-continuous)} $\exists G_f>0$, such that $\norm{\nabla f(\theta,\phi)}_2\leq G_f$. $\forall (\theta, \phi)\in\Theta\times\Phi\subset\mathbb{R}^3$.
        \item \label{it:3} \textbf{(Smooth)} $\exists L_f>0$, such that for $\forall (\theta_1,\phi_1), (\theta_2,\phi_2)\in\Theta\times\Phi\subset\mathbb{R}^3$
        $$\norm{\nabla f(\theta_1,\phi_1)-\nabla f(\theta_2,\phi_2)}_2\leq L_f\norm{(\theta_1,\phi_1)-(\theta_2,\phi_2)}_2.$$
    \end{enumerate}
\end{assumption}

\begin{assumption}\label{asm:v}
    Assume that the feasible area for $\theta$, $\phi_1$, and $\phi_2$ are all one-dimensional, which means $\theta\in\Theta\subset\mathbb{R}$, $\phi_1\in\Phi_1\subset\mathbb{R}$, $\phi_2\in\Phi_2\subset\mathbb{R}$ and $\phi=(\phi_1,\phi_2)\in\Phi=\Phi_1\times\Phi_1\subset\mathbb{R}^2$. Assume that for arbitrary fixed state $s\in\mathcal{S}$, the value functions $V^{(i)}_{\pi_\phi}(s;\theta):\Theta\time\Phi\rightarrow\mathbb{R}$, $(i=1,2)$ satisfy the following conditions:
    \begin{enumerate}\setcounter{enumi}{3}
        \item\label{it:4} \textbf{(Second-order gradient bounded)} $\exists G_v>0$, such that for $\forall (\theta, \phi)\in\Theta\times\Phi\subset\mathbb{R}^3$ and $i=1,2$,
        $$\left|\frac{\partial^2V^{(i)}_{\pi_\phi}(s;\theta)}{\partial\theta\partial\phi_i}\right|\leq G_v,\quad \left|\frac{\partial^2V^{(i)}_{\pi_\phi}(s;\theta)}{\partial\phi_i^2}\right|\leq G_v\quad \text{and}\quad \left|\frac{\partial^2V^{(i)}_{\pi_\phi}(s;\theta)}{\partial\phi_i\partial\phi_{-i}}\right|\leq G_v.$$
        
        \item\label{it:5} \textbf{(Third-order gradient bounded)} $\exists L_v>0$, such that for $\forall (\theta, \phi)\in\Theta\times\Phi\subset\mathbb{R}^3$ and $i=1,2$,
        $$\left|\frac{\partial^3V^{(i)}_{\pi_\phi}(s;\theta)}{\partial\theta^2\partial\phi_i}\right|\leq L_v,\quad \left|\frac{\partial^3V^{(i)}_{\pi_\phi}(s;\theta)}{\partial\theta\partial\phi_i^2}\right|\leq L_v\quad \text{and}\quad \left|\frac{\partial^3V^{(i)}_{\pi_\phi}(s;\theta)}{\partial\theta\partial\phi_i\partial\phi_{-i}}\right|\leq L_v.$$
        
        \item\label{it:6} \textbf{(Positive definite matrix)} $\exists \rho_v>0$, such that for $\forall (\theta, \phi)\in\Theta\times\Phi\subset\mathbb{R}^3$,
        $$\left[\begin{matrix}
          \frac{\partial^2V^{(1)}_{\pi_\phi}(s;\theta)}{\partial\phi_1^2} & \frac{\partial^2V^{(1)}_{\pi_\phi}(s;\theta)}{\partial\phi_1\partial\phi_2} \\
          \frac{\partial^2V^{(2)}_{\pi_\phi}(s;\theta)}{\partial\phi_2\phi_1} & \frac{\partial^2V^{(2)}_{\pi_\phi}(s;\theta)}{\partial\phi_2^2} \\
        \end{matrix}\right] \succ \rho_v I,$$
        where $I$ is the identity matrix.
    \end{enumerate} 
\end{assumption}

\begin{lemma}\label{lma:asm}
  Suppose that Assumption \ref{asm:v} holds. $\phi_*(\theta):\Theta\rightarrow\Phi$ is the mapping from incentive parameter theta $\theta$ to policy parameter of Nash policy for regularized MG $\mathcal{G}_{\theta}'$ defined in section \ref{subsec:bi-level}. Then we have
  \renewcommand{\theenumi}{\roman{enumi}}%
  \begin{enumerate}
      \item\label{it:i} For $\forall \theta\in\Theta$, we have $\norm{\nabla_\theta\phi_*(\theta)}_F\leq\frac{2G_v}{\rho_v}$.
      \item\label{it:ii} For $\forall \theta_1, \theta_2\in\Theta$, we have $\norm{\nabla_\theta\phi_*(\theta_1)-\nabla_\theta\phi_*(\theta_2)}_F\leq \frac{2L_v}{\rho_v}\left(1+4G_v+\frac{16G_v^2}{\rho_v}\right)\norm{\theta_1-\theta_2}_2$,
  \end{enumerate}
  where $\norm{\cdot}_F$ is the Frobenius Norm for matrix.
\end{lemma}
\begin{proof}
    See Appendix \ref{subsec:proof_lma_asm}.
\end{proof}

\begin{proposition}\label{prop:asm}
    If Assumption \ref{asm:f_ori} and \ref{asm:v} hold, we could conclude that $f_*(\theta)$ defined in optimization problem (\ref{bilevel:prob2}) satisfies the following conditions:
    \begin{itemize}
        \item[(i)] \textbf{(Bounded)} $|f_*(\theta)|\leq M_f$ for $\forall \theta\in\Theta$.
        
        \item[(ii)] \textbf{(Smooth)} $f_*(\theta)$ is $L$-smooth ($\nabla_\theta f_*(\theta)$ is $L$-Lipschitz continuous), where $$L=L_f(1+\frac{2G_v}{\rho_v})(1+\frac{4G_v}{\rho_v})+\frac{2L_vG_f}{\rho_v}(1+4G_v+\frac{16G_v^2}{\rho_v}).$$
    \end{itemize}
\end{proposition}
\begin{proof}
    Firstly, $|f_*(\theta)|\leq M_f$ is easy to generated from condition (\ref{it:1}) since $\phi^*(\theta)\in\Phi$. We are going to focus on showing $f_*(\theta)$ is $L$-smooth in the following proof.
    
    For $\forall \theta_1,\theta_2\in\Theta$,
    \begin{align}\label{eq:prop1}
        \norm{\nabla_\theta f_*(\theta_1)-\nabla_\theta f_*(\theta_2)}_2
        =&\norm{\nabla_\theta f(\theta_1,\phi_*(\theta_1))+\nabla_\theta\phi_*(\theta_1)^\intercal\nabla_\phi f(\theta_1,\phi_*(\theta_1))-\nabla_\theta f(\theta_2,\phi_*(\theta_2))-\nabla_\theta\phi_*(\theta_2)^\intercal\nabla_\phi f(\theta_2,\phi_*(\theta_2))}_2\notag\\
        \leq&\norm{\nabla_\theta f(\theta_1,\phi_*(\theta_1))-\nabla_\theta f(\theta_2,\phi_*(\theta_2))}_2 + \norm{\nabla_\theta\phi_*(\theta_1)^\intercal\nabla_\phi f(\theta_1,\phi_*(\theta_1))-\nabla_\theta\phi_*(\theta_2)^\intercal\nabla_\phi f(\theta_2,\phi_*(\theta_2))}_2\notag\\
        \leq&\norm{\nabla_\theta f(\theta_1,\phi_*(\theta_1))-\nabla_\theta f(\theta_2,\phi_*(\theta_2))}_2 +\norm{\nabla_\theta\phi_*(\theta_1)}_F\norm{\nabla_\phi f(\theta_1,\phi_*(\theta_1))-\nabla_\phi f(\theta_3,\phi_*(\theta_2))}_2\notag\\
        &+\norm{\nabla_\phi f(\theta_2,\phi_*(\theta_2))}_2\norm{\nabla_\theta\phi_*(\theta_1)-\nabla_\theta\phi_*(\theta_2)}_F
    \end{align}
    We are going to bound three terms in formula (\ref{eq:prop1}) one by one.
    \begin{itemize}
        \item[i)] By the condition (\ref{it:3}) in Assumption \ref{asm:f_ori} and property (\ref{it:i}) in Lemma \ref{lma:asm}, we have
        \begin{align}\label{eq:prop2}
            \norm{\nabla_\theta f(\theta_1,\phi_*(\theta_1))-\nabla_\theta f(\theta_2,\phi_*(\theta_2))}_2
            &\leq \norm{\nabla_\theta f(\theta_1,\phi_*(\theta_1))-\nabla_\theta f(\theta_2,\phi_*(\theta_2))}_2\notag\\
            &\leq L_f\norm{(\theta_1,\phi_*(\theta_1))-(\theta_2,\phi_*(\theta_2))}_2\quad &\text{(Condition (\ref{it:3}) in Assumption \ref{asm:f_ori})}\notag\\
            &\leq L_f(\norm{\theta_1-\theta_2}_2+\norm{\phi_*(\theta_1)-\phi_*(\theta_2)}_2)\notag\\
            &\leq L_f(\norm{\theta_1-\theta_2}_2+\max_{\theta\in\Theta}\norm{\nabla_\theta\phi_*(\theta)}_F\norm{\theta_1-\theta_2}_2) \quad &\text{(Mean value Theorem)}\notag\\
            &\leq L_f(1+\frac{2G_v}{\rho_v})\norm{\theta_1-\theta_2}_2 \quad &\text{(Property (\ref{it:i}) in Lemma \ref{lma:asm})}
        \end{align}
    
    \item[ii)]By property (\ref{it:i}) in Lemma \ref{lma:asm} and the condition (\ref{it:3}) in Assumption \ref{asm:f_ori}, we have
    \begin{align}\label{eq:prop3}
        &\norm{\nabla_\theta\phi_*(\theta_1)}_F\norm{\nabla_\phi f(\theta_1,\phi_*(\theta_1))-\nabla_\phi f(\theta_3,\phi_*(\theta_2))}_2\notag\\
        \leq& \frac{2G_v}{\rho_v}\norm{\nabla_\phi f(\theta_1,\phi_*(\theta_1))-\nabla_\phi f(\theta_3,\phi_*(\theta_2))}_2&\text{(Property (\ref{it:i}) in Lemma \ref{lma:asm})}\notag\\
        \leq& \frac{2G_vL_f}{\rho_v}\norm{(\theta_1,\phi_*(\theta_1))-(\theta_2,\phi_*(\theta_2))}_2 &\text{(Condition (\ref{it:3}) in Assumption \ref{asm:f_ori})}\notag\\
        \leq& \frac{2G_vL_f}{\rho_v}(1+\frac{2G_v}{\rho_v})\norm{\theta_1-\theta_2}_2
    \end{align}
    
    \item[iii)]By condition (\ref{it:2}) in Assumption \ref{asm:f_ori} and property (\ref{it:ii}) in Lemma \ref{lma:asm}, we have
    \begin{align}\label{eq:prop4}
        &\norm{\nabla_\phi f(\theta_2,\phi_*(\theta_2))}_2\norm{\nabla_\theta\phi_*(\theta_1)-\nabla_\theta\phi_*(\theta_2)}_F\notag\\
        \leq& G_f\norm{\nabla_\theta\phi_*(\theta_1)-\nabla_\theta\phi_*(\theta_2)}_F&\text{(Condition (\ref{it:2}) in Assumption \ref{asm:f_ori})}\notag\\
        \leq&\frac{2L_vG_f}{\rho_v}\left(1+4G_v+\frac{16G_v^2}{\rho_v}\right)\norm{\theta_1-\theta_2}_2 &\text{(Property (\ref{it:ii}) in Lemma \ref{lma:asm})}
    \end{align}
    \end{itemize}
    Combine inequalities (\ref{eq:prop1}-\ref{eq:prop4}), we could get
    \begin{align}
        \norm{\nabla_\theta f_*(\theta_1)-\nabla_\theta f_*(\theta_2)}_2\leq L_f(1+\frac{2G_v}{\rho_v})(1+\frac{4G_v}{\rho_v})+\frac{2L_vG_f}{\rho_v}(1+4G_v+\frac{16G_v^2}{\rho_v})\norm{\theta_1-\theta_2}_2
    \end{align}
    We already prove that $f_*(\theta)$ is $L$-smooth with $L=L_f(1+\frac{2G_v}{\rho_v})(1+\frac{4G_v}{\rho_v})+\frac{2L_vG_f}{\rho_v}(1+4G_v+\frac{16G_v^2}{\rho_v})$.
\end{proof}
\subsection{Proof for Lemma \ref{lma:asm}}\label{subsec:proof_lma_asm}
\begin{proof}[Proof for Lemma \ref{lma:asm}.]
When conditions (\ref{it:4}-\ref{it:6}) in Assumption \ref{asm:v} hold for $V^{(i)}_{\pi_\phi}(s;\theta)$ on every state $s\in\mathcal{S}$, it is obvioud that conditions (\ref{it:4}-\ref{it:6}) also hold for $\mathbb{E}_{v^{(*)}}V^{(i)}_{\pi_\phi}(s;\theta)$. For notation simplicity, we denote $\mathbb{E}_{v^{(*)}}V^{(i)}_{\pi_\phi}(s;\theta):=V^i(\theta, \phi)$ in this section. We are going to prove 2 properties one by one.
    \begin{itemize}
        \item[i)] Firstly, we are going to bound $\norm{\nabla_\theta\phi_*(\theta)}_F$. For arbitrary $\theta\in\Theta$ and $\phi\in\Phi$, since $\phi_1\in\mathbb{R}$, $\phi_2\in\mathbb{R}$ and $\phi=(\phi_1,\phi_2)\in\mathbb{R}^2$, formula (\ref{notation:gradient u}-\ref{notation:gradient u2}) in section \ref{sec:algorithm} could be written as
        \begin{align}\label{eq:lma_asm_1}
            \nabla_{\theta}u_\theta(\phi)=\left[\begin{matrix} \frac{\partial^2V^1(\theta,\phi)}{\partial\theta\phi_1}\\ \frac{\partial^2V^2(\theta,\phi)}{\partial\theta\phi_2}
            \end{matrix}\right]\in\mathbb{R}^{2\times 1} , \quad\quad\quad
            \nabla_\phi u_\theta (\phi)=\left[\begin{matrix} \frac{\partial^2V^1(\theta,\phi)}{\partial\phi_1^2} & \frac{\partial^2V^1(\theta,\phi)}{\partial\phi_1\partial\phi_2}\\
            \frac{\partial^2V^2(\theta,\phi)}{\partial\phi_1\partial\phi_2} & \frac{\partial^2V^1(\theta,\phi)}{\partial\phi_2^2}
            \end{matrix}\right]\in\mathbb{R}^{2\times 2}
        \end{align}
        By formula (\ref{eq:grad_phi_new}) in Lemma \ref{lma:phitheta} and condition (\ref{it:4}, \ref{it:6}) in Assumption \ref{asm:v}, we have
        \begin{align}
            \norm{\nabla_\theta\phi^*(\theta)}_F\leq&\norm{\left[\begin{matrix} \frac{\partial^2V^1(\theta,\phi)}{\partial\phi_1^2} & \frac{\partial^2V^1(\theta,\phi)}{\partial\phi_1\partial\phi_2}\\
            \frac{\partial^2V^2(\theta,\phi)}{\partial\phi_1\partial\phi_2} & \frac{\partial^2V^1(\theta,\phi)}{\partial\phi_2^2}
            \end{matrix}\right]^{-1}\left[\begin{matrix} \frac{\partial^2V^1(\theta,\phi)}{\partial\theta\phi_1}\\ \frac{\partial^2V^2(\theta,\phi)}{\partial\theta\phi_2}
            \end{matrix}\right]}_F\notag\\
            \leq&\norm{\left[\begin{matrix} \frac{\partial^2V^1(\theta,\phi)}{\partial\phi_1^2} & \frac{\partial^2V^1(\theta,\phi)}{\partial\phi_1\partial\phi_2}\\
            \frac{\partial^2V^2(\theta,\phi)}{\partial\phi_1\partial\phi_2} & \frac{\partial^2V^1(\theta,\phi)}{\partial\phi_2^2}
            \end{matrix}\right]}_F^{-1}\norm{\left[\begin{matrix} \frac{\partial^2V^1(\theta,\phi)}{\partial\theta\phi_1}\\ \frac{\partial^2V^2(\theta,\phi)}{\partial\theta\phi_2}
            \end{matrix}\right]}_F
            &\text{(Condition (\ref{it:6}) in Assumption \ref{asm:v})}\notag\\
            \leq&\frac{1}{\rho_v}\left(\left|\frac{\partial^2V^1(\theta,\phi)}{\partial\theta\phi_1}\right|+\left|\frac{\partial^2V^2(\theta,\phi)}{\partial\theta\phi_2}\right|\right)\notag\\
            \leq&\frac{2G_v}{\rho_v} &\text{(Condition (\ref{it:4}) in Assumption \ref{asm:v})}
        \end{align}
        
        \item[ii)] Secondly, we are going to bound $\norm{\nabla_\theta\phi_*(\theta_1)-\nabla_\theta\phi_*(\theta_2)}_F$. For $\forall \theta_1,\theta_2\in\Theta$,
        \begin{align}
            \norm{\nabla_\theta\phi_*(\theta_1)-\nabla_\theta\phi_*(\theta_2)}_F=&\norm{\left[\nabla_\phi u_{\theta_1}(\phi_*(\theta_1))\right]^{-1}\nabla_\theta u_{\theta_1}(\phi_*(\theta_1))-\left[\nabla_\phi u_{\theta_2}(\phi_*(\theta_2))\right]^{-1}\nabla_\theta u_{\theta_2}(\phi_*(\theta_2))}_F\notag\\
            \leq&\norm{\left[\nabla_\phi u_{\theta_1}(\phi_*(\theta_1))\right]^{-1}}_F\norm{\nabla_\theta u_{\theta_1}(\phi_*(\theta_1))-\nabla_\theta u_{\theta_2}(\phi_*(\theta_2))}_F\notag\\
            &+\norm{\nabla_\theta u_{\theta_2}(\phi_*(\theta_2))}_F\norm{\left[\nabla_\phi u_{\theta_1}(\phi_*(\theta_1))\right]^{-1} - \left[\nabla_\phi u_{\theta_2}(\phi_*(\theta_2))\right]^{-1}}_F
        \end{align}
        By Condition (\ref{it:6}) we have
        \begin{align}\label{eq:lma_f1}
            \norm{\left[\nabla_\phi u_{\theta_1}(\phi_*(\theta_1))\right]^{-1}}_F\leq\frac{1}{\rho}.
        \end{align}
        By Condition (\ref{it:4}) we have
        \begin{align}\label{eq:lma_f2}
            \norm{\nabla_\theta u_{\theta_2}(\phi_*(\theta_2))}_F\leq \left|\frac{\partial^2V^1(\theta_2,\phi_*(\theta_2))}{\partial\theta\phi_1}\right|+\left|\frac{\partial^2V^2(\theta_2,\phi_*(\theta_2))}{\partial\theta\phi_2}\right|\leq 2G_v.
        \end{align}
        By Condition (\ref{it:5}) we have
        \begin{align}\label{eq:lma_f3}
            \norm{\nabla_\theta u_{\theta_1}(\phi_*(\theta_1))-\nabla_\theta u_{\theta_2}(\phi_*(\theta_2))}_F\leq&\left|\frac{\partial^2V^1(\theta_1,\phi_*(\theta_1))}{\partial\theta\phi_1}-\frac{\partial^2V^1(\theta_2,\phi_*(\theta_2))}{\partial\theta\phi_1}\right|+\left|\frac{\partial^2V^2(\theta_1,\phi_*(\theta_1))}{\partial\theta\phi_2}-\frac{\partial^2V^2(\theta_2,\phi_*(\theta_2))}{\partial\theta\phi_2}\right|\notag\\
            \leq&\max_{\theta,\phi}\left|\frac{\partial^3V^1(\theta_1,\phi_*(\theta_1))}{\partial\theta^2\phi_1}\right|\norm{\theta_1-\theta_2}_2+\max_{\theta,\phi}\left|\frac{\partial^3V^2(\theta_1,\phi_*(\theta_1))}{\partial\theta^2\phi_2}\right|\norm{\theta_1-\theta_2}_2\notag\\
            \leq& 2L_v\norm{\theta_1-\theta_2}.
        \end{align}
        Last but not least,
        \begin{align}
            &\norm{\left[\nabla_\phi u_{\theta_1}(\phi_*(\theta_1))\right]^{-1} - \left[\nabla_\phi u_{\theta_2}(\phi_*(\theta_2))\right]^{-1}}_F \\
            =& \norm{\left[\begin{matrix} \frac{\partial^2V^1(\theta_1,\phi_*(\theta_1))}{\partial\phi_1^2} & \frac{\partial^2V^1(\theta_1,\phi_*(\theta_1))}{\partial\phi_1\partial\phi_2} \\
            \frac{\partial^2V^2(\theta_1,\phi_*(\theta_1))}{\partial\phi_1\partial\phi_2} & \frac{\partial^2V^2(\theta_1,\phi_*(\theta_1))}{\partial\partial\phi_2^2} \\
            \end{matrix}\right]^{-1}
            -\left[\begin{matrix} \frac{\partial^2V^1(\theta_2,\phi_*(\theta_2))}{\partial\phi_1^2} & \frac{\partial^2V^1(\theta_2,\phi_*(\theta_2))}{\partial\phi_1\partial\phi_2} \\
            \frac{\partial^2V^2(\theta_2,\phi_*(\theta_2))}{\partial\phi_1\partial\phi_2} & \frac{\partial^2V^2(\theta_2,\phi_*(\theta_2))}{\partial\partial\phi_2^2} \\
            \end{matrix}\right]^{-1}
            }_F
        \end{align}
        For notation simplicity we denote
        \begin{align*}
            &a_1=\frac{\partial^2V^1(\theta_1,\phi_*(\theta_1))}{\partial\phi_1^2}, b_1 = \frac{\partial^2V^1(\theta_1,\phi_*(\theta_1))}{\partial\phi_1\partial\phi_2}, c_1 = \frac{\partial^2V^2(\theta_1,\phi_*(\theta_1))}{\partial\phi_1\partial\phi_2}, d_1 = \frac{\partial^2V^2(\theta_1,\phi_*(\theta_1))}{\partial\partial\phi_2^2}\\
            &a_2=\frac{\partial^2V^1(\theta_2,\phi_*(\theta_2))}{\partial\phi_1^2}, b_2 = \frac{\partial^2V^1(\theta_2,\phi_*(\theta_2))}{\partial\phi_1\partial\phi_2}, c_2 = \frac{\partial^2V^2(\theta_2,\phi_*(\theta_2))}{\partial\phi_1\partial\phi_2}, d_2 = \frac{\partial^2V^2(\theta_2,\phi_*(\theta_2))}{\partial\partial\phi_2^2}\\
            &M_1 = det\left(\left[\begin{matrix}
              a_1 & b_1 \\ c_1 & d_1
            \end{matrix}\right]\right), 
            M_2 = det\left(\left[\begin{matrix}
              a_2 & b_2 \\ c_2 & d_2
            \end{matrix}\right]\right)
        \end{align*}
        Therefore, we have
        \begin{align}\label{ieq:M1}
            &\norm{\left[\nabla_\phi u_{\theta_1}(\phi_*(\theta_1))\right]^{-1} - \left[\nabla_\phi u_{\theta_2}(\phi_*(\theta_2))\right]^{-1}}_F\notag\\
            =&\norm{\frac{1}{M_1}\left[\begin{matrix}
              d_1 & -b_1 \\ -c_1 & a_1
            \end{matrix}\right]-\frac{1}{M_2}\left[\begin{matrix}
              d_2 & -b_2 \\ -c_2 & a_2
            \end{matrix}\right]}_F\notag\\
            \leq&\left|\frac{a_1}{M_1}-\frac{a_2}{M_2}\right|+\left|\frac{b_1}{M_1}-\frac{b_2}{M_2}\right|+\left|\frac{c_1}{M_1}-\frac{c_2}{M_2}\right|+\left|\frac{d_1}{M_1}-\frac{d_2}{M_2}\right|
        \end{align}
        We are first going to bound $\left|\frac{a_1}{M_1}-\frac{a_2}{M_2}\right|$.
        \begin{align*}
            \left|\frac{a_1}{M_1}-\frac{a_2}{M_2}\right|\leq\frac{1}{|M_1|}\left|a_1-a_1\right|+\frac{|a_2|}{|M_1M_2|}|M_1-M_2|.
        \end{align*}
        By Condition (\ref{it:6}), we have
        \begin{align}\label{eq:M_1}
            |M_1|\geq\rho_v\quad \text{and}\quad|M_2|\geq\rho_v.
        \end{align}
        By Condition (\ref{it:5}) we have
        \begin{align}\label{eq:M_2}
            \left|a_1-a_1\right|\leq L_v\norm{\theta_1-\theta_2}_2.
        \end{align}
        By Condition (\ref{it:4}) we have
        \begin{align}\label{eq:M_3}
            |a_2|\leq G_v.
        \end{align}
        Similarly, by Condition (\ref{it:4},\ref{it:5}) we have
        \begin{align}\label{eq:M_4}
            |M_1-M_2|=&|a_1d_1-b_1c_1-a_2d_2+b_2c_2|\notag\\
            \leq& |a_1||d_1-d_2|+|d_2||a_1-a_2|+|b_1||c_1-c_2|+|c_2||b_1-b_2|\notag\\
            \leq& 4G_vL_v\norm{\theta_1-\theta_2}_2.
        \end{align}
        Combine formula (\ref{eq:M_1}-\ref{eq:M_4}) we have
        \begin{align}
            \left|\frac{a_1}{M_1}-\frac{a_2}{M_2}\right|\leq(\frac{L_v}{\rho_v}+\frac{G_v}{\rho_v^2}\times4G_vL_v)\norm{\theta_1-\theta_2}_2=\frac{L_v}{\rho_v}(1+\frac{4G_v^2}{\rho_v})\norm{\theta_1-\theta_2}_2.
        \end{align}
        Similarly, we could compute the bound for $\left|\frac{b_1}{M_1}-\frac{b_2}{M_2}\right|$, $ \left|\frac{c_1}{M_1}-\frac{c_2}{M_2}\right|$ and $\left|\frac{d_1}{M_1}-\frac{d_2}{M_2}\right|$. Therefore, we have
        \begin{align}\label{eq:lma_f4}
            \norm{\left[\nabla_\phi u_{\theta_1}(\phi_*(\theta_1))\right]^{-1} - \left[\nabla_\phi u_{\theta_2}(\phi_*(\theta_2))\right]^{-1}}_F\leq\frac{4L_v}{\rho_v}(1+\frac{4G_v^2}{\rho_v})\norm{\theta_1-\theta_2}_2
        \end{align}
        Combine formula (\ref{eq:lma_f1},\ref{eq:lma_f2},\ref{eq:lma_f3},\ref{eq:lma_f4}) we have
        \begin{align}
            \norm{\nabla_\theta\phi_*(\theta_1)-\nabla_\theta\phi_*(\theta_2)}_F\leq& \frac{1}{\rho_v}\times 2L_v\norm{\theta_1-\theta_2}_2+2G_v\times\frac{4L_v}{\rho_v}(1+\frac{4G_v^2}{\rho_v})\norm{\theta_1-\theta_2}_2\notag\\
            \leq& \frac{2L_v}{\rho_v}(1+4G_v+\frac{16G_v^2}{\rho_v})\norm{\theta_1-\theta_2}_2.
        \end{align}
    \end{itemize}
\end{proof}

\section{Proof for Theorem \ref{thm:outer_conv}}\label{sup:proof_outer}

\begin{proof}[Proof for Theorem \ref{thm:outer_conv}]
    By Assumption \ref{asm:f_1},the gradient $\nabla_\theta f_*(\theta)$ is L-Lipschitz continuous implies for $\forall k$
    \begin{align}\label{eq:L-smooth}
        f_*(\theta_{k+1})\leq f_*(\theta_k)+\nabla_\theta f_*(\theta_k)^\intercal(\theta_{k+1}-\theta_k)+\frac{L}{2}\norm{\theta_{k+1}-\theta_k}^2.
    \end{align}
    Since $\theta_{k+1}=\theta_k-\beta\nabla_{\theta}f_*(\theta_k)$, we have $\theta_{k+1}-\theta_k=-\beta\nabla_{\theta}f_*(\theta_k)$. Substitute $\theta_{k+1}-\theta_k$ in formula \ref{eq:L-smooth}, we have
    \begin{align}\label{eq:f2}
        f_*(\theta_{k+1})-f_*(\theta_k)\leq&-\beta\nabla_{\theta}f_*(\theta_k)^\intercal\nabla_{\theta}f_*(\theta_k)+\frac{L\beta}{2}\norm{\nabla_{\theta}f_*(\theta_k)}^2\notag\\
        =& \beta(\frac{L\beta}{2}-1)\norm{\nabla_{\theta}f_*(\theta_k)}^2.
    \end{align}
    Since the learning rate $\beta=\frac{1}{L}$, substitute $\beta$ in formula \ref{eq:f2} we have
    \begin{align}\label{eq:f3}
        f_*(\theta_{k+1})-f_*(\theta_k)\leq-\frac{1}{2L}\norm{\nabla_{\theta}f_*(\theta_k)}^2.
    \end{align}
    Therefore, summing for $k=0,...,T$ for both LHS and RHS of inequality \ref{eq:f3} we have
    \begin{align*}
        \sum_{k=0}^T\left[f_*(\theta_{k+1})-f_*(\theta_k)\right]&\leq -\frac{1}{2L}\sum_{k=0}^T\norm{\nabla_{\theta}f_*(\theta_k)}^2
        \Longrightarrow \sum_{k=0}^T\norm{\nabla_{\theta}f_*(\theta_k)}^2 \leq 2L(f_*(\theta_0)-f_*(\theta_T))
    \end{align*}
    Since $\theta^* = \argmin_\theta f_*(\theta)$, we have $f_*(\theta_T)\geq f_*(\theta^*)$. Then
    \begin{align}
        \sum_{k=0}^T\norm{\nabla_{\theta}f_*(\theta_k)}^2 \leq 2L(f_*(\theta_0)-f_*(\theta^*))
    \end{align}
    Therefore, since the square of norms are all positive, the minimum of the square of norms must be smaller than the mean of the square of norms. we have
    \begin{align}
        \min_{k=0,...,T}\norm{\nabla_{\theta}f_*(\theta_k)}^2\leq\frac{1}{T+1}\sum_{k=0}^T\norm{\nabla_{\theta}f_*(\theta_k)}^2 \leq \frac{2L(f_*(\theta_0)-f_*(\theta^*))}{T+1}.
    \end{align}
    By Assumption \ref{asm:f_1}, $|f_*(\theta)|\leq M$ for all $\theta\in\Theta$ and $f_*(\theta_0)\geq f_*(\theta^*)$, we could conclude
    $$f_*(\theta_0)-f_*(\theta^*) = |f_*(\theta_0)-f_*(\theta^*)|\leq|f_*(\theta_0)|+|f_*(\theta^*)|\leq 2M.$$
    Therefore,
    \begin{align}
        \min_{k=0,...,T}\norm{\nabla_{\theta}f_*(\theta_k)}^2\leq\frac{1}{T+1}\sum_{k=0}^T\norm{\nabla_{\theta}f_*(\theta_k)}^2 \leq \frac{2L(f_*(\theta_0)-f_*(\theta^*))}{T+1}\leq \frac{4LM}{T+1}.
    \end{align}
\end{proof}

\newpage
\section{Experimental Details}\label{supp:exp}
\subsection{DASAC}
\begin{algorithm}[H]
\caption{DASAC: SAC-based Differentiable Arbitrating in MARL. }
\label{alg_IncenRl_sac}
\begin{algorithmic}
\STATE {\bfseries Input:}Seq. of stepsize $\left(\{\bar{\alpha}^{(i)}_t\}_{t\in[0,T\!-\!1],i\in\{1,2\}},\beta_k\right)$, initial regularization parameters $\lambda^i_0, i\in\left\{1,2\right\}$, 
truncation parameters $\left\{Q_{\max}^{(i)},\mathcal{E}_{\max}^{(i)}\right\}i\in\left\{1,2\right\}$, target entropy for SAC algorithm $\bar{\mathcal{H}}$.
\STATE Initialization actor parameter $\theta_0\in\Theta$, marginalized Q estimator parameter $\psi_0\in\Psi$, and centralized joint Q estimator parameter (namely centralized critic) $\xi_0\in\Xi$.

\FOR{k=0,1,...}
  \STATE$\mathcal{Q}_{\psi_0^i}^{(i)}(s,a^i;\theta_k)\leftarrow 0 \;(i\in\{1,2\})$
  
  \FOR{t=0,1...}
  \STATE Sample transitions $(s_t, u_t^i, u_t^{-i}, r_t^i, s_{t+1})$ from the replay buffer
  \STATE \textcolor{gray}{// {\it Update the actor}}
  \STATE Update policy via minimizing the KL divergence
  \begin{align*}
      \pi_{\phi_t^i}^i(\cdot|s_t;\theta_k) \leftarrow \argmin_{\phi\in\Phi}D_{KL}\left(\pi_{\phi}||exp\left(\frac{1}{\lambda^i_t}Q_{\psi_t^i}(s_t,a_t^i)\right)\right)
  \end{align*}
  \STATE \textcolor{gray}{// {\it Optimize the centralized joint Q estimator}}
  \STATE Compute the value function of the next state with the centralized joint Q function
    \begin{align*}
      V^{cent,(i)}_{t+1}=Q^{cent}_{\xi_t^{(i)}}(s_{t+1},a_t^i,a_t^{-i};\theta_k)-\lambda^i_t\log(\pi_{\phi_t^i}(a_t^i|s_{t+1};\theta_k))+\lambda^{-i}_t\log(\pi_{\phi_t^{-i}}(a_t^{-i}|s_{t+1};\theta_k))
    \end{align*}
  \STATE Update the parameter $\xi^i$ by minimizing $$\left(r_t^i + (1 - {\it is\_done})\gamma V^{cent,(i)}_{t+1} - Q^{cent}_{\xi_t^{(i)}}(s_t,u_t^{(i)},u_t^{(-i)};\theta_k)\right)^2$$
  \STATE \textcolor{gray}{// {\it Optimize the marginalized Q estimator}}
  \STATE Compute the marginalized $Q^{(i)}$-function from the learned centralized joint Q function 
  \begin{align*}
      \widehat{Q}_{\pi_{\phi_t}}^{(i)}(s,a_t^i;\theta_k)=\mathbb{E}_{\hat{a}_t^{-i}\sim\pi^{-i}_{\phi_t}}\left[Q^{cent}_{\xi^{(i)}_t}(s_t,a_t^i,\hat{a}_t^{-i};\theta_k)\right]
  \end{align*}
  \STATE Update the estimated ideal energy function
  \begin{align*}
      \widehat{\mathcal{Q}}_{t+1}^{(i)}(s,a_t^i;\theta_k)=&(1-\bar{\alpha}^{(i)}_t)\cdot \mathcal{Q}_{\phi_t^i}^{(i)}(s,a_t^i;\theta_k)+\bar{\alpha}^{(i)}_t\cdot \widehat{Q}_{\pi_{\phi_t}}^{(i)}(s,a_t^i;\theta_k).
  \end{align*}
  \STATE Update the parameters $\psi$ of marginalized Q estimator to obtain $\mathcal{Q}_{\psi^i}^{(i)}(s,a^i;\theta)\in\mathcal{F}_{\mathcal{Q}_{\max}^{(i)}}$:
  \begin{align*}
    \psi_{t+1}\leftarrow \argmin_{\psi}\mathbb{E}_{\sigma_t}\left[\sum_{i\in\left\{1,2\right\}}\!\!\!\left(\mathcal{Q}_{\psi^i}^{(i)}(s,a_t^i;\!\theta)-\widehat{\mathcal{E}}_{t+1}^{(i)}(s,a_t^i;\!\theta)\right)^2\right]\!.
  \end{align*}
  \STATE \textcolor{gray}{// {\it Auto-tune the regularization parameters}}
  \STATE Update the regularization parameters $\lambda^i$
  \begin{align*}
      \lambda^i_{t+1}\leftarrow\argmin_{\lambda}\mathbb{E}_{a^i_t\sim\pi_{\phi_t}^i}\left[\lambda^i_t\log\pi_{\phi_t}^i(a_t^i|s_t)-\lambda^i_t\bar{\mathcal{H}}\right]
  \end{align*}
  \STATE$\psi^{(*)}(\theta_k)=\psi_{t+1}.$
  \ENDFOR
  \STATE Update incentive parameter
  $$\theta_{k+1}\leftarrow\theta_k-\beta_k\nabla f_{*}(\theta_k),$$
  where $\nabla f_{*}(\theta_k)$ is defined in (\ref{eq:grad_f}).
\ENDFOR
\end{algorithmic}
\end{algorithm}

\subsection{GridSearch(M) and BayesOpt baselines}\label{supp:GSM}

For comparison, we apply the grid search over the incentive parameters as a zeroth-order algorithm. We first perform a round of coarse-grained grid search, using the same modified MASAC algorithm under the incentive parameters from 0.0 to 0.5 with a step size of 0.05, and then for two games respectively, a round of fine-grained search is applied. For Predator-prey, we do grid search from 0.3 to 0.35 with a step size of 0.005 while we search from 0.35 to 0.4 with the same step size for Running-with-scissors.

For the improved zeroth-order baseline, \textit{BayesOpt}, we use the most popular expected improvement (EI) as the acquisition function and use \textit{L-BFGS-B} method to maximize the acquisition function with a random start point in order to compute the next sampling point.

\subsection{Implementation Details}

\subsubsection{Computation Resources}

We conduct our experiments on one NVIDIA GeForce RTX 3090 GPU with 24 GB GDDR6X memory. We implement our codes on the PyTorch framework with CUDA acceleration, which we can run in parallel on one GPU card.

\subsubsection{Hyper-parameters}
\quad
\begin{table}[H]
    \centering
    \begin{tabular}{lc}
        \toprule
        Hyper-parameter &Value\\
        \midrule
        initialization method &orthogonal\\
        num GRU layers &1 \\
        Actor MLP hidden state dim  & 32 \\
        Actor RNN hidden state dim & 16 \\
        Critic MLP\&RNN hidden state dim & 64 \\
        num FC after & 1\\
        optimizer &Adam \\
        optimizer eps &1e-5\\
        weight decay & 0\\
        activation function & ReLU\\
        use reward normalization &true \\
        episode length &25 \\
        last action layer gain &0.01\\
        batch size &64\\
        buffer size &5000 \\
        gamma &0.99\\
        epsilon &from 1.0 to 0.05\\
        epsilon anneal time &50000\\
        Q function loss &MSE loss\\
        initial regularization parameter & 1\\
        target entropy coefficient &0.3\\
        \bottomrule
    \end{tabular}
    \vspace{1mm}
    \caption{Hyper-parameters used in our experiments by DASAC.}
\end{table}


\end{document}